\documentclass[twocolumn]{IEEEtran}

\usepackage{enumitem}
\usepackage[normalem]{ulem}

\usepackage{pifont}
\usepackage{amsmath}
\usepackage{amsthm}
\usepackage{amssymb}

\usepackage[T1]{fontenc}
\usepackage{graphicx}

\usepackage{color}
\usepackage[normalem]{ulem}

\graphicspath{{./figure/}}

\newtheorem{ntn}{Notation}
\newtheorem{dfn}{Definition}

\newtheorem{prop}[dfn]{Proposition}
\newtheorem{rem}[dfn]{Remark}
\newtheorem{thm}[dfn]{Theorem}

\newtheorem{lem}[dfn]{Lemma}
\newtheorem{prob}{Problem}
\newtheorem{example}{Example}

\newcommand{\w}{\omega}
\newcommand{\ov}{\overline}
\newcommand{\E}{\mathbf{E}}

\newcommand{\ComplexField}{\mathbf{C}}

\newcommand{\transposesymbol}{T}
\newcommand{\genericvector}{v}
\newcommand{\genericvectordimension}{n}
\newcommand{\genericvectorindex}{i}
\newcommand{\genericvectorindexalt}{j}
\newcommand{\genericvectorindexstart}{1}
\newcommand{\genericvectorindexend}{\genericvectordimension}
\newcommand{\genericvectorindexsubset}{I}
\newcommand{\genericvectorindexsubsetdimension}{\genericvectordimension_{\genericvectorindexsubset}}
\newcommand{\genericvectorindexsubsetindexstart}{\genericvectorindex_{1}}
\newcommand{\genericvectorindexsubsetindexend}{\genericvectorindex_{\genericvectorindexsubsetdimension}}
\newcommand{\genericvectorindexsubsetalt}{J}
\newcommand{\genericvectorindexsubsetaltdimension}{\genericvectordimension_{\genericvectorindexsubsetalt}}
\newcommand{\genericvectorindexsubsetaltindexstart}{\genericvectorindexalt_{1}}
\newcommand{\genericvectorindexsubsetaltindexend}{\genericvectorindexalt_{\genericvectorindexsubsetaltdimension}}
\newcommand{\genericmatrix}{Q}
\newcommand{\genericmatrixdimensionA}{\genericvectordimension}
\newcommand{\genericmatrixdimensionB}{\genericvectordimension}
\newcommand{\identitymatrix}{\mathcal{I}}

\newcommand*{\vv}[1]{\vec{\mkern0mu#1}}
\newcommand{\graph}{G}
\newcommand{\mutilatedgraph}{\overline{\graph}}
\newcommand{\vertex}{y}
\newcommand{\vertexalt}{w}
\newcommand{\vertexset}{V}
\newcommand{\vertexnumber}{n}
\newcommand{\vertexindex}{i}
\newcommand{\vertexindexalt}{j}
\newcommand{\vertexindexaltb}{k}
\newcommand{\vertexindexsep}{z}
\newcommand{\edgeset}{E}
\newcommand{\directededgeset}{\vv \edgeset}

\newcommand{\path}{\pi}
\newcommand{\pathlength}{\ell}
\newcommand{\pathnodeindex}{p}
\newcommand{\pathindexstart}{1}
\newcommand{\pathindexend}{\pathlength}
\newcommand{\pathindexendminusone}{\pathlength-1}

\newcommand{\parent}[2]{pa_{#1}\left(#2\right)}
\newcommand{\child}[2]{ch_{#1}\left(#2\right)}
\newcommand{\ancestor}[2]{an_{#1}\left(#2\right)}
\newcommand{\descendant}[2]{de_{#1}\left(#2\right)}

\newcommand{\MarkovBlanket}[2]{MB_{#1}\left(#2\right)}
\newcommand{\subsetvertex}{{\vertex}_{I}}
\newcommand{\subsetvertexalt}{{\vertex}_{J}}
\newcommand{\subsetvertexsep}{{{\vertex}_{Z}}}
\newcommand{\subsetvertexprime}{{\vertex}_{I'}}
\newcommand{\subsetvertexaltprime}{{\vertex}_{J'}}

\newcommand{\dsep}[1]{dsep_{#1}}
\newcommand{\vertexsetsmallertocheckdsep}{\node_{U}}

\newcommand{\PSD}{\Phi}
\newcommand{\ZTvar}{\mathbf{z}}
\newcommand{\TFspanspace}{X}
\newcommand{\CTFspanspace}{X^{C}}
\newcommand{\wiener}{W}
\newcommand{\Wiener}[3][]{
	\wiener_{
		#2
		\ifx&#1&
		\else
		,[#1]
		\fi
		|#3}
	}
\newcommand{\CWiener}[3][]{
	\wiener^{C}_{
		#2
		\ifx&#1&
		\else
		,[#1]
		\fi
		|#3}
	}
\newcommand{\wienerprojection}[2]{ {\hat #1}_{#2}}

\newcommand{\ldim}{\mathcal{\graph}}
\newcommand{\adjTF}{\genericmatrix}
\newcommand{\nodeset}{V}
\newcommand{\node}{y}
\newcommand{\nodealt}{x}
\newcommand{\nodenumber}{n}
\newcommand{\nodeindex}{i}
\newcommand{\nodeindexalt}{j}
\newcommand{\nodeindexaltb}{k}
\newcommand{\nodeindexstart}{1}
\newcommand{\nodeindexend}{\nodenumber}

\newcommand{\subsetnodeindex}{I}
\newcommand{\subsetnodeindexalt}{J}
\newcommand{\subsetnodeindexsep}{Z}

\newcommand{\subsetnodeindexsepchild}{\subsetnodeindexsep_{\subsetnodeindex}}

\newcommand{\subsetnodeindexsepchildalt}{\subsetnodeindexsep_{\subsetnodeindexalt}}

\newcommand{\DSFOutputMatrix}{Q}
\newcommand{\DSFInputMatrix}{P}

\newcommand{\subsetnode}{{\node}_{\subsetnodeindex}}
\newcommand{\subsetnodeprime}{{\node}_{\subsetnodeindex'}}
\newcommand{\subsetnodealt}{{\node}_{\subsetnodeindexalt}}
\newcommand{\subsetnodealtprime}{{\node}_{\subsetnodeindexalt'}}
\newcommand{\subsetnodesep}{{{\node}_{\subsetnodeindexsep}}}
\newcommand{\subsetnodesepchild}{{{\node}_{\subsetnodeindexsepchild}}}
\newcommand{\subsetnodesepchildalt}{{{\node}_{\subsetnodeindexsepchildalt}}}
\newcommand{\wsep}{wsep}
\newcommand{\cwsep}{cwsep}

\newcommand{\subsetindexnodesmallertocheckdsep}{U}
\newcommand{\nodesetsmallertocheckdsep}{\node_{\subsetindexnodesmallertocheckdsep}}
\newcommand{\nodesetsmallertocheckdsepcomplement}{{\nodesetsmallertocheckdsep^{c}}}

\newcommand{\subsetnodeindexenlarged}{\bar{I}}
\newcommand{\subsetnodeindexaltenlarged}{\bar{J}}
\newcommand{\subsetnodeenlarged}{{\node}_{\subsetnodeindexenlarged}}
\newcommand{\subsetnodealtenlarged}{{\node}_{\subsetnodeindexaltenlarged}}

\newcommand{\inputsignal}{u}
\newcommand{\noise}{e}
\newcommand{\noisenumber}{\nodenumber}

\newcommand{\noiseindexstart}{1}
\newcommand{\noiseindexend}{\noisenumber}
\newcommand{\subsetnoise}{{\noise}_{I}}
\newcommand{\subsetnoisealt}{{\noise}_{J}}

\newcommand{\pdg}{PDG}

\title{\LARGE \bf
	Signal selection for estimation and identification in networks of dynamic systems: a graphical model approach
}

\author{
	Donatello Materassi \quad
	Murti V.~Salapaka 
	\thanks{Donatello Materassi is with Department of Electrical Engineering and Computer Science, University of Tennessee.}
	\thanks{Murti Salapaka is with Department of Electrical and Computer Engineering, University of Minnesota.}
}

\begin{document}

%
%
%
\maketitle
\begin{abstract}
    Network systems have become a ubiquitous modeling tool in many areas of science where nodes in a graph represent distributed processes and edges between nodes represent a form of dynamic coupling.
    When a network topology is already known (or partially known), two associated goals are (i) to derive estimators for nodes of the network which cannot be directly observed or are impractical to measure; and (ii) to quantitatively identify the dynamic relations between nodes. 
    In this article we address both problems in the challenging scenario where only some outputs of the network are being measured and the inputs are not accessible.
    The approach makes use of the notion of $d$-separation for the graph associated with the network.
    In the considered class of networks, it is shown that the proposed technique can determine or guide the choice of optimal sparse estimators.
    The article also derives identification methods that are applicable to cases where loops are present providing a different perspective on the problem of closed-loop identification.
    The notion of $d$-separation is a central concept in the area of probabilistic graphical models, thus an additional contribution is to create connections between control theory and machine learning techniques.
\end{abstract}

\section*{Introduction}

Networks of dynamic systems have become a ubiquitous tool in science to describe interactions in modular systems.
Applications span areas as diverse as 
economics (see e.g. \cite{AtaHor11}),
social systems (see e.g. \cite{AceDah2011,chetty2017applying})
biology (see e.g. \cite{EisSpe98,DelNin08}),
cognitive sciences (see e.g. \cite{QuiCol11}),
and
geology (see e.g. \cite{BaiMon06}).
In many of these applications 
the system variables being observed are not the responses to known or manipulable inputs that are actively injected to probe the network, but rather are measurements acquired while the system is naturally operating and is responding to excitations that are unknown.
Scenarios where inputs are neither manipulable nor measurable include systems with critical or uninterruptible functions where injecting probing signals is deemed unacceptable to system operations (such as the power grid, or the transport infrastructure), or when it is too expensive or impractical to actively manipulate the nodes (such as a financial network \cite{Ale01} or a gene network \cite{EisSpe98}).

This article considers networks in the Dynamical Structure Function (DSF) form, introduced in \cite{GonWar08}, and addresses two fundamental problems in the domain of this class of networked systems:
(i) estimation of the output signal (behavior) of an individual node,
and
(ii) consistent identification of the dynamic relation (transfer function) between two nodes in the network.
Both problems are investigated  in the presence of latent agents with focus on 
methods that preclude the manipulation or observation of the inputs.
The article also bridges problems of signal estimation and consistent identification to the graphical model framework.

Toward estimation of an agent's behavior in a dynamically networked system, 
graphical conditions are leveraged to derive simpler estimators which make use of a reduced number of signals.
One first fundamental result is a rigorous characterization of the effects of observations' inclusion or exclusion on the estimate of an agent's activity.
For example, when a set of observed nodes is used to estimate another node, it stands to reason that the introduction of new observed nodes renders some of the nodes in the original set irrelevant.
The proven characterization also demonstrates that it is possible to remove some observations and make other observations irrelevant, too.

The prevalent approach in identification theory assumes the design of an appropriate input to inject into the system (or subsystem) to be identified \cite{goodwin1977dynamic}.
In the case of identification in networks, some approaches also consider knocking-out one or more agents in their identification procedure \cite{NabMes10}.
Since these techniques rely on the design of an input or on the active suppression agents in the networks, they are not applicable to systems that are not manipulable.
There are also several data-driven techniques which aim at reconstructing the underlying graph of a network \cite{MatInn10,MatSal12} and which also obtain the transfer functions describing the network's dynamics among the nodes  \cite{GonWar08,yuan2011robust,MatSal12,yue2017linear,hayden2017network}.
Some of these techniques make limiting assumptions on the structures they can identify (i.e. \cite{MatInn10} is limited to trees). Others assume that the system is not manipulable, but still rely on the measurement of the forcing inputs.
For example \cite{GonWar08} finds necessary and sufficient informativity conditions to recover the network's dynamics among the nodes given the transfer function from inputs to outputs (or an estimate of it). 
If the inputs are not measured and just modeled as stochastic, multiple transfer functions could be compatible with the observed outputs, given that the spectral factorization of their power spectral density is in general not unique \cite{hayden2017network}.
Work in \cite{yuan2011robust} robustifies the approach in \cite{GonWar08} incorporating the priori knowledge of the network structure in an optimization problem, but still requires an estimate of the transfer function from inputs to outputs or the actual measurements of the forcing inputs.
Among the approaches based on DSF models, \cite{hayden2017network} is arguably the most related to this article.
The authors consider a stochastic scenario, referred to as ``blind identification'', where the forcing inputs are not accessible and the objective of the identification has to be determined from the output spectral density only.
In \cite{hayden2017network}, under the assumption of strictly causality for the network dynamics and mutually independent forcing inputs, it is shown that, given the outputs' spectral density matrix, its spectral factorizations compatible with the DSF dynamics are a finite number.  Furthermore, under the additional minimum phase condition on the network's dynamics, the spectral factor is shown to be unique (modulo a multiplication by a signed identity matrix), so that the network can be fully identified via the methods in \cite{GonWar08}.
As an important feature, \cite{hayden2017network} requires no a priori knowledge of the network structure.
Under the same strict causality and minimum-phase conditions, similar results which do not require a priori knowledge of the structure are obtained in \cite{MatSal12} using a Granger causality approach \cite{Gra69}. 

This article follows a problem formulation close to the blind identification in \cite{hayden2017network}: the considered class of networks is expressed in a DSF form and the unknown inputs are modeled as stochastic processes.
Some of the assumptions in \cite{hayden2017network} are relaxed: the dynamics is not required to be strictly causal, and only some of the outputs are being measured. Thus, while \cite{hayden2017network} assumes full knowledge of the power spectral density matrix of the outputs, here only a minor of the output spectral density matrix is known.
At the same time, compared to \cite{hayden2017network}, the results of this article make use of partial a priori knowledge about the network structure and attempt to identify a single transfer function and not the whole dynamics.
An example (Example~\ref{ex:DSF comparison}) shows how the identification problem formulated in this article can also be applied to formulate blind identification problems as in \cite{hayden2017network}, but considering non-mutually independent forcing inputs.

On the other hand, if the network topology is fully known, there is a wealth of closed loop parametric identification techniques which have been extended to general networks in order to identify individual transfer functions.
These techniques include the Direct Method, the Joint IO method, the Two-Stage identification \cite{van2013identification}, and the instrumental variable method \cite{DanVan15}.
Most of these approaches can address scenarios where inputs are non manipulable if observed internal signals in the network can be used as predictors and/or instrumental variables.
They result in sufficient conditions on how to select these signals to arrive at consistent estimation of an individual transfer function \cite{DanVan16}.
However, these techniques often require some form of knowledge about the strict causality or the degree of delay embedded in some operators (especially if feedback loops are present \cite{DanVan15}).

This article follows a related, but non-parametric, approach.
The identification of a transfer function is achieved by computing a multi-input Wiener filter estimating a signal over a selected set of other signals, in such a way that the one of the entries of the Wiener filter matches the expression of the transfer function to be identified.
A precise characterization is provided for the set of agents whose activities, if available, result in an unbiased and consistent estimation of the dynamics between a specific pair of agents. 
No knowledge about the strict causality of specific operators is required and, compared to other techniques, the results hold even when the transfer functions are zero, but this information is not available.
These features allow one to seamlessly deal with latent nodes and situations where the network topology is not known exactly, by using purely graphical criteria.

The resulting conditions are graphical and constructive.  They depend exclusively on the structure of the network and the location of the edge associated with the transfer function to be identified.
Furthermore, the articles' framework builds bridges between the area of probabilistic graphical models \cite{Pea88} and extends them to the area of network of dynamical systems \cite{MatSal12,materassi2015identification}.
In particular, in this article, the notion of $d$-separation on graphs, widely used in the domain of probabilistic graphical models (see \cite{Pea88}), plays an enabling role: it is established that $d$-separation on the graphical representation implies a notion of independence for signals in networks of linear dynamical systems based on Wiener filtering. This is achieved both in the non-causal and in the more challenging causal case, broadening the framework developed in \cite{MatSal14}.


The article has the following structure. Section~\ref{sec:preliminary} introduces preliminary concepts including the notion of $d$-separation for directed graphs that is key for the development of the main results. Section~\ref{sec:LDMI} introduces a flexible class of networks via the DSF formalism. Section~\ref{sec:problem} formalizes the two problems which are the focus of this article.
Section~\ref{sec:manipulations} describes graphical manipulations for our network models.
Section~\ref{sec:dseparation} shows that $d$-separation defined on the graph representation of implies specific sparsity properties of Wiener filters on such class of networks.
Section~\ref{sec:idresults} makes use of $d$-separation to select appropriate signals in order to obtained unbiased identification of a transfer function in the network.
Throughout the article, efficacy of the results are illustrated via examples.

\section*{Notation and basic definitions}
\noindent
\begin{itemize}
    \item   $\{\node_{\nodeindex}\}_{\nodeindex=1}^{N}$: set of $N$ column vector processes indexed by $\nodeindex$
    \item   $\node=(\node_{1}^{\transposesymbol},...,\node_{N}^{\transposesymbol})^{\transposesymbol}$: vector of stochastic processes obtained by stacking $\node_{1},...,\node_{N}$
    \item   $I,J$: set of indeces in $\{1,...,N\}$
    \item   $\PSD_{\node}(\ZTvar)$: power spectral density of the vector process $\node$
    \item   $G=(V,\vec{E})$: directed graph with nodes $V$ and edges $\vec{E}$
    \item   $\Wiener{\nodealt}{\node}(\ZTvar)$ and $\CWiener{\nodealt}{\node}(\ZTvar)$: respectively non-causal and causal Wiener filter estimating the process $\nodealt$ from the process $\node$
    \item   $\Wiener[\node_{\nodeindex}]{\nodealt}{\node}(\ZTvar)$, $\CWiener[\node_{\nodeindex}]{\nodealt}{\node}(\ZTvar)$: component of the non-causal and causal Wiener filter estimating the process $\nodealt$ from the process $\node$ associated with the subvector $\node_{\nodeindex}$ of $\node$
\end{itemize}

The following notation for the definition of subvectors and submatrices will be adopted through the article.
\begin{ntn}[Ordered and unordered sets]
	An unordered set (a collection of elements where the listing order is not relevant) is denoted using curly brackets.
	An ordered set (a collection of elements where the listing order is relevant) is denoted using round brackets.
\end{ntn}
We extend the notions of element and subset of an unordered set to ordered sets in the natural way, so that we can write
$3\in (1,2,3,4,5)$ and $(1,2,4)\subseteq (5,4,3,2,1)$

For vectors and matrices, we use an index subset notation.
\begin{ntn}[Index subset notation for vectors]
	Let
	$\genericvector
	:=
	(\genericvector_{\genericvectorindexstart}^{\transposesymbol}
		,...,
		\genericvector_{\genericvectorindexend}^{\transposesymbol})^{\transposesymbol}$
	be a vector defined by
	$\genericvectordimension$
	subvectors,
	$\genericvector_{\genericvectorindexstart},...,\genericvector_{\genericvectorindexend}$,
	and let
	$\genericvectorindexsubset
	:=
	(\genericvectorindexsubsetindexstart,...,\genericvectorindexsubsetindexend)$
	be an ordered set of integers in
	$\left\{\genericvectorindexstart,...,\genericvectorindexend\right\}$.
	We denote by
	$\genericvector_{\genericvectorindexsubset}
	:=
	(\genericvector_{\genericvectorindexsubsetindexstart}^{\transposesymbol}
	,...,
	\genericvector_{\genericvectorindexsubsetindexend}^{\transposesymbol})^{\transposesymbol}$,
	the vector obtained by considering the subvectors in
	$\genericvector$
	indexed by
	$(\genericvectorindexsubsetindexstart,...,\genericvectorindexsubsetindexend)$.
\end{ntn}
\begin{ntn}[Index subset notation for matrices]
	Let
	$\genericmatrix$
	be a matrix with a
	$\genericmatrixdimensionA \times \genericmatrixdimensionB$
	block structure and let
	$\genericvectorindexsubset
	:=
	(\genericvectorindexsubsetindexstart,...,\genericvectorindexsubsetindexend)$
	and
	$\genericvectorindexsubsetalt
	:=
	(\genericvectorindexsubsetaltindexstart,...,\genericvectorindexsubsetaltindexend)$
	be two ordered sets of integers in
	$\left\{\genericvectorindexstart,...,\genericvectorindexend\right\}$.
	We denote by
	$\genericmatrix_{\genericvectorindexsubsetalt \genericvectorindexsubset}$
	the submatrix of
	$\genericmatrix$
	obtained by considering the row blocks indexed by
	$(\genericvectorindexsubsetaltindexstart,...,\genericvectorindexsubsetaltindexend)$.
	and the column blocks indexed by
	$(\genericvectorindexsubsetindexstart,...,\genericvectorindexsubsetindexend)$.
\end{ntn}

\section{Preliminaries}\label{sec:preliminary}
In this section we recall some basic concept of graph theory, including the notion of $d$-separation, and define the class of Linear Dynamic Influence Models, derived from the notion of Dynamical Structural Function (DSF) introduced in \cite{GonWar08}.
\subsection{Basic notions of graph theory}
In this section we want to recall some fundamental notions of graph theory which are functional to the subsequent development.
We define directed graphs and their restrictions with respect to subsets of nodes (see also \cite{Die06}).
\begin{dfn}[Directed Graphs and Restrictions]
	A {\it directed graph}
	$\graph$
	is a pair
	$(\vertexset,\directededgeset)$
	where
	$\vertexset$
	is a set of vertices (or nodes) and
	$\directededgeset$
	is a set of edges (or arcs) which are ordered pairs of elements of
	$\vertexset$.
	The restriction of
	$\graph$
	to the node set
	$\vertexset'\subseteq \vertexset$
	is the graph
	$\graph'=(\vertexset ',\directededgeset\,')$
	where
	$\directededgeset\,'=
	\{
		(\vertex_{\vertexindex},\vertex_{\vertexindexalt})\in \directededgeset\,|\,
		\vertex_{\vertexindex}\in \vertexset'
		\,\text{and}\,
		\vertex_{\vertexindexalt}\in \vertexset'
	\}$.
	{$\square$}
\end{dfn}
\noindent Figure~\ref{fig:genericgraph}(a) shows a directed graph and Figure~\ref{fig:genericgraph}(b) shows its restriction to a subset of its nodes.
\begin{figure}[h!]
	\centering
	\begin{tabular}{cc}
		\includegraphics[width=0.35\columnwidth]{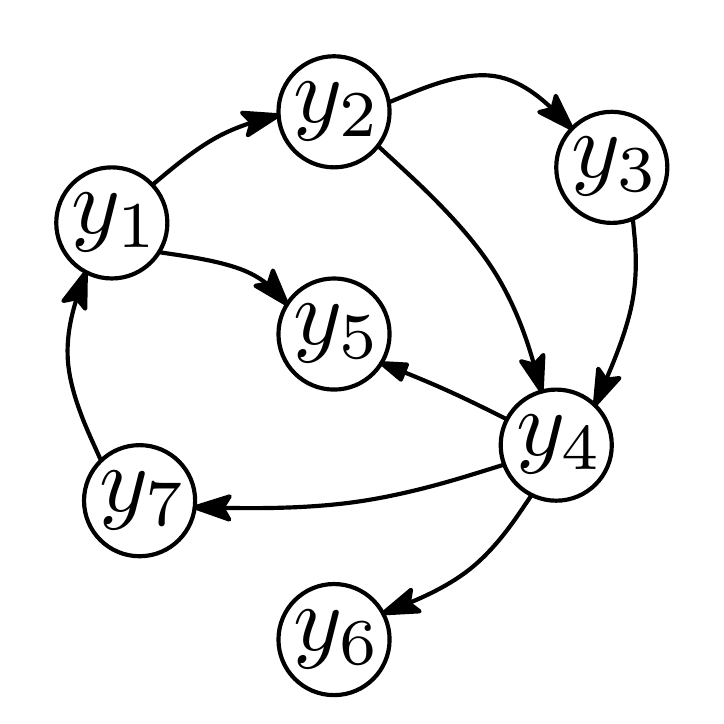} &
		\includegraphics[width=0.35\columnwidth]{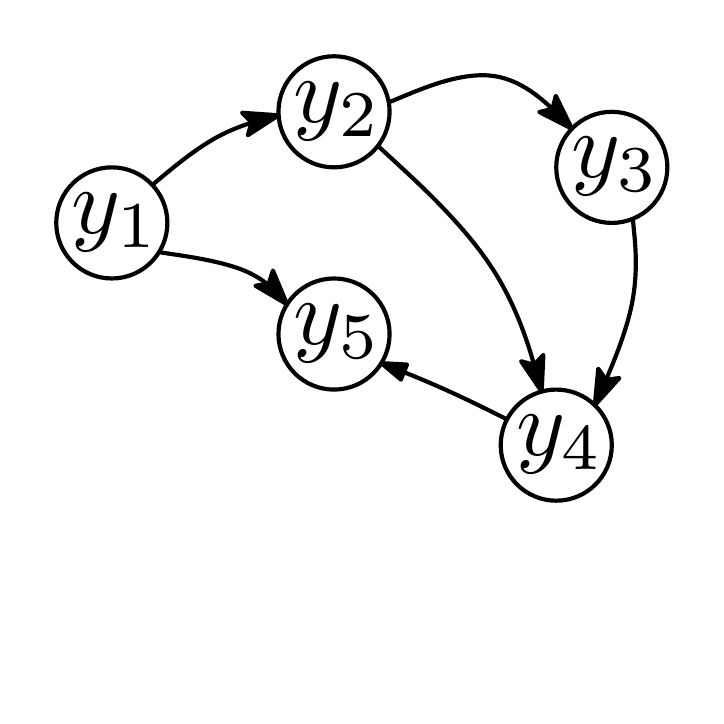}\\
		(a) & (b)
	\end{tabular}
	\caption{A directed graph (a) and its restriction to the nodes $\node_{1}$, $\node_{2}$, $\node_{3}$,$\node_{4}$, and $\node_{5}$ (b).
	\label{fig:genericgraph}
	}
\end{figure}

Self-loops are edges that start and terminate on the same node.
\begin{dfn}[Self-loops]
	A directed graph
	$\graph=(\vertexset,\directededgeset)$
	has a self-loop on node
	$\vertex_{\vertexindex}\in\vertexset$
	if
	$(\vertex_{\vertexindex},\vertex_{\vertexindex})\in \directededgeset$.
\end{dfn}

\noindent Chains and paths are fundamental concepts in graphs. 
\begin{dfn}[Paths, Chains]
	Consider a directed graph
	$\graph=(\vertexset, \directededgeset)$
	with vertices
	$\vertex_{1},...,\vertex_{\vertexnumber}$.
	A {\it chain} starting from
	$\vertex_{\vertexindex}$
	and ending in
	$\vertex_{\vertexindexalt}$
	is an ordered sequence of distinct edges in
	$\directededgeset$
	\begin{align*}
		(\,(\vertex_{\path_{1}}, \vertex_{\path_{2}}), (\vertex_{\path_{2}}, \vertex_{\path_{3}}), ...\,,(\vertex_{\path_{\pathlength-1}},\vertex_{\path_{\pathlength}})\,)
	\end{align*}
	where
	$\vertex_{\vertexindex}=\vertex_{\path_{1}}$, 
	$\vertex_{\vertexindexalt}=\vertex_{\path_{\pathlength}}$,
	and
	$(\vertex_{\path_{\pathnodeindex}},\vertex_{\path_{\pathnodeindex+1}})\in \directededgeset$
	for all
	$\pathnodeindex={\pathindexstart},...,{\pathindexendminusone}$. 
	A {\it path} between two vertices
	$\vertex_{\vertexindex}$
	and
	$\vertex_{\vertexindexalt}$
	is an ordered sequence of distinct ordered pairs of nodes
	\begin{align*}
		(\,(\vertex_{\path_{\pathindexstart}}, \vertex_{\path_{2}}), (\vertex_{\path_{2}}, \vertex_{\path_{3}}), ...\,,(\vertex_{\path_{\pathindexend-1}},\vertex_{\path_{\pathindexend}})\,)
	\end{align*}
		where
	$\vertex_{\vertexindex}=\vertex_{\path_{\pathindexstart}}$,
	$\vertex_{\vertexindexalt}=\vertex_{\path_{\pathindexend}}$,
	and
	either 
	$(\vertex_{\path_{\pathnodeindex}},\vertex_{\path_{\pathnodeindex+1}})\in \directededgeset$
	or
	$(\vertex_{\path_{\pathnodeindex+1}},\vertex_{\path_{\pathnodeindex}})\in \directededgeset$
	for all
	$\pathnodeindex={\pathindexstart},...,{\pathindexendminusone}$. 
	 We refer to the nodes $\vertex_{\path_{2}},...,\vertex_{\path_{\pathindexendminusone}}$ as internal nodes of a chain or path.
	 {$\square$}
\end{dfn}
As it follows from the definition, chains are a special case of paths.
All paths (and consequently also all chains) can be suggestively denoted by separating the nodes in the sequence $\{\vertex_{\path_{\pathnodeindex}}\}_{\pathnodeindex=\pathindexstart}^{\pathindexend}$
with the arrow symbol
$\rightarrow$
if 
$(\vertex_{\path_{\pathnodeindex-1}},\vertex_{\path_{\pathnodeindex}})\in \directededgeset$
or the symbol
$\leftarrow$
if 
$(\vertex_{\path_{\pathnodeindex}},\vertex_{\path_{\pathnodeindex-1}})\in \directededgeset$.
For example, in Figure~\ref{fig:genericgraph}(a), the path
$\node_{1}\rightarrow \node_{2} \rightarrow \node_{4} \rightarrow \node_{6}$
is also a chain, while 
$\node_{1}\rightarrow \node_{5} \leftarrow \node_{4} \leftarrow \node_{3}$
is a path, but not a chain.
Here we also say that, there is a chain from vertex
$\node_{1}$
to vertex
$\node_{6}$
in the graph.

From the concept of chain, we can derive the notions of ancestry and descendance, which are common in the theory of bayesian networks (see also \cite{Pea88}).

\begin{dfn}[Parents, children, ancestors, descendants]
	Consider a directed graph $\graph=(\vertexset,\directededgeset)$.
	A vertex
	$\vertex_{\vertexindex}$
	is a {\it parent} of a vertex
	$\vertex_{\vertexindexalt}$
	if there is a directed edge from
	$\vertex_{\vertexindex}$
	to
	$\vertex_{\vertexindexalt}$.
	In such a case
	$\vertex_{\vertexindexalt}$
	is a {\it child} of
	$\vertex_{\vertexindex}$.
	Also
	$\vertex_{\vertexindex}$
	is an  {\it ancestor} of
	$\vertex_{\vertexindexalt}$
	if
	$\vertex_{\vertexindex}=\vertex_{\vertexindexalt}$
	or if there is a chain from
	$\vertex_{\vertexindex}$
	to
	$\vertex_{\vertexindexalt}$.
	In such a case $\vertex_{\vertexindexalt}$
	is a {\it descendant} of $\vertex_{\vertexindex}$.
	Given a set $\subsetvertex \subseteq \vertexset$, we define the following sets
	\small
	\begin{align*}
		& \parent{\graph}{\subsetvertex}
		:=
		\left\{\vertexalt\in \vertexset \,|\,
			 \exists \,\vertex\in\subsetvertex:
			 \vertexalt
			 \text{ is a parent of }
			 \vertex
		\right\},\\
		& \child{\graph}{\subsetvertex}
		:=
		\left\{\vertexalt\in \vertexset \,|\,
			 \exists \,\vertex\in\subsetvertex:
			 \vertexalt
			 \text{ is a child of }
			 \vertex
		\right\},\\
		& \ancestor{\graph}{\subsetvertex}
		:=
		\left\{\vertexalt\in \vertexset \,|\,
			 \exists \,\vertex\in\subsetvertex:
			 \vertexalt
			 \text{ is an ancestor of }
			 \vertex
		\right\}, \text{ and}\\
		& \descendant{\graph}{\subsetvertex}
		:=
		\left\{\vertexalt\in \vertexset \,|\,
			 \exists \,\vertex\in\subsetvertex:
			 \vertexalt
			 \text{ is a descendant of }
			 \vertex
		\right\}.
	\end{align*}
	\normalsize
\end{dfn}
On a path we define forks and colliders (see again \cite{Pea88}).
\begin{dfn}[Forks and colliders]
	A path involving the nodes
	$\vertex_{\path_{\pathindexstart}},...,\vertex_{\path_{\pathindexend}}$
	has a {\it fork} at
	$\vertex_{\path_{\pathnodeindex}}$,
	for $\pathindexstart<p<\pathindexend$,
	if
	$\vertex_{\path_{\pathnodeindex-1}}$
	and
	$\vertex_{\path_{\pathnodeindex+1}}$
	are both children of
	$\vertex_{\path_{\pathnodeindex}}$
	(that is
	$\vertex_{\path_{\pathnodeindex-1}}
	\leftarrow 
	\vertex_{\path_{\pathnodeindex}}
	\rightarrow
	\vertex_{\path_{\pathnodeindex+1}}$
	appears in the path).
	A path has an {\it inverted fork} (or a collider) at
	$\vertex_{\path_{\pathnodeindex}}$
	if
	$\vertex_{\path_{\pathnodeindex-1}}$
	and
	$\vertex_{\path_{\pathnodeindex+1}}$
	are both parents of
	$\vertex_{\path_{\pathnodeindex}}$
	(that is
	$\vertex_{\path_{\pathnodeindex-1}}
	\rightarrow 
	\vertex_{\path_{\pathnodeindex}}
	\leftarrow
	\vertex_{\path_{\pathnodeindex+1}}$
	appears in the path).
\end{dfn}
The following definition introduces a notion of separation on subsets of vertices in a directed graph \cite{Pea88}.
\begin{dfn}[$d$-separation and $d$-connection]
	Consider a directed graph
	$\graph=(\vertexset,\directededgeset)$
	and three mutually disjoint sets of vertices
	$\subsetvertex,\subsetvertexsep,\subsetvertexalt\subseteq \vertexset$.
	The set
	$\subsetvertexsep$
	is said to
	$d$-separate
	$\subsetvertex$
	and
	$\subsetvertexalt$
	if for every
	$\vertex_{\vertexindex}\in \subsetvertex$
	and
	$\vertex_{\vertexindexalt}\in \subsetvertexalt$
	all paths between
	$\vertex_{\vertexindex}$ and $\vertex_{\vertexindexalt}$
	meet at least one of the following conditions:
	\begin{enumerate}
		\item	the path contains a node
				$\vertex_{\vertexindexsep}\in \subsetvertexsep$
				that is not a collider
		\item	the path contains a collider
				$\vertex_{\vertexindexaltb}$
				such that neither
				$\vertex_{\vertexindexaltb}$
				nor its descendants belong to
				$\subsetvertexsep$.
	\end{enumerate}
	If
	$\subsetvertexsep$
	$d$-separates
	$\subsetvertex$
	and
	$\subsetvertexalt$
	in the graph 
	$\graph$,
	we write
	$\dsep{\graph}(\subsetvertex,\subsetvertexsep,\subsetvertexalt)$.
	Otherwise we write
	$\neg\dsep{\graph}(\subsetvertex,\subsetvertexsep,\subsetvertexalt)$
	and say that
	$\subsetvertexsep$
	$d$-connects
	$\subsetvertex$
	and
	$\subsetvertexalt$.
\end{dfn}
Since the notion of $d$-separation plays a central role in the developments of this article we illustrate it with some examples.
\begin{example}[Examples of $d$-separation]
	\normalsize
	In the graph
	$\graph$
	of Figure~\ref{fig:genericgraph}(a)
	we have that
	$\dsep{\graph}(\node_{3},\{\node_{4}\},\node_{6})$.
	Also we have that 
	$\neg\dsep{\graph}(\node_{4},\{\node_{2},\node_{5},\node_{7}\},\node_{1})$
	because
	$\node_{5}$
	is a collider on the path
	$\node_{4}\rightarrow\node_{5}\leftarrow\node_{1}$.
	The smallest set that $d$-separates $\node_{1}$ and $\node_{4}$ is $\{\node_{2},\node_{7}\}$, but it also holds that 
	$\dsep{\graph}(\node_{4},\{\node_{2},\node_{3},\node_{7}\},\node_{1})$,
	$\dsep{\graph}(\node_{4},\{\node_{2},\node_{6},\node_{7}\},\node_{1})$, and
	$\dsep{\graph}(\node_{4},\{\node_{2},\node_{3},\node_{6},\node_{7}\},\node_{1})$.
	We have that 
	$\neg\dsep{\graph}(\node_{2},\{\node_{1},\node_{6}\},\node_{3})$
	because 
	$\node_{6}$ is a descendant of $\node_{4}$ which is a collider in the path
	$\node_{2}\rightarrow\node_{4}\leftarrow\node_{3}$.
\end{example}

\subsection{Wiener Filter and Wiener separation}
We introduce a class of linear operators to transform sets of rationally related random processes.
\begin{dfn}
	The set $\mathcal{F}$ is defined as the set of real-rational transfer functions that are analytic and invertible on the unit circle $\{\ZTvar\in\ComplexField|~~|\ZTvar|=1\}$.
	Given a transfer function $\DSFInputMatrix(\ZTvar)\in\mathcal{F}$, it can be uniquely represented in the time domain by a bi-infinite sequence $p_k$ (the impulse response of $\DSFInputMatrix(\ZTvar)$) satisfying
	\begin{align}
		\DSFInputMatrix(\ZTvar)=\sum_{k=-\infty}^{\infty}p_k \ZTvar^{-k}
	\end{align}
	for all $|\ZTvar|=1$. 
	If $k< 0$ implies $p_{k}=0$, then we say that the transfer function is causal. We define the space of causal transfer functions as $\mathcal{F}^{+}$.
	If $k\leq 0$ implies $p_{k}=0$, then we say that the transfer function is strictly causal.
\end{dfn}
We also define the following notation.
\begin{dfn} We define $\node = \DSFInputMatrix(\ZTvar)\inputsignal$ as the process obtained by computing the convolution of $\inputsignal$ with the impulse response of $\DSFInputMatrix(\ZTvar)$, namely $\node(t) = \sum_{-\infty}^{\infty}{p_{t-k}\inputsignal(k)}$ for all $t \in \mathbb{Z}$. \end{dfn}
For example, $\frac{1}{\ZTvar}\inputsignal$ denotes the process $\inputsignal$ delayed by one time step.
For a set of processes  $S=\{\node_{1},...,\node_{n}\}$, we denote by $\DSFInputMatrix(\ZTvar)S$ the set $\{\DSFInputMatrix(\ZTvar)\node_{1},...,\DSFInputMatrix(\ZTvar)\node_{n}\}$. For example $\frac{1}{\ZTvar}S$ denotes the set of all the processes in $S$ delayed by one time step.  
The defined linear operators can be used to define the following spaces of stochastic processes.
\begin{dfn}[Transfer Function Spaces and Spans]
	For a finite number of wide-sense stationary signals
	$\node_{\nodeindexstart},...,\node_{\nodeindexend}$,
	the rational transfer function span (tfspan) and the rational causal transfer span (ctfspan) are defined as
	\small
	\begin{align*}
		\text{tfspan}\{\node_{\nodeindexstart},...,\node_{\nodeindexend}\}
		&:=
		\left\{q=\sum_{\nodeindex=\nodeindexstart}^{\nodeindexend} \DSFInputMatrix_{i}(\ZTvar)
			\node_{\nodeindex}
		~\lvert~\DSFInputMatrix_{i}(\ZTvar)\in\mathcal{F}
		\right\}\\
        \text{ctfspan}\{\node_{\nodeindexstart},...,\node_{\nodeindexend}\}
        &:=
		\left\{q=\sum_{\nodeindex=\nodeindexstart}^{\nodeindexend} \DSFInputMatrix_{i}(\ZTvar)
			\node_{\nodeindex}
		~\lvert~\DSFInputMatrix_{i}(\ZTvar)\in\mathcal{F}^{+}
		\right\}.
	\end{align*}
	\normalsize
	where the variable $\ZTvar$ is related to a discrete time-shift operator.
\end{dfn}
We provide a specific formulation for Wiener filter that covers both non-causal and causal estimators.
\begin{prop}[Wiener Filtering for processes with rational Power Spectral Density]\label{prop: my wiener}
	Let
	$\nodealt$
	and
	$\node_{\nodeindexstart},...,\node_{\nodeindexend}$
	be  vector processes that are jointly stationary
	with rational power cross-spectral densities and rational auto-spectral densities.
	Define
	$\node:=(\node_{\nodeindexstart}^{T},...,\node_{\nodeindexend}^{T})^T$
	and
	\begin{align*}
		\TFspanspace
		&:=
		\text{tfspan}\{\node_{\nodeindexstart},...,\node_{\nodeindexend}\}\\
		\CTFspanspace
		&:=
		\text{ctfspan}\{\node_{\nodeindexstart},...,\node_{\nodeindexend}\}.
	\end{align*}
	Consider the problems
	\begin{align*}
		\inf_{q \in \TFspanspace} \E[(\nodealt-q)^{T}(\nodealt-q)] 
		\quad \text{and}
		\inf_{q \in \CTFspanspace} \E[(\nodealt-q)^{T}(\nodealt-q)] 
	\end{align*}
	where,
	$\E[\cdot],$
	is the expectation operator.
	Then the solutions for both problems exist 
	are unique and given respectively by
	$\wienerprojection{\nodealt}{\node}=W_{\nodealt|\node}(z)\node$
	and
	$\wienerprojection{\nodealt}{\node}^{C}=W^{C}_{\nodealt|\node}(\ZTvar)\node$
	where,
    if
	$\PSD_{\node}(e^{i\w})\succ 0$,
	for
	$\w\in[-\pi,\pi]$,
    the transfer function $W_{\nodealt|\node}(z)$ and $W^{C}_{\nodealt|\node}(z)$ are unique and referred to as Wiener filter and causal Wiener filter, respectively.
	Moreover, $\wienerprojection{\nodealt}{\node}$ is the only element in $\TFspanspace$ such that, for any $q\in \TFspanspace$,
	\begin{align}\label{eq:perp hilbert projection thm}
		\E[(\nodealt-\wienerprojection{\nodealt}{\node})q]=0
	\end{align}
	and
	$\wienerprojection{\nodealt}{\node}^{C}$ is the only element in $\CTFspanspace$ such that, for any $q\in \CTFspanspace$,
	\begin{align}\label{eq:perp hilbert projection thm causal}
		\E[(\nodealt-\wienerprojection{\nodealt}{\node}^{C})q]=0.
	\end{align}
\end{prop}
\begin{proof}
    The result is a direct consequence of standard Wiener filter theory. Rational power spectal densities guarantee the existence of a spectral factorization. See \cite{KaiSay00}.
\end{proof}
\begin{dfn}
	Let
	$\nodealt$
	and
	$\node=(\node_{\nodeindexstart}^{T}|...|\node_{\nodeindexend}^{T})^T$
	be jointly stationary processes.
    Let
	$\Wiener{\nodealt}{\node}(z)\node$
	and
	$\CWiener{\nodealt}{\node}(z)\node$
	be respectively the non-causal and causal Wiener estimates of 
	$\nodealt$
	from
	$\node$.
	We denote the component of the Wiener filter operating on $\node_{\nodeindex}$ as
	$\Wiener[\node_{\nodeindex}]{\nodealt}{y}$, namely
	\begin{align*}
		\Wiener{\nodealt}{y}&=
			\left(\Wiener[\node_{\nodeindexstart}]{\nodealt}{y},
				\cdots, \Wiener[\node_{\nodeindexend}]{\nodealt}{y} \right)\\
		\CWiener{\nodealt}{y}&=
			\left(\CWiener[\node_{\nodeindexstart}]{\nodealt}{y},
				\cdots, \CWiener[\node_{\nodeindexend}]{\nodealt}{y} \right).
	\end{align*}
	As an extension, given an ordered subset
	$\subsetnode$
	of components of
	$\node$,
	we use the notation 
	$\Wiener[\subsetnode]{\nodealt}{y}$ 
	(or $\CWiener[\subsetnode]{\nodealt}{y}$)
	to denote the vector of components of the non-causal Wiener filter (or causal Wiener filter) associated with the vector of signals
	$\subsetnode$
	in the estimation of 
	$\nodealt$.
\end{dfn}
\begin{dfn}[Wiener separation]
	Let
	$\subsetnode$,
	$\subsetnodealt$
	and
	$\subsetnodesep$
	be disjoint subsets in a set of jointly stationary random processes.
	We say that
	$\subsetnode$
	is (non-causally) Wiener-separated from
	$\subsetnodealt$
	by
	$\subsetnodesep$
	if the Wiener estimate of
	$\subsetnodealt$
	from
	$\subsetnode$
	and the Wiener estimate of 
	$\subsetnodealt$
	from
	$\subsetnode\cup \subsetnodesep$ are identical, namely:
	\begin{align*}
		\Wiener{\subsetnodealt}{(\subsetnode,\subsetnodesep)}(\ZTvar)
		\left(\begin{array}{c}
			\subsetnode\\\subsetnodesep
		\end{array}\right)
		=
		\Wiener{\subsetnodealt}{\subsetnodesep}(\ZTvar)\,
		\subsetnodesep.
	\end{align*}
	We denote this relation as
	$wsep(\subsetnode,\subsetnodesep,\subsetnodealt)$.
	An equivalent way of expressing $wsep(\subsetnode,\subsetnodesep,\subsetnodealt)$ is by saying that all components of the Wiener filter associated with $\subsetnode$ vanish, namely:
	\begin{align*}
		\Wiener[\subsetnode]{\subsetnodealt}{(\subsetnode,\subsetnodesep)}=0.
	\end{align*}
	Analogously, we use the notation
    $cwsep(\subsetnode,\subsetnodesep,\subsetnodealt)$
    to denote that
		$\CWiener[\subsetnode]{\subsetnodealt}{(\subsetnode,\subsetnodesep)}(\ZTvar)=0$.
\end{dfn}

\section{Linear Dynamic Influence Models and their graphical representations}\label{sec:LDMI}
In the last decade several network models have emerged, such as Dynamical Structural Function (DSF) models \cite{GonWar08}, Linear Dynamic Graphs (LDGs) \cite{MatInn09b,MatSal12} and the dynamic interconnection proposed in \cite{dankers2012dynamic}.
These models share very similar characteristics, the main one being that they can be interpreted via graphs in a way analogous to Signal Flow Graphs \cite{shannon1942theory}.
While in their original formulations these models were not exactly equivalent, subsequent generalizations have made more and more similar (i.e. DSF was originally considering only strictly causal transfer functions, while LDGs was considering non-strictly causal transfer functions from the beginning).
In this article, we consider a specific class of dynamic networks defined by input-output relations among wide-sense stationary stochastic processes following the Dynamical Structural Function (DSF) approach proposed in \cite{GonWar08}, even though these results can be seamlessly applied to the network models in \cite{MatSal12} or \cite{dankers2012dynamic}.

We remind that DSF is an input/output representation of a network in the form
\begin{align*}
    \node=\DSFOutputMatrix(\ZTvar)\node+\DSFInputMatrix(\ZTvar)\inputsignal
\end{align*}
where $\node$ is a vector of outputs associated with the nodes of the network and $\inputsignal$ is the vector of inputs forcing it. The pair $(\DSFOutputMatrix(\ZTvar),\DSFInputMatrix(\ZTvar))$ defines the input-output relation in the model:
the matrix $\DSFOutputMatrix$ represents how the nodes are interconnected and affect each other, while the matrix 
$\DSFInputMatrix$ represents how the inputs are influencing the nodes. In DSF representation the number of forcing inputs and the number of output nodes can be different.
Also, DSF considers in general signals which can be either stochastic and deterministic.

In this article we specifically consider scenarios where the forcing signals in $\inputsignal$ are not observable and can be modeled as wide-sense-stationary stochastic processes.
The forcing inputs are assumed as mutually independent and act separately on each node of the network leading to a block diagonal structure for $\DSFInputMatrix(\ZTvar)$, which is assumed stable.
In the context of this article, since the forcing inputs are independent and not observable, there is no loss of generality, by defining $\noise=\DSFInputMatrix(\ZTvar)\inputsignal$,
and consider $\noise$ as the vector of inputs directly exciting the network.
In other words, the dynamics in $\DSFInputMatrix(z)$ can be incorporated in the power spectral density of
$\noise$,
$\PSD_{\noise}(\ZTvar)=\DSFInputMatrix(\ZTvar)\PSD_{\inputsignal}(\ZTvar)\DSFInputMatrix^{*}(\ZTvar)$ which has the same diagonal structure of $\PSD_{\inputsignal}(\ZTvar)$.
This motivates the following definition.

\begin{dfn}[Linear Dynamic Influence Models]\label{dfn:LDIM}
    Consider a network following the DSF dynamics
    $\node=\DSFOutputMatrix(\ZTvar)\node+\DSFInputMatrix(\ZTvar)\inputsignal$,
    as described in \cite{GonWar08} where
    \begin{itemize}
	\item 
        $\node
		=
		(\node_{\nodeindexstart}^{\transposesymbol}
			,...,
			\node_{\nodeindexend}^{\transposesymbol})^{\transposesymbol}$ are the $\nodenumber$ output nodes;
        \item
        $\inputsignal
		=
		(\inputsignal_{\noiseindexstart}^{\transposesymbol}
			,...,
			\inputsignal_{\noiseindexend}^{\transposesymbol})^{\transposesymbol}$
        are the $\nodenumber$ wide sense stationary inputs;
        \item   $\DSFOutputMatrix(\ZTvar)\in\mathcal{F}^{\nodenumber\times\nodenumber}$ and $\DSFInputMatrix(\ZTvar)\in (\mathcal{F^{+}})^{\nodenumber\times\nodenumber}$
        are stable $(\nodenumber\times\nodenumber)$-block matrices;
        \item	$\PSD_{\inputsignal}(\ZTvar)=\identitymatrix$ and
        $\DSFInputMatrix(\ZTvar)$ is block diagonal.
    \end{itemize}
    Define $\noise=\DSFInputMatrix(\ZTvar)\inputsignal$, so that $\PSD_{\noise}(\ZTvar)$ is block-diagonal and the dynamics can be written as $\node=\DSFOutputMatrix(\ZTvar)\node+\noise$.
    We say that the DSF network is a LDIM and we denote it with the pair $(\DSFOutputMatrix(\ZTvar),\noise)$.
    
	We say that the LDIM is causal if all entries of
	$\DSFOutputMatrix(\ZTvar)$
	are causal.
	We say that the LDIM is well-posed (causally well-posed) if
	$(\identitymatrix-\DSFOutputMatrix_{\subsetnodeindex\subsetnodeindex}(\ZTvar))$
	is invertible (causally invertible) for all $\subsetnodeindex\subseteq \{1,...,\nodenumber\}$ with no poles on the unit circle.
	The well-posedness (causal well-posedness) of a LDIM guarantees that
	$\node$ is wide-sense stationary.
\end{dfn}
Observe that the concept of causality is often a requirement for standard well-posedness (see \cite{willems1971analysis}).
Here, the dynamic matrix $\DSFOutputMatrix(\ZTvar)$ determining a LDIM is assumed not necessarily causal (unless the LDIM is causal): the motivation behind this choice is to allow the general application of the article results, potentially, to other domains where transfer functions might not necessarily be causal (i.e. video or image processing). This is the reason behind the distinct notions of well-posedness and causal well-posedness in Definition~\ref{dfn:LDIM}.
For a more extensive discussion on the well-posedness of DSF structure we refer to \cite{woodbury2018dynamic}

\begin{dfn}[Graphical representations of a LDIM and its associated perfect directed  graph (\pdg)]
	Given a LDIM,
	$(\adjTF(\ZTvar),\noise),$
	with output processes
	$\node
	=
	(\node_{\nodeindexstart}^{T}|...|\node_{\nodeindexend}^{T})^{T}$
	and a directed graph
	$\graph=(\nodeset, \directededgeset)$
	where
	$\nodeset
	=
	\{\node_{\nodeindexstart}, ..., \node_{\nodeindexend}\}$
	we say that 
	$\graph$
	is a graphical representation of the LDIM
	if
	$(\node_{\nodeindex}, \node_{\nodeindexalt})\notin\directededgeset$
	implies
	$\adjTF_{\nodeindexalt\nodeindex}(\ZTvar)= 0$.
	If it also holds that
	$\adjTF_{\nodeindexalt\nodeindex}(\ZTvar)= 0$
	implies
	$(\node_{\nodeindex}, \node_{\nodeindexalt})\notin\directededgeset$,
	then we say that
	$\graph$
	is the (unique) Perfect Directed Graph (\pdg) associated  with the LDIM.
\end{dfn}
Abusing the nomenclature we will sometimes refer to nodes, edges, paths and chains of a LDIM even though, formally, we should refer to them as nodes, edges, paths and chains of its graphical representation or its \pdg.

\section{Problem formulations: irrelevance of measurements and identification of a transfer function under graph constraints}\label{sec:problem}
Given a LDIM, we consider two basic problems.
We refer to the first problem as irrelevance of a measurement.
\begin{prob}[Irrelevance of a node in an estimate]\label{prob:irrelevance}
	Given a LDIM with graphical representation $G$, determine if the least square estimate of a node $\node_{\nodeindexalt}$ from the signals $\node_{S}$ and $\node_{\nodeindex}$ is equivalent to the estimate of $\node_{\nodeindexalt}$ from $y_{S}$ only.
\end{prob}

In other terms, Problem~\ref{prob:irrelevance} boils down to determining that the signal $\node_{\nodeindex}$ is irrelevant for the estimate of $\node_{\nodeindexalt}$ when $\node_{S}$ is available.
Observe that Problem~\ref{prob:irrelevance} is formulated assuming only knowledge of the graphical representation, thus, interestingly, the solution is going to be based purely on graphical conditions and holds for all LDIMs with the same structure.

The second basic problem is about the identification of a transfer function under graph constraints via non-invasive observations.
\begin{prob}[Module identification from partial knowledge of the structure and partial observation of the outputs]
	Consider a LDIM $\ldim$ with nodes $y_{1},...,y_{n}$ and forcing signals $\noise_{1},...,\noise_{n}$. Assume that
	\begin{itemize}
		\item	the dynamics is given by $\node=\DSFOutputMatrix(\ZTvar)\node+\noise$ with unknown $\DSFOutputMatrix(\ZTvar)$;
		\item	the graph $\graph$ is a known graphical representation for $\ldim$;
		\item	only a prespecified subset of the signals
		$\node_{1},...,\node_{\nodenumber}$ is observable (namely a minor of the PSD $\Phi_{\node}(\ZTvar)$ is known)
		\item	the forcing signals $\noise_{1},...,\noise_{n}$ are not observable.
	\end{itemize}
	For fixed
	$\nodeindex,\nodeindexalt\in \{1,...,\nodenumber\}$, $\nodeindex\neq \nodeindexalt$,
	determine the transfer function $\DSFOutputMatrix_{\nodeindexalt\nodeindex}(\ZTvar)$.
\end{prob}
A related problem is described in \cite{hayden2017network}, where no a priori knowledge of the network is provided, but the transfer marix $\DSFOutputMatrix(\ZTvar)$ is assumed strictly causal and all the node signals are observed (or equivalently, the full matrix $\PSD_{\node}(\ZTvar)$ is known).
In \cite{hayden2017network} the problem is tackled by showing that only a finite number of factorizations of $\PSD_{\node}(\ZTvar)=G(\ZTvar)G^{*}(\ZTvar)$ are compatible with the assumptions made on the network dynamics. Furthermore, modulo a multiplication by a signed identity matrix, the factor $G(\ZTvar)$ is unique in the case of a minimum phase $(I-\DSFOutputMatrix(\ZTvar))$.
After computing such a spectral factor, the results in \cite{GonWar08} are applied to recover the network.

In our scenario, the spectral factorizations compatible with the data (partial measurement of $\node$) and the network dynamics (non-necessarily strictly causal $\DSFOutputMatrix(\ZTvar)$) are in general uncountably many. It is the partial a priori knowledge of the network structure that will enable the recovery of the target transfer function $\DSFOutputMatrix_{\nodeindexalt\nodeindex}(\ZTvar)$.
Importantly, we consider situations where only a subset of the output nodes in $\node$ is being measured.
This will play a fundamental role to suitably apply our results to situations where the inputs are non-target specific \cite{chetty2017applying}, for example when dealing with confounding variables, by assuming that an input is an unmeasured node.

\section{Graphical manipulations}\label{sec:manipulations}
Different techniques can be used to manipulate LDIMs to obtain a new network of dynamic systems.
For example, the mechanisms of immersion described in \cite{langston1998algorithmic} and used in \cite{DanVan16} can be applied to a LDIM, but in general they do not lead to the definition of another LDIM because of the potential correlation occurring between the unknown external signals.
Instead, we will introduce two operations to manipulate a LDIM and obtain another LDIM: marginalization and self-loop removal.

The marginalization of a node in a LDIM is the definition of another LDIM with the same output processes, but such that the incoming links in the marginalized node are removed.
\begin{lem}[Node marginalization lemma]\label{lem:node marginalization}
	Consider a LDIM
	$\ldim=(\adjTF(\ZTvar),\noise)$
	with output processes
	$(\node_{\nodeindexstart},...,\node_{\nodeindexend})$.
	Let
	$(\nodeset,\directededgeset)$
	be a graphical representation of the LDIM
	$\ldim$.
	Let
	$\subsetnodealt \subset \nodeset$
	and
	$\subsetnode=V\setminus \subsetnodealt$.
	Define two vector processes:
	\begin{align}
		\node'&:=
		\left(\begin{array}{c}
			\node_{\subsetnodeindexalt}\\
			\noise_{\subsetnodeindex}
		\end{array}\right)
		\quad \text{and} \quad
		\noise':=
		\left(\begin{array}{c}
			\noise_{\subsetnodeindexalt}\\
			\noise_{\subsetnodeindex}
		\end{array}\right),
	\end{align}
	and the transfer matrices,
	\begin{align}	
		\adjTF'_{\subsetnodeindexalt\subsetnodeindexalt}
			&=\adjTF_{\subsetnodeindexalt\subsetnodeindexalt}+ \adjTF_{\subsetnodeindexalt\subsetnodeindex}(\identitymatrix-\adjTF_{\subsetnodeindex\subsetnodeindex})^{-1}\adjTF_{\subsetnodeindex\subsetnodeindexalt},\\
		\adjTF'_{\subsetnodeindexalt\subsetnodeindex}&=\adjTF_{\subsetnodeindexalt\subsetnodeindex}(\identitymatrix-\adjTF_{\subsetnodeindex\subsetnodeindex})^{-1},\\
		\adjTF'_{\subsetnodeindex\subsetnodeindex}&=0,\\
		\adjTF'_{\subsetnodeindex\subsetnodeindexalt}&=0.
	\end{align}
	The following properties hold:
	\begin{enumerate}
		\item\label{item:newLDIM}
			The vector processes $\node'$ and $\noise'$ define LDIM (``the marginalization of $\ldim$ with respect to $\node_{\subsetnodeindex}$'') which satisfies the input-output relation:
			\begin{align}\label{eq:ldgprime}
				\node'=\noise'+\adjTF'(\ZTvar)\node'.
			\end{align}
		\item\label{item:Hji'=0}
			Given two nodes
			$\node_{\nodeindex}\in\subsetnode$
			and
			$\node_{\nodeindexalt}$ in $\subsetnodealt$,
			if there is no chain of the form
			$\node_{\nodeindex}\rightarrow \node_{\path_{1}}\rightarrow \ldots\rightarrow \node_{\path_{\ell}}\rightarrow \node_{\nodeindexalt}$,
			with
			$\pathindexend\geq 1$,
			where
			$\{\node_{\path_{1}}, \ldots, \node_{\path_{\ell}}\}\subseteq \subsetnode$,
			then
			${\adjTF}_{{\nodeindexalt}{\nodeindex}}=0$
			implies
			${\adjTF}_{{\nodeindexalt}{\nodeindex}}'=0$.
		\item\label{item:tfUnchanged}
			Given two nodes
			$\node_{j_{1}}$
			and
			$\node_{j_{2}}\in\subsetnodealt$,
			if there is no chain of the form
			$\node_{j_{1}}\rightarrow \node_{\path_{1}}\rightarrow \ldots\rightarrow \node_{\path_{\ell}}\rightarrow \node_{j_{2}}$,
			with $\pathindexend\geq 1$,
			where
			$\{\node_{\path_{1}}, \ldots, \node_{\path_{\ell}}\}\subseteq \subsetnode$,
			then
			${\adjTF}_{{j_{2}}{j_{1}}}'
			={\adjTF}_{{j_{2}}{j_{1}}}$.
	\end{enumerate}
	Furthermore if $(\adjTF(\ZTvar),\noise)$ is causal and causally well-posed, then $(\adjTF'(\ZTvar),\noise')$ is causal and causally well-posed.
\end{lem}
\begin{proof}
	Note that 
	\begin{align*}
		\subsetnode
			&=\subsetnoise+\adjTF_{\subsetnodeindex\subsetnodeindexalt}\subsetnodealt
				+\adjTF_{\subsetnodeindex\subsetnodeindex}\subsetnode\\
			&=(\identitymatrix-\adjTF_{\subsetnodeindex\subsetnodeindex})^{-1}\subsetnoise
				+(\identitymatrix-\adjTF_{\subsetnodeindex\subsetnodeindex})^{-1}\adjTF_{\subsetnodeindex\subsetnodeindexalt}\subsetnodealt.
	\end{align*}
	Thus, substituting
	$\subsetnode$
	in
	$\subsetnodealt=\adjTF_{\subsetnodeindexalt\subsetnodeindexalt}\subsetnodealt
		+\adjTF_{\subsetnodeindexalt\subsetnodeindex}\subsetnode+\subsetnoisealt$
	it follows that, 
	\begin{align}\label{eq:yJinNewLDIM}
		\subsetnodealt'&:=
		\subsetnodealt
			=\subsetnoisealt
				+\adjTF_{\subsetnodeindexalt\subsetnodeindex}(\identitymatrix-\adjTF_{\subsetnodeindex\subsetnodeindex})^{-1}\subsetnoise+\\
				&\qquad\quad+[\adjTF_{\subsetnodeindexalt\subsetnodeindexalt}+\adjTF_{\subsetnodeindexalt\subsetnodeindex}(\identitymatrix-\adjTF_{\subsetnodeindex\subsetnodeindex})^{-1}\adjTF_{\subsetnodeindex\subsetnodeindexalt}]\subsetnodealt\\
			&=\subsetnoisealt
				+\adjTF'_{\subsetnodeindexalt\subsetnodeindex}\subsetnode'
				+\adjTF'_{\subsetnodeindexalt\subsetnodeindexalt}\subsetnodealt'.
	\end{align}
	Moreover,
	\begin{align*}
		\node_{\subsetnodeindex}'&:=\noise_{\subsetnodeindex}'
		=\noise_{\subsetnodeindex}+\adjTF'_{\subsetnodeindex\subsetnodeindex}\subsetnode'
			+\adjTF'_{\subsetnodeindex\subsetnodeindexalt}\subsetnodealt'.
	\end{align*}
	Thus \ref{item:newLDIM}) is proven.
	Now, let
	$T(\ZTvar):=(I-\adjTF_{\subsetnodeindex\subsetnodeindex}(\ZTvar))$
	where
	$T(\ZTvar)$
	is an invertible
	$N\times N$
	matrix.
	From Cayley-Hamilton theorem \cite{HorJoh12} it follows that there exist coefficients
	$a_{0}(\ZTvar),\ldots , a_{N-1}(\ZTvar) $
	such that:
	\begin{align*}
		T^N+a_{N-1}(\ZTvar)T^{N-1}+\ldots+(-1)^N det(T)I=0,
	\end{align*}
	where $a_0(\ZTvar):=(-1)^{N-1} det(T)\not= 0$.
	Multiplying the above identity by $T^{-1}$ it follows that,
	\begin{align*}
		T^{-1}=\sum_{k=0}^{N-1}c_k T^k,
	\end{align*}
	for some coefficients
	$c_{k}$
	which in turn admits an expansion of the form
	$T^{-1}=\sum_{k=0}^{N-1}d_k \adjTF_{\subsetnodeindex\subsetnodeindex}^k$.
	Thus, there exist coefficients
	$d_k$
	such that
	$(I-\adjTF_{\subsetnodeindex\subsetnodeindex})^{-1}=\sum_{k=0}^{N-1}d_k \adjTF_{\subsetnodeindex\subsetnodeindex}^k$.
	Consider
	$\node_{\nodeindex}\in\subsetnodeindex$
	and
	$\node_{\nodeindexalt}\in\subsetnodeindexalt$.
	Note that the entry
	$(\path_{1},\nodeindex)$
	of the matrix
	$\adjTF_{\subsetnodeindex\subsetnodeindex}$
	is $0$ if there is no link
	from $\node_{\nodeindex}$ to $\node_{\path_{1}}$.
	Also, note that the entry
	$(\path_{2},\nodeindex)$
	of the matrix $(\adjTF_{\subsetnodeindex\subsetnodeindex})^2$ can be written as:
	\begin{align*}
		[(\adjTF_{\subsetnodeindex\subsetnodeindex})^2]_{\path_{2}\nodeindex}=\sum_{\path_{1}\in\subsetnodeindex} [\adjTF_{\subsetnodeindex\subsetnodeindex}]_{\path_{2}\path_{1}}[\adjTF_{\subsetnodeindex\subsetnodeindex}]_{\path_{1}\nodeindex},
	\end{align*}
	and thus, if there is no chain of the form
	$\node_{\nodeindex}\to\node_{\path_{1}}\rightarrow \node_{\path_{2}}$
	for any
	$\node_{\path_{1}}\in\subsetnode$
	then
	$[\adjTF_{\subsetnodeindex\subsetnodeindex}]_{\path_{2}\path_{1}}[\adjTF_{\subsetnodeindex\subsetnodeindex}]_{\path_{1}\nodeindex}=0$
	for all
	$\path_{1}\in\subsetnodeindex$ implying
	$[(\adjTF_{\subsetnodeindex\subsetnodeindex})^2]_{\path_{2}\nodeindex}=0$.
	Similarly the entry 
	$(\path_{3},\nodeindex)$
	of the matrix
	$(\adjTF_{\subsetnodeindex\subsetnodeindex})^3$
	can be written as
	\begin{align*}
		[(\adjTF_{\subsetnodeindex\subsetnodeindex})^3]_{\path_{3}\nodeindex}=
		\sum_{\path_{2}\in\subsetnodeindex}\sum_{\path_{1}\in\subsetnodeindex}
			[\adjTF_{\subsetnodeindex\subsetnodeindex}]_{\path_{3}\path_{2}}
			[\adjTF_{\subsetnodeindex\subsetnodeindex}]_{\path_{2}\path_{1}}
			[\adjTF_{\subsetnodeindex\subsetnodeindex}]_{\path_{1}\nodeindex},
	\end{align*}
	and thus, if there is no chain of the form
	$\node_{\nodeindex}\to\node_{\path_{1}}\rightarrow \node_{\path_{2}}\rightarrow\node_{\path_{3}}$
	with
	$\{\node_{\path_{1}},\node_{\path_{2}},\node_{\path_{3}}\}\subseteq\subsetnode$
	then we have 
	$[\adjTF_{\subsetnodeindex\subsetnodeindex}]_{\path_{3}\path_{2}}[\adjTF_{\subsetnodeindex\subsetnodeindex}]_{\path_{2}\path_{1}}[\adjTF_{\subsetnodeindex\subsetnodeindex}]_{\path_{1}\nodeindex}=0$
	for all
	$\path_{1},\path_{2}\in\subsetnodeindex$ implying
	$[(\adjTF_{\subsetnodeindex\subsetnodeindex})^2]_{\path_{3}\nodeindex}=0$.
	Iterating this argument it can be shown that if there is no chain of the form
	$\node_{\nodeindex}\to\node_{\path_{1}}\rightarrow ...\rightarrow \node_{\path_{\ell}}$
	with
	$\{\node_{\path_{1}},...,\node_{\path_{\pathlength-1}}\}\subseteq\subsetnode$
	for
	$\ell\geq 1$
	then
	$[(\adjTF_{\subsetnodeindex\subsetnodeindex})^\ell]_{\path_{\ell}\nodeindex}=0$.
	Suppose there is no chain of the form
	$\node_{\nodeindex}\rightarrow \node_{\path_{1}}\rightarrow \ldots\rightarrow \node_{\path_{\ell}}\rightarrow \node_{\nodeindexalt}$
	with
	$\{\node_{\path_{1}},...,\node_{\path_{\pathlength}}\}\subseteq\subsetnode$.
	Then
	\begin{align*}
		&[\adjTF'_{\subsetnodeindexalt\subsetnodeindex}]_{\nodeindexalt\nodeindex}
		=[\adjTF_{\subsetnodeindexalt\subsetnodeindex}(I-\adjTF_{\subsetnodeindex\subsetnodeindex})^{-1}]_{\nodeindexalt\nodeindex}
		=\sum_{k\in\subsetnodeindex}
			\adjTF_{\nodeindexalt k}[(I-\adjTF_{\subsetnodeindex\subsetnodeindex})^{-1}]_{k\nodeindex}\\
		&\quad=\sum_{k\in\subsetnodeindex}\sum_{m=0}^{N-1}
			\adjTF_{\nodeindexalt k}d_{m}[\adjTF_{\subsetnodeindex\subsetnodeindex}^{m}]_{k\nodeindex}
		=\sum_{k\in\subsetnodeindex}
			\adjTF_{\nodeindexalt k}d_{0}[\adjTF_{\subsetnodeindex\subsetnodeindex}^{0}]_{k\nodeindex}
		=\adjTF_{\nodeindexalt\nodeindex}d_{0},
	\end{align*}
	where the last two equalities follow respectively from the fact that
	there is no chain of the form
	$\node_{\nodeindex}\rightarrow \node_{\path_{1}}\rightarrow \ldots\rightarrow \node_{\path_{\ell}}\rightarrow \node_{\nodeindexalt}$
	with
	$\{\node_{\path_{1}},...,\node_{\path_{\pathlength}}\}\subseteq\subsetnode$
	and the fact that
	$[\adjTF_{\subsetnodeindex\subsetnodeindex}^{0}]_{k\nodeindex}=\identitymatrix_{k\nodeindex}=0$
	for $k\neq\nodeindex$. This proves~\ref{item:Hji'=0}).
	
	For
	$\node_{j_{1}}\in\node_{\subsetnodeindexalt}$
	and
	$\node_{j_{2}}\in\node_{\subsetnodeindexalt}$,
	following a similar argument,
	we show that
	\small
	\begin{align*}
		[\adjTF'_{\subsetnodeindexalt\subsetnodeindexalt}]_{j_{2}j_{1}}&=
		[\adjTF_{\subsetnodeindexalt\subsetnodeindex}(I-\adjTF_{\subsetnodeindex\subsetnodeindex})^{-1}\adjTF_{\subsetnodeindex\subsetnodeindexalt}
	+\adjTF_{\subsetnodeindexalt\subsetnodeindexalt}]_{j_{2}j_{1}}=\\
	&=\adjTF_{j_{2}j_{1}}
	+\sum_{k\in\subsetnodeindex}\sum_{h\in\subsetnodeindex}
		[\adjTF_{\subsetnodeindexalt\subsetnodeindex}]_{j_{2}k}
		[(I-\adjTF_{\subsetnodeindex\subsetnodeindex})^{-1}]_{kh}
		[\adjTF_{\subsetnodeindex\subsetnodeindexalt}]_{hj_{1}}\\
	&=\adjTF_{j_{2}j_{1}}
	+\sum_{k\in\subsetnodeindex}\sum_{h\in\subsetnodeindex}\sum_{m=0}^{N-1}
		d_{m}
		[\adjTF_{\subsetnodeindexalt\subsetnodeindex}]_{j_{2}k}
		[\adjTF_{\subsetnodeindex\subsetnodeindex}^{m}]_{kh}
		[\adjTF_{\subsetnodeindex\subsetnodeindexalt}]_{hj_{1}}\\
	&=\adjTF_{j_{2}j_{1}}+
	\sum_{k\in\subsetnodeindex}\sum_{h\in\subsetnodeindex}
		d_{0}
		[\adjTF_{\subsetnodeindexalt\subsetnodeindex}]_{j_{2}k}
		\identitymatrix_{kh}
		[\adjTF_{\subsetnodeindex\subsetnodeindexalt}]_{hj_{1}}\\
	&=\adjTF_{j_{2}j_{1}}+
	\sum_{k\in\subsetnodeindex}
		[\adjTF_{\subsetnodeindexalt\subsetnodeindex}]_{j_{2}k}
		[\adjTF_{\subsetnodeindex\subsetnodeindexalt}]_{kj_{1}}
	=
	\adjTF_{j_{2}j_{1}}
	\end{align*}
	\normalsize
	where the last two equalities follow respectively from the fact that
	there is no chain of the form
	$\node_{\nodeindex}\rightarrow \node_{\path_{1}}\rightarrow \ldots\rightarrow \node_{\path_{\ell}}\rightarrow \node_{\nodeindexalt}$
	with
	$\{\node_{\path_{1}},...,\node_{\path_{\pathlength}}\}\subseteq\subsetnode$
	and the fact that
	$[\adjTF_{\subsetnodeindex\subsetnodeindex}^{0}]_{kh}=\identitymatrix_{kh}=0$
	for $k\neq h$.	
\end{proof}
In Figure~\ref{fig:marginalization} we report the graphical representation of a LDIM resulting from the marginalization of one of its nodes.
\begin{figure}
    \centering
    \begin{tabular}{cc}
        \includegraphics[width=0.4\columnwidth]{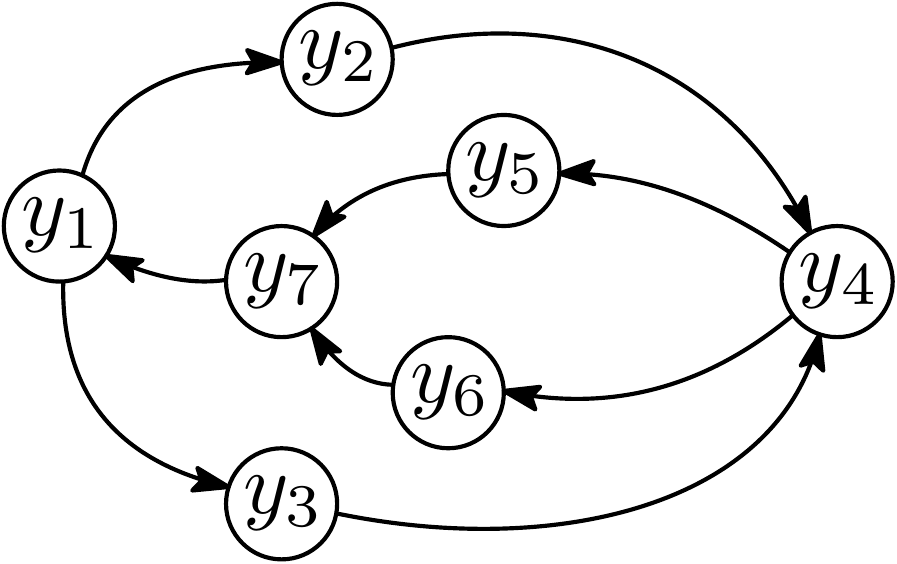} & 
        \includegraphics[width=0.4\columnwidth]{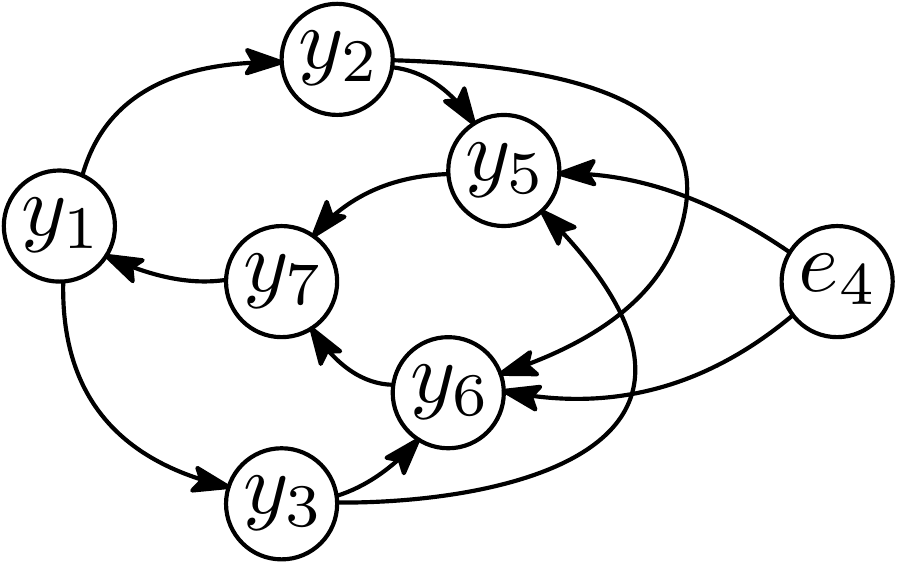}\\
        (a) & (b)
    \end{tabular}
    \caption{\label{fig:marginalization} Graphical representation of a LDIM (a) and the graphical representation resulting from the marginalization of $\node_{4}$. For simplicity it is assumed that in this example the procedure has not generated any self-loops.}
\end{figure}

The marginalization procedure can introduce non-zero elements in the diagonal blocks of the transfer matrix $\adjTF'(\ZTvar)$,  even if  the diagonal blocks of $\adjTF$ are zero.
These non-zero elements in the diagonal blocks of $\adjTF'(\ZTvar)$, however, can be removed using the lemma described below.
\begin{lem}[Self-loop removal]\label{lem:selfloopremoval}
	Consider a set of processes related to each other via:
	\begin{align*}
		\node=\noise+\adjTF(\ZTvar)\node
	\end{align*}
	where
	$\node=(\node_{1}^{T}|...|\node_{\nodeindexend}^{T})$, 
	$\noise=(\noise_{1}^{T}|...|\noise_{\noiseindexend}^{T})$. 
	Consider, $\adjTF$ to have a $n\times n$ block structure commensurate with dimensions of $\node_i,\ i=1,\ldots,n.$
	Assume that  block transfer matrix
	$\adjTF_{kk}(\ZTvar)\neq 0$
	for all
	$k\in K\subset \{1,2,\ldots,n\}$,
	and
	$\adjTF_{kk}(\ZTvar)= 0$
	for all
	$k\notin K$.
	Define the block diagonal transfer matrix
	$D(\ZTvar)$
	such that
	$D_{kk}(\ZTvar)=(\identitymatrix-\adjTF_{kk})^{-1}$ for all $k$.
	A new set of processes $\node'$ following the relation
	\begin{align*}
		\node'=\noise'+\adjTF'(\ZTvar)\node,
	\end{align*}
	can be defined with the following properties:
	\begin{enumerate}
		\item	\label{prop:nochangey} $\node'=\node$.
		\item	\label{prop:newnoise}	$\noise'=D(\ZTvar)\noise$ (thus $e'_{j}=e_{j}$ if $j\notin K$ and $\PSD_{\noise'}=D(\ZTvar)\PSD_{\noise}D^{*}(\ZTvar)$).
		\item	\label{prop:noselfloops}	$\adjTF'$,  imposed with the same block structure of $\adjTF$, has zero entries in its diagonal blocks.
		\item	\label{prop:newoffdiagonal}	$\adjTF'_{ki}=D_{kk}\adjTF_{ki}$ for $i\neq k$ (thus $\adjTF_{ki}'\neq 0$ if $\adjTF_{ki}\neq 0$ for $i\neq j$ and $\adjTF_{ji}'=\adjTF_{ji}$ if $j\notin K$ and  for all $i$).	
	\end{enumerate}
	Furthermore if $(\adjTF(\ZTvar),\noise)$ is a LDIM, then $(\adjTF'(\ZTvar),\noise')$ is a LDIM and if 
	$(\adjTF(\ZTvar),\noise)$ is a causal and causally well-posed LDIM, then $(\adjTF'(\ZTvar),\noise')$ is a causal and causally well-posed LDIM.
\end{lem}
\begin{proof}
	For every $k\in K$ observe that:
	\begin{align*}
		\node_{k}=\noise_{k}+\adjTF_{kk}\node_{k}+\sum_{\nodeindex \neq k}\adjTF_{k\nodeindex}\node_{\nodeindex}.
	\end{align*}
	This implies that
	\begin{align*}
		\node_{\nodeindexaltb}=(\identitymatrix-\adjTF_{kk})^{-1}\noise_{k}
			+\sum_{\nodeindex \neq k}(\identitymatrix-\adjTF_{kk})^{-1}\adjTF_{k\nodeindex}\node_{\nodeindex}.
	\end{align*}
	Define $\node':=\node$. Also, for $k\in K$, define 
	\begin{align*}
		\noise_{\nodeindexaltb}'&:=(\identitymatrix-\adjTF_{kk})^{-1}\noise_{k}\\
		\adjTF_{\nodeindexaltb\nodeindex}'&:=(\identitymatrix-\adjTF_{kk})^{-1}\adjTF_{k\nodeindex} \quad \text{for all } \nodeindex\neq k\\
		\adjTF_{\nodeindexaltb\nodeindexaltb}'&:=0,
	\end{align*}
	and for $k\notin K$ define
	\begin{align*}
		\noise_{k}'&:=\noise_{k},\\
		\adjTF_{k\nodeindex}'&=\adjTF_{k\nodeindex} \quad \text{for all } \nodeindex.
	\end{align*}
	All the properties in the theorem statement can be verified by inspection.
\end{proof}
\begin{rem}
	Note that in Lemma~\ref{lem:selfloopremoval} we have that
	\begin{align*}
		\adjTF'=D\adjTF-diag(D_{11}\adjTF_{11},...,D_{nn}\adjTF_{nn}),
	\end{align*}
	where $n$ is the block dimension of the transfer matrix $\adjTF$.
\end{rem}
The self-loop removal lemma can be used to remove all self-loops after marginalizing a LDIM  in order to obtain a new LDIM with no self-loops.
\begin{rem}\label{rem:newgraphproperties}
	The marginalization of a LDIM, with the subsequent removal of all self-loops, produces a new LDIM with no self-loops and  with the property that nodes
	$\node_{\subsetnodeindex}$
	have no incoming links and have outgoing links only to the nodes in
	$\child{\graph}{\nodeset\setminus\noise_{\subsetnodeindex}}$.
\end{rem}

\section{Relation between $d$-separation and Wiener separation in LDIMs}\label{sec:dseparation}
The main goal of this section is to show that if the graph
$\graph$
is a graphical representation of the LDIM
$\ldim=(\adjTF(\ZTvar),\noise)$
then
$\dsep{\graph}(\subsetnode,\subsetnodesep,\subsetnodealt)$
implies 
$\wsep(\subsetnode,\subsetnodesep,\subsetnodealt)$.
Furthermore, if the LDIM is causal, we also have that 
$\dsep{\graph}(\subsetnode,\subsetnodesep,\subsetnodealt)$
implies 
$\cwsep(\subsetnode,\subsetnodesep,\subsetnodealt)$
for the LDIM.
Preliminary results, but only limited to the non-causal scenario, were shown in \cite{MatSal14} following a derivation similar to the one provided in \cite{Kos99} for graphical models of gaussian random variables.
Such a derivation  cannot be repeated for the causal scenario, since it would require an analytical expression for a matrix spectral factorization. Here, we provide a more general derivation that can tackle both scenarios at once.
Such a novel derivation still borrows three lemmas from \cite{MatSal14} and \cite{Kos99} which are of general utility since they can be used to obtain simpler $d$-separation relations in graphs.

When three disjoints sets 
$\subsetvertex$,
$\subsetvertexalt$,
and
$\subsetvertexsep$
are a partition of the node set $V$ of a graph, 
the following lemma provides a characterization of $d$-separation.
\begin{lem}\label{lem:linkTypeRestricted}\label{lem:separation under 3 conditions}
	Consider a directed graph $\graph=(\vertexset,\directededgeset)$.
	Suppose
	$\subsetvertex$,
	$\subsetvertexalt$
	and
	$\subsetvertexsep$
	are three disjoint subsets of
	$\vertexset$
	such that
	$\vertexset=\subsetvertex\cup\subsetvertexalt\cup\subsetvertexsep$;
	thus
	$\subsetvertex$,
	$\subsetvertexalt$
	and
	$\subsetvertexsep$
	partition the set 
	$\vertexset$.
	Then
	$\subsetvertexsep$
	$d$-separates
	$\subsetvertex$
	and
	$\subsetvertexalt$
	(i.e. $\dsep{\graph}(\subsetvertex,\subsetvertexsep,\subsetvertexalt)$)
	if and only if there are no paths of the form:
	\begin{enumerate}[label=\alph*)]
		\item\label{path a} $\vertex_{\vertexindex}\rightarrow \vertex_{\vertexindexalt},$
		\item\label{path b} $\vertex_{\vertexindexalt}\rightarrow \vertex_{\vertexindex}$,
		\item\label{path c} $\vertex_{\vertexindex}\rightarrow \vertex_{\vertexindexsep}\leftarrow \vertex_{\vertexindexalt}.$
	\end{enumerate}
	where
	$\vertex_{\vertexindex}\in\subsetvertex$,
	$\vertex_{\vertexindexalt}\in \subsetvertexalt$
	and 
	$\vertex_{\vertexindexsep}\in \subsetvertexsep.$
\end{lem}
\begin{proof}
    See Lemma~5 in \cite{MatSal14} for the proof.
\end{proof}

This lemma, instead, provides a way to verify $d$-separation in a graph $\graph$ using a graph $\graph'$ that is a restriction of $\graph$.
\begin{lem}\label{lem:separation on restriction}
	Let
	$\graph=(\vertexset,\directededgeset)$
	be a directed graph.
	Let
	$\subsetvertex$,
	$\subsetvertexalt$
	and
	$\subsetvertexsep$
	be disjoint subsets of
	$\vertexset$.
	Let
	$\vertexset'=an_{\graph}(\subsetvertex\cup\subsetvertexalt\cup\subsetvertexsep)$
	and let
	$\graph'=(\vertexset',\directededgeset\, ')$
	be the restriction of
	$\graph$
	to
	$\vertexset'$.
	Then
	\begin{align*}
		\dsep{\graph}(\subsetvertex,\subsetvertexsep,\subsetvertexalt)
		\Leftrightarrow
		\dsep{\graph'}(\subsetvertex,\subsetvertexsep,\subsetvertexalt).
	\end{align*}
\end{lem}
\begin{proof}
	This lemma is the same as Corollary~1 in \cite{Kos99}. We refer to \cite{Kos99} for its proof.
\end{proof}
The following lemma allows one to check
if
$\dsep{\graph}(\subsetvertex,\subsetvertexsep,\subsetvertexalt)$
holds by testing
$\dsep{\graph}(\subsetvertexprime,\subsetvertexsep,\subsetvertexaltprime)$
where
$\subsetvertexprime$, $\subsetvertexsep$, $\subsetvertexaltprime$
partition the set of all ancestors of 
$\subsetvertex$, $\subsetvertexsep$, $\subsetvertexalt$ and
$\subsetvertex\subseteq\subsetvertexprime$
and
$\subsetvertexalt\subseteq\subsetvertexaltprime$.
\begin{lem}\label{lem:inferDsepOnEnlargedSets}
	Given a directed graph
	$\graph=(\vertexset,\directededgeset)$
	and three disjoint sets
	$\subsetvertex,\subsetvertexalt,\subsetvertexsep \subset \vertexset$.
	Let
	$\vertexsetsmallertocheckdsep:=an(\subsetvertex,\subsetvertexalt,\subsetvertexsep)$.
	Let 
	$\subsetvertexaltprime$
	be the set of all nodes in
	$\vertexsetsmallertocheckdsep$
	that are $d$-separated from the set
	$\subsetvertex$
	by
	$\subsetvertexsep$
	\begin{align*}
		\subsetvertexaltprime:=
			\{\vertex_{\vertexindexalt}\in\nodesetsmallertocheckdsep|\dsep{\graph}(\subsetvertex,\subsetvertexsep,\vertex_{\vertexindexalt})\}.
	\end{align*}
	Let the set
	$\subsetvertexprime$
	be defined as complement of
	$\{\subsetvertexaltprime\cup \subsetvertexsep\}$
	in
	$\nodesetsmallertocheckdsep$:
	\begin{align*}
		\subsetvertexprime:=\nodesetsmallertocheckdsep\backslash \{\subsetvertexaltprime\cup \subsetvertexsep\}.
	\end{align*}
	Then 
	\begin{equation}\label{eq:dsepIZJimpliesI'ZJ'}
		\dsep{\graph}(\subsetvertex,\subsetvertexsep,\subsetvertexalt)
			\Leftrightarrow
		\dsep{\graph}(\subsetvertexprime,\subsetvertexsep,\subsetvertexaltprime).
	\end{equation}
\end{lem}
\begin{proof}
	This lemma is the same as Proposition~2 in \cite{Kos99}. We refer to \cite{Kos99} for its proof.
\end{proof}

In a graph
$\graph=(\vertexset,\directededgeset)$,
the Markov Blanket of a set 
$\subsetnode$
is the set given by the nodes in $\vertexset\setminus\subsetnode$ that are parents, children or parents of children of nodes in 
$\subsetnode$.
\begin{dfn}[Markov Blanket]
    Given a graph
    $\graph$
    and a set
    $\subsetnode$,
    define the Markov Blanket of $\subsetnode$ as 
    \begin{align*}
        \MarkovBlanket{\graph}{\subsetnode}:=
            \parent{\graph}{\subsetnode}
            \cup
            \child{\graph}{\subsetnode}
            \cup
            \parent{\graph}{\child{\graph}{\subsetnode}}\setminus \subsetnode.
    \end{align*}
\end{dfn}

\begin{thm}[Markov Blanket Separation]\label{thm:MarkovBlanketSeparation}
    Consider a LDIM
    $\ldim$
    with set of nodes $\vertexset$
    and no self-loop in
    $\subsetnode$.
    For every
    $\subsetnodesep\supseteq \MarkovBlanket{\graph}{\subsetnode}$,
    we have that
    $\wsep\left(\subsetnode,\subsetnodesep,\vertexset\setminus (\subsetnodesep \cup \subsetnode)\right)$.
    Furthemore, if 
    $\ldim$
    is causal and causally well-posed we also have
    $\cwsep\left(\subsetnode,\subsetnodesep,\vertexset\setminus (\subsetnodesep \cup \subsetnode)\right)$.
\end{thm}
\begin{proof}
    If $\subsetnode$ has a single element, this result is equivalent to Theorem~27 and Theorem~34 in \cite{MatSal12} for the non-causal and causal case respectively.
    In order to show Wiener separation for
    a vector process $\subsetnode$,
    it is enough to show Wiener separation for every element
    in
    $\subsetnode$.
    We show this via some manipulations.
    Define $\node_{\overline{I}}:=\subsetnode\setminus\{\node_{i}\}$ and marginalize $\node_{\overline{I}}$.
    Let $\graph_{i}'$ be the graphical representation of the graph after the marginalization of $\node_{\overline{I}}$. 
    We want to show that $\MarkovBlanket{\graph_{i}'}{\node_{i}}\subseteq \MarkovBlanket{\graph}{\subsetnode}$.
    Trivially all children of and parents of $\node_{i}$ in $\graph_{i}'$ are respectively children and parents of $\subsetnode$ of $\graph$.
    The marginalization of $\node_{\overline{I}}$, though can create some edges in $\graph_{i}'$ that are not in $\graph$.
    There is an edge from $\node_{k}$ to $\node_{\ell}$ in $\graph_{i}'$ that is not in $\graph$ if and only if there is a path in $\graph$ from $\node_{k}$ to $\node_{\ell}$ that involves only nodes in $\node_{\overline{I}}$. Thus $\node_{k}$ needs to be a parent of $\node_{I}$ in $\graph$ and $\node_{\ell}$ needs to be a child of $\node_{I}$.
    Necessarily, then, a coparent of $\node_{i}$ in $\graph_{i}'$ is a coparent of $\node_{i}$ in $\graph$ or a parent of $\node_{I}$ in $\graph$. This implies that 
    $\MarkovBlanket{\graph_{i}'}{\node_{i}}\subseteq \MarkovBlanket{\graph}{\subsetnode}$.
    Then, since $\MarkovBlanket{\graph_{i}'}{\node_{i}} \subseteq \subsetnodesep$ for all $\node_{i}\in\subsetnode$, we have,
    from Theorem~27 in \cite{MatSal12}, that
    $\wsep\left(\node_{i},\subsetnodesep,\vertexset\setminus (\subsetnodesep \cup \subsetnode)\right)$ for all $\node_{i}\in\subsetnode$.
    Indentical step prove the statement 
    $\cwsep\left(\node_{i},\subsetnodesep,\vertexset\setminus (\subsetnodesep \cup \subsetnode)\right)$
    in the causal scenario,
    where the fact that
    $\ldim$
    is a causal and causally well-posed guarantees that
    $\ldim_{\node}'$ is causal and causally well posed
    so that 
    Theorem~34 in \cite{MatSal12} can be applied.
\end{proof}

Lemmas~\ref{lem:separation under 3 conditions},
\ref{lem:separation on restriction},
and \ref{lem:inferDsepOnEnlargedSets},
along with Theorem~\ref{thm:MarkovBlanketSeparation}
allow us to introduce the first main contribution of this article.
The following theorem estabilishes a strong connection between the concept of $d$-separation defined on the graph and the notion of Wiener separation that is instead associated with LDIM.
\begin{thm}\label{thm:non-causal Wiener-separation}
	In a LDIM  with graphical representation $\graph=(\vertexset,\directededgeset)$
	let
	$\subsetnode,\subsetnodesep,\subsetnodealt$
	be disjoint subsets of
	$\vertexset$. Then
	\begin{align*}
	\dsep{\graph}(\subsetnode,\subsetnodesep,\subsetnodealt)
		\Rightarrow
	\wsep(\subsetnode,\subsetnodesep,\subsetnodealt).
	\end{align*}
\end{thm}
\begin{proof}
    Let
    $\subsetnode$,
    $\subsetnodesep$,
    and
    $\subsetnodealt$
    be three disjoint subsets such that
    $\dsep{\graph}(\subsetnode,\subsetnodesep,\subsetnodealt)$.
    Define
    $\vertexset'=\ancestor{\graph}{\subsetnode\cup\subsetnodesep\cup\subsetnodealt}$
    and observe that
    $\subsetnode\cup\subsetnodesep\cup\subsetnodealt\subseteq \vertexset'$.
    Since the transfer functions from
    $\node_{\vertexset\setminus\vertexset'}$
    to
    $\node_{\vertexset}$
    are identically zero, we have that
    $\node_{\vertexset'}=\adjTF_{\vertexset'\vertexset'}(\ZTvar)\node_{\vertexset'}+\noise_{\vertexset'}$
    allowing us to define a LDIM
    $\ldim'$
    over the nodes
    $\node_{\vertexset'}$.
    Observe that
    the restriction $\graph'$ of
    $\graph$
    to
    $\vertexset'=\ancestor{\graph}{\subsetnode\cup\subsetnodesep\cup\subsetnodealt}$
    is a graphical representation for $\ldim'$.
    From Lemma~\ref{lem:separation on restriction}, we have that
    $\dsep{\graph'}(\subsetnode,\subsetnodesep,\subsetnodealt)$.
    Defining
    $\subsetnodeprime\subseteq\subsetnode$
    and
    $\subsetnodealtprime\subseteq\subsetnodealt$
    as 
    in Lemma~\ref{lem:inferDsepOnEnlargedSets} we get 
    $\dsep{\graph'}(\subsetnodeprime,\subsetnodesep,\subsetnodealtprime)$.
    Observe that
    $\subsetnodeprime$,
    $\subsetnodealtprime$
    and
    $\subsetnodesep$
    are a partition of
    $\vertexset'$,
    thus, from
    Lemma~\ref{lem:separation under 3 conditions}
    we have that
    $\MarkovBlanket{\graph'}{\subsetnodeprime}\subseteq\subsetnodesep$.
    Also, we can write the dynamics of $\ldim'$ as
    \scriptsize
    \begin{align}
        \left(\begin{array}{c}
            \subsetnodeprime\\
            \subsetnodesep\\
            \subsetnodealtprime
        \end{array}\right)
        =
        \left(\begin{array}{ccc}
            \adjTF_{I'I'} & \adjTF_{I'Z} & 0\\
            \adjTF_{ZI'} & \adjTF_{ZZ} & \adjTF_{ZJ'}\\
            0 & \adjTF_{J'Z} & \adjTF_{J'J'}
        \end{array}\right)
        \left(\begin{array}{c}
            \subsetnodeprime\\
            \subsetnodesep\\
            \subsetnodealtprime
        \end{array}\right)
        +
        \left(\begin{array}{c}
            \noise_{I'}\\
            \noise_{Z}\\
            \noise_{J'}
        \end{array}\right).
    \end{align}
    \normalsize
    By applying the self-loop removal Lemma, this can be rewritten as 
    \scriptsize
    \begin{align}
        \left(\begin{array}{c}
            \subsetnodeprime\\
            \subsetnodesep\\
            \subsetnodealtprime
        \end{array}\right)
        =
        \left(\begin{array}{ccc}
            0 & \tilde\adjTF_{I'Z} & 0\\
            \tilde\adjTF_{ZI'} & 0 & \tilde\adjTF_{ZJ'}\\
            0 & \tilde\adjTF_{J'Z} & 0
        \end{array}\right)
        \left(\begin{array}{c}
            \subsetnodeprime\\
            \subsetnodesep\\
            \subsetnodealtprime
        \end{array}\right)
        +
        \left(\begin{array}{c}
            \noise_{I'}\\
            \noise_{Z}\\
            \noise_{J'}
        \end{array}\right),
    \end{align}
    \normalsize
    so that we obtain a LDIM
    $\ldim''$
    with the same output signals as in
    $\ldim'$.
    In the graphical representation $\graph''$ of $\ldim''$ we still have that
    $\MarkovBlanket{\graph''}{\subsetnodeprime}\subseteq\subsetnodesep$:
    (i) no edge connects directly a node in
    $\subsetnodeprime$
    and a node in
    $\subsetnodealtprime$;
    (ii)
    if a node in $\subsetnodeprime$ and a node in $\subsetnodealtprime$ shared a child in $\graph''$
    then there would be from Lemma~\ref{lem:selfloopremoval}
    a node in $\subsetnodeprime$ and a node in $\subsetnodealtprime$ sharing a child in $\graph'$
    and this is ruled out by the fact that
    $\dsep{\graph'}(\subsetnodeprime,\subsetnodesep,\subsetnodealtprime)$.
    From
    Theorem~\ref{thm:MarkovBlanketSeparation}
    applied to the LDIM $\ldim'$
    we have that
    $\wsep\left(\subsetnodeprime,\MarkovBlanket{\graph''}{\subsetnodeprime},\vertexset'\setminus(\MarkovBlanket{\graph''}{\subsetnodeprime}\cup \subsetnodeprime) \right)$.
    Thus, the Wiener filter estimating 
    $\subsetnodeprime$
    from
    $\vertexset'\setminus \subsetnodeprime$
    has zero components for all signals in
    $\vertexset'\setminus(\MarkovBlanket{\graph''}{\subsetnodeprime}\cup \subsetnodeprime)$.
    Since
    $\MarkovBlanket{\graph'}{\subsetnodeprime}\subseteq\subsetnodesep$,
    we have also that the Wiener filter estimating 
    $\subsetnodeprime$
    from
    $\vertexset'\setminus \subsetnodeprime$
    has zero components for all signals in
    $\vertexset'\setminus (\subsetnodesep \cup \subsetnodeprime)$
    implying
    $\wsep\left(\subsetnodeprime,\subsetnodesep,\vertexset'\setminus (\subsetnodesep \cup \subsetnodeprime) \right)$.
    Since 
    $\subsetnodealt\subseteq \subsetnodealtprime = \vertexset'\setminus (\subsetnodesep \cup \subsetnodeprime)$
    we also have 
    $\wsep\left(\subsetnodeprime,\subsetnodesep,\subsetnodealt \right)$.
    Finally, since 
    $\subsetnode\subseteq \subsetnodeprime$
    we also have that
    $\wsep\left(\subsetnode,\subsetnodesep,\subsetnodealt \right)$.
    The proof for the causal case follows the same steps, where the fact that the LDIM
    $\ldim$ is causal and causally well-posed guarantees via Lemma~\ref{lem:node marginalization} and Lemma~\ref{lem:selfloopremoval} that the LDIMs $\ldim'$ and $\ldim''$ are causal and causally well-posed.
\end{proof}

The importance of Theorem~\ref{thm:non-causal Wiener-separation} can be summarized as follows:
when estimating a certain node from the observations of a set of other nodes,
a graphical test ($d$-separation) allows one to determine if some of the observations are irrelevant for the estimate and thus can be disregarded.
Example~\ref{ex:irrelevant for estimate} illustrates the power of Theorem~\ref{thm:non-causal Wiener-separation} by showing how it is possible to answer questions on sensor locations without solving any optimization problem but relying entirely on the graphical properties of the network topology.
\begin{example}\label{ex:irrelevant for estimate}
	\normalsize
	Consider a LDIM with a graphical representation given by the graph
	$\graph$
	shown in Figure~\ref{fig:exampleIrrelevance}.
	\begin{figure}[b]
		\centering
		\includegraphics[width=0.75\columnwidth]{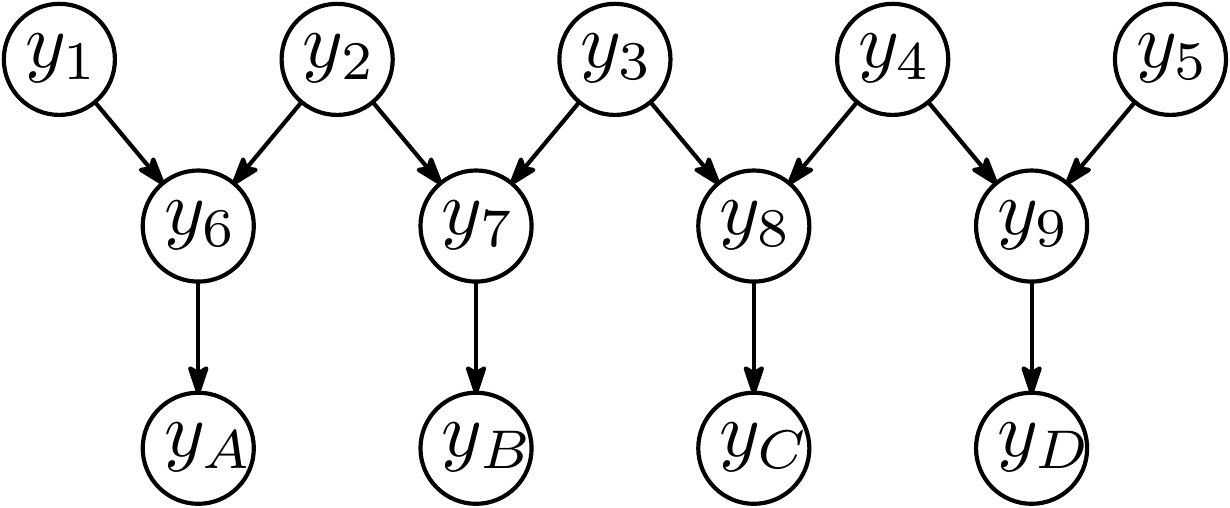}
		\caption{
			Graphical representation of the LDIM of Example~\ref{ex:irrelevant for estimate}.
			The objective is to minimize the least square estimation error on
			$\node_{1}$
			using only three of the four observed nodes
			$\node_{A}$, $\node_{B}$, $\node_{C}$, $\node_{D}$.
			By the iterative application of Theorem~\ref{thm:non-causal Wiener-separation}, it is shown that, irrespective of the power spectal densities of the signals, because of the system structure, the triplet $\node_{A}$, $\node_{B}$, $\node_{C}$ always minimizes the least square error in the estimate.
			\label{fig:exampleIrrelevance}
		}
	\end{figure}
	The objective is to estimate 
	$\node_{1}$.
	As a first scenario assume that the nodes
	$\node_{A}$, $\node_{B}$, $\node_{C}$, $\node_{D}$
	are being measured. 
	Let us assume that $\node_{6}$ is added to the set of observed signals.
	Using Theorem~\ref{thm:non-causal Wiener-separation}, we conclude that addition of
	$\node_{6}$ 
	to the set of observed nodes
	$\{\node_{A}$, $\node_{B}$, $\node_{C}$, $\node_{D}\}$
	makes
	$\node_{A}$
	irrelevant for the estimate of
	$\node_{1}$,
	since
	\begin{align*}
		\dsep{\graph}(\node_{1},\left\{\node_{6},\node_{B}, \node_{C},\node_{D},\right\},\node_{A}).
	\end{align*}
	Thus, having both $\node_{A}$ and $\node_{6}$
	being measured in addition to $\node_{B}$, $\node_{C},$ and $\node_{D}$ being measured is wasteful in estimating $\node_{1}$.
	This is quite intuitive, since 
	$\node_{6}$
	is on the path from $\node_{1}$ to node $\node_{6}$.
	Similarly, the further addition of $\node_{7}$ to the set of observed nodes makes its child $\node_{B}$ irrelevant as well.
	Indeed we have
	\begin{align*}
		\dsep{\graph}(\node_{1},\left\{\node_{6},\node_{7}, \node_{C},\node_{D},\right\},\left\{\node_{A},\node_{B}\right\}).
	\end{align*}
	In general, if we use the nodes
	$\node_{6},\node_{7},\node_{8},\node_{9}, \node_{A},\node_{B},\node_{C},\node_{D}$
	for estimating $\node_{1}$, Theorem~\ref{thm:non-causal Wiener-separation} allows us to conclude that
	$\node_{A},\node_{B},\node_{C},\node_{D}$ are not needed, because
	\begin{align*}
		\dsep{\graph}(\node_{1},\left\{\node_{6},\node_{7}, \node_{8},\node_{9},\right\},\left\{\node_{A},\node_{B},\node_{C},\node_{D}\right\}).
	\end{align*}
	
	A less intuitive consequence of Theorem~\ref{thm:non-causal Wiener-separation} is the following.
	Let us go back to the situation where the observed nodes are again
	$\node_{A}$, $\node_{B}$, $\node_{C}$, $\node_{D}$,
	but now we have a sparsity requirement: we can only use $3$ of those signals for the estimation of $\node_{1}$.
	Observe that, even though
	$\node_{B}$, $\node_{C}$, $\node_{D}$
	are independent of $\node_{1}$,
	their information, in  combination with
	$\node_{A}$,
	can help improve the estimate of $\node_{1}$ that would be obtained using only $\node_{A}$.
	Indeed it can be verified that, in general, the Wiener filter estimating $y_{1}$
	from
	$\node_{A}$, $\node_{B}$, $\node_{C}$, $\node_{D}$
	has $4$ non-zero entries.
	However, the goal is to find a triplet of observed signals optimally estimating $y_{1}$ in the least square sense.
	We start with a straighforward observation: $y_{A}$ necessarily belongs to an optimal triplet since any triplet containing $y_{A}$ performs at least as well as $\{y_{B}$, $y_{C}$, $y_{D}\}$, which are all independent of $y_{1}$.
	This can also be seen via Theorem~\ref{thm:non-causal Wiener-separation}, since
	\begin{align*}
		\dsep{\graph}(\node_{1},\emptyset,\left\{\node_{B},\node_{C}, \node_{D}\right\}).
	\end{align*}
	Having estabilished that $y_{A}$ belongs to an optimal triplet for the estimation of $y_{1}$, we find now that
	$\{y_{A},y_{B}\}$ is necessarily contained in an optimal triplet.
	Indeed, any triplet containing $\{y_{A},y_{B}\}$ performs as well as any other triplet containing $\{y_{A}\}$ when estimating $y_{1}$, since
	\begin{align*}
		\dsep{\graph}(\node_{1},\node_{A},\left\{\node_{C}, \node_{D}\right\}).
	\end{align*}
	implies that the estimate using only $\{y_{A}\}$ performs as well as the estimate using 
	$\{y_{A},y_{C},y_{D}\}$.
    Given that $\{y_{A},y_{B}\}$ is contained in an optimal triplet for the estimation of $y_{1}$, we eventually find that
    $\{y_{A},y_{B},y_{C}\}$ is necessarily optimal, since 
	\begin{align*}
		\dsep{\graph}(\node_{1},\left\{\node_{A},\node_{B}\right\} , \node_{D}).
	\end{align*}
    implies that $\{y_{A},y_{B}\}$ performs as well as $\{y_{A},y_{B},y_{D}\}$.
    Thus, irrespective of the actual transfer functions and power spectral densities of the network
    $\{y_{A},y_{B},y_{C}\}$ is always an optimal triplet for the estimate of $y_{1}$.
	Observe how only graphical properties lead to this conclusion without the need to solve any involved optimization problem with sparsity contraints.
\end{example}

\section{Identification of individual transfer functions in a LDIM}\label{sec:idresults}
This section presents the contributions of the article about the identification of individual transfer functions in a LDIM.

The single door criterion is a powerful tool developed for the identification of parameters in structural equation models \cite{Pea00}. Here, we provide a similar criterion for the identification of transfer functions in a LDIM.

\begin{lem}\label{lem:generalized single door}
	Let
	$\graph=(\nodeset,\directededgeset)$
	be a graphical representation of the LDIM
	$\ldim=(\adjTF(\ZTvar),\noise)$
	with output
    $\node=(\node_{1}^{T},...,\node_{\nodeindexend}^{T})$,
	where $G$ is allowed to have self-loops.
	Let
	$\mutilatedgraph$
	be the graph obtained by removing the link
	$(\node_{\nodeindex},\node_{\nodeindexalt})$
	from the graph $\graph$.
	For some $\subsetnodesep\subset \node$, let 
	$\nodesetsmallertocheckdsep:=\ancestor{\graph}{\node_{\nodeindex},\node_{\nodeindexalt},\subsetnodesep}$.
	Define the following sets of processes:
	\small
	\begin{align*}
		\subsetnodealt &:=\{\node_{\nodeindexaltb}\in \nodesetsmallertocheckdsep|\ \ 
			\dsep{\ov{G}}(\node_{\nodeindexaltb},\subsetnodesep,\node_{\nodeindex})\},&\\
		\subsetnodealtenlarged &:=\subsetnodealt\setminus \{\node_{\nodeindexalt}\}\\
		\node_{K}&:=\{\node_{k}\in\subsetnodesep| \text{ there is an edge } \node_{\nodeindexalt}\to \node_{k}&\\
		& \qquad \qquad \qquad \text{ or there is a chain } \node_{\nodeindexalt}\to ...\to \node_{k}\\
		& \qquad \qquad \qquad \text{ with internal nodes all in } \node_{J} \}.&
	\end{align*}
	\normalsize
Assume $\PSD_{(\node_{\nodeindex},\node_{\nodeindexalt},\node_{Z})}(\ZTvar)>0$ for all $|\ZTvar|=1$ and
\begin{itemize}
    \item[A1.]\label{item:A1} $\node_i$ and $\node_j$ are $d$-separated by $\node_Z$  in the graph $\overline{G}$;
    \item[A2.]\label{item:A2} in the restriction of $G$ to $\node_{J}$, $\node_{j}$ is involved in no loops. 
    \item[A3.]\label{item:A3}There is no link of the form $\node_{k}\to\node_{\nodeindexalt}$ for $\node_{k}\in\node_{K}$; 
\end{itemize}
Under the above assumptions, 
if
$\node_{K}$
is empty we have 
\begin{align}
    \adjTF_{\nodeindexalt\nodeindex}=\Wiener[\node_{\nodeindex}]{\node_{\nodeindexalt}}{(\node_{\nodeindex},\node_{Z})}.
\end{align}
Instead, if $\node_{K}$ is not empty, the following holds
	\begin{align}\label{eq:generalized single door}
	\adjTF_{\nodeindexalt\nodeindex}
		=\left[
				\identitymatrix
				-\Wiener[\node_{K}]{\node_{\nodeindexalt}}{(\node_{\nodeindex},\node_{Z})}
				\adjTF_{Kj}^{(f)}
		\right]^{-1}
		\Wiener[\node_{\nodeindex}]{\node_{\nodeindexalt}}{(\node_{\nodeindex},\node_{Z})},
	\end{align}
	where
	\begin{equation}
	\begin{aligned}\label{eq:HKJf}
		\adjTF_{Kj}^{(f)}&=
				\left(
					\identitymatrix-\left( \adjTF_{KK}+\adjTF_{K\bar J}(\identitymatrix-\adjTF_{\bar J \bar J})^{-1}\adjTF_{\bar JK}
				\right)\right)^{-1}\cdot\\
				&\qquad\qquad\cdot
				\left(
					\adjTF_{Kj}+\adjTF_{K\bar{J}}(\identitymatrix-\adjTF_{\bar{J}\bar{J}})^{-1}\adjTF_{\bar{J}j}
				\right).
	\end{aligned}
	\end{equation}
Here
	$\Wiener[\node_{\nodeindex}]{\node_{\nodeindexalt}}{(\node_{\nodeindex},\node_{Z})}$
	and $W_{\node_j[\node_K]|(\node_i,\node_Z)}$  are the component corresponding to
	$\node_{\nodeindex}$ and $\node_K$, respectively,
	in the Wiener filter that estimates the process
	$\node_{\nodeindexalt}$
	based on processes
	$\node_{\nodeindex}$
	and	 
	${\subsetnodesep}$.
\end{lem}
\begin{proof}
	See the Appendix.
\end{proof}

We illustrate the application of Lemma~\ref{lem:generalized single door} with an example.
\begin{example}\label{ex:revolvingdoorsets}
    Consider a LDIM with graph representation as in Figure~\ref{fig:revolvingdoorsets}(a) where we have chosen $\node_{i}=\node_{1}$ and $\node_{j}=\node_{2}$.
    \begin{figure}[h!]
        \centering
        \begin{tabular}{cc}
        \includegraphics[width=0.3\columnwidth]{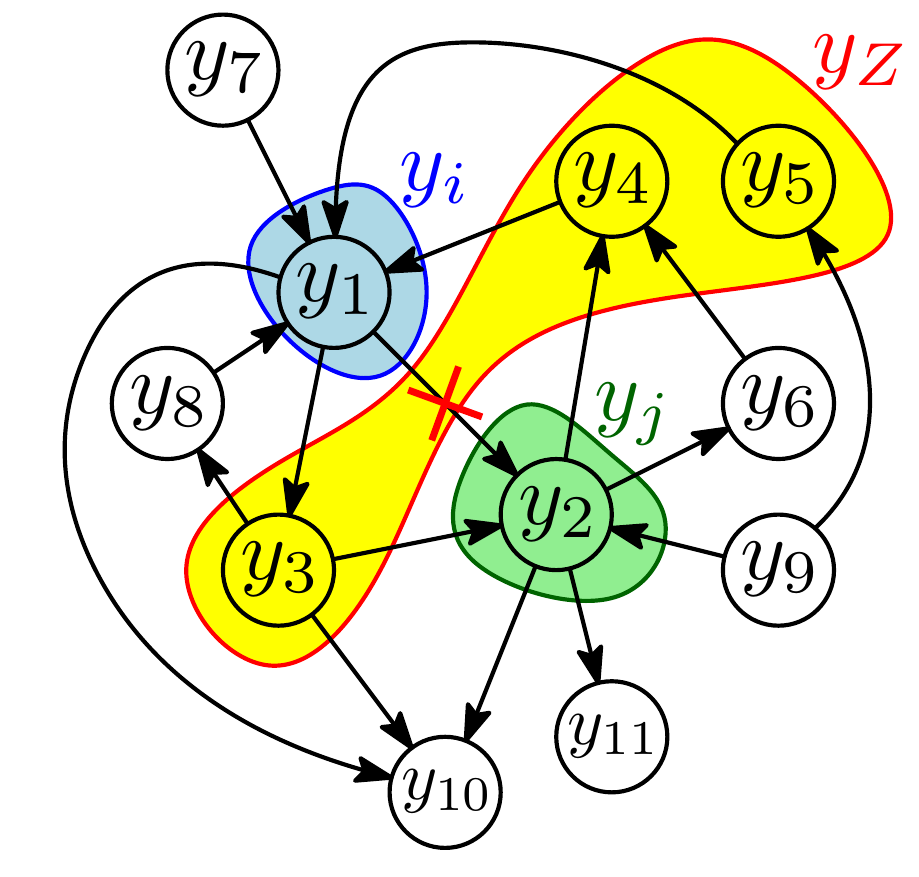} &
        \includegraphics[width=0.3\columnwidth]{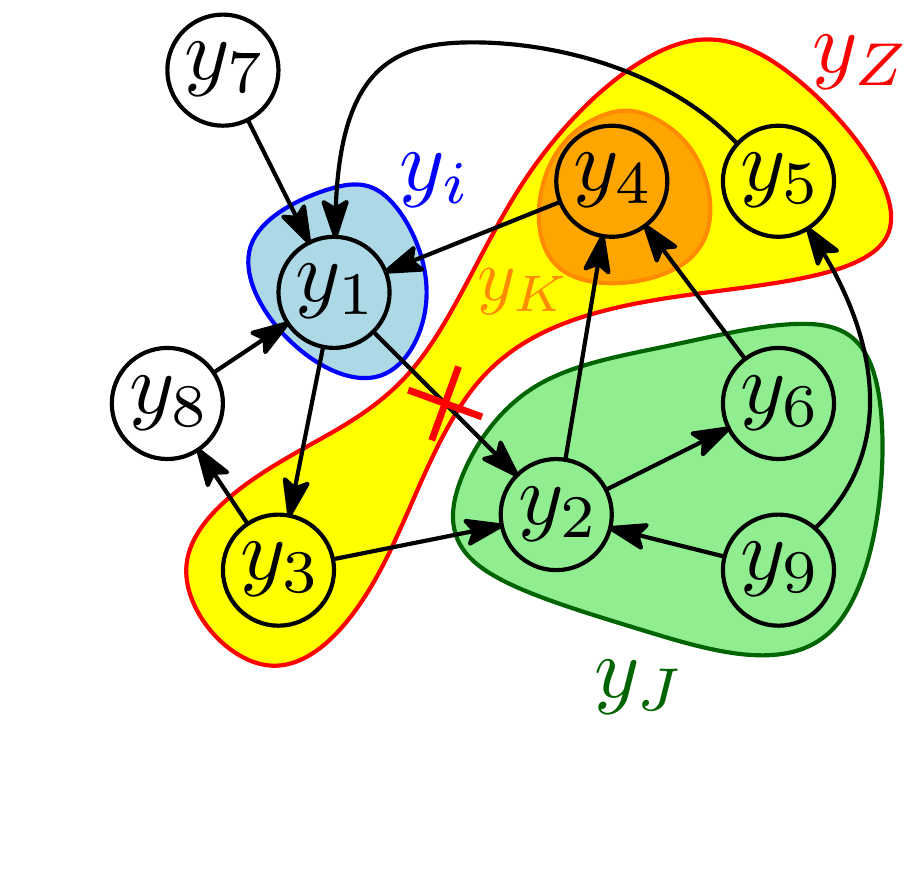} \\
        (a) & (b)
        \end{tabular}
        \caption{
        (a) Application of Lemma~\ref{lem:generalized single door} with $y_{i}=y_{1}$, $y_{j}=y_{2}$ and 
        $y_{Z}=\{y_{3},y_{4},y_{5}\}$ in Example~\ref{ex:revolvingdoorsets}.
        (b) Restriction to $an_{G}(\{y_{i},y_{j}\}\cup y_{Z})$ and sets $y_{J}$ and $y_{K}$.
        Observe that if the ``door'' (edge) from $y_{i}$ and $y_{j}$ is removed, the two nodes are $d$-separated by $y_{Z}$.
        \label{fig:revolvingdoorsets}
        }
    \end{figure}
    Let also $\node_{Z}=\{\node_{3},\node_{4},\node_{5}\}$.
    Observe that Assumption A1. is verified since $dsep(\node_{1},\{\node_{3},\node_{4},\node_{5}\},\node_{2})$ in the graph $\overline{G}$ obtained by removing the edge $\node_{1}\to \node_{2}$ (marked with a cross).
    The restriction of the LDIM to $\node_{U}=an_{G}\{\{\node_{i},\node_{j}\}\cup \node_{Z}\}$ is shown in Figure~\ref{fig:revolvingdoorsets}(b), where it is also shown the set $\node_{J}$ of all the nodes in $\node_{U}$ that are $d$-separated from $\node_{i}$ by $\node_{Z}$ in $\overline{G}$. $\node_{J}$ is given by $\{\node_{2},\node_{6},\node_{9}\}$.
    Assumption A2. is then trivially verified.
    The only node in $\node_{K}$ is $\node_{4}$, because it is the only node in $\node_{Z}$ for which there is a chain starting from $\node_{j}$ with all internal nodes in $\node_{J}$.
    Thus, Assumption A3. is verified as well. The application of the Lemma~\ref{lem:generalized single door} yields
    \begin{align*}
        \adjTF_{21}
        =
        (1-W_{\node_{2}[\node_{4}]|(\node_{1},\node_{4})}\adjTF^{(f)}_{41})^{-1}W_{\node_{2}[\node_{1}]|(\node_{1},\node_{4})},
    \end{align*}
    where the expression of $\adjTF^{(f)}_{41}$ is given by \eqref{eq:HKJf}.
\end{example}

In Equation~(\ref{eq:generalized single door}), the term
$\Wiener[\node_{\nodeindex}]{\node_{\nodeindexalt}}{(\node_{\nodeindex},\node_{Z})}$
is a function of the power spectral densities of
$\node_{\nodeindex}$,
$\node_{\nodeindexalt}$
and
$\subsetnodesep$), thus computable from data.
Instead, the term
$\adjTF_{Kj}^{(f)}$
is in general not computable from the power spectral densities.
If it is known that the set
$\subsetnodesep$
contains no descendants of
$\node_{\nodeindex}$
we have that $\node_{K}$ is empty and thus
$\adjTF_{\nodeindexalt\nodeindex}$
matches the expression of 
$\Wiener[\node_{\nodeindex}]{\node_{\nodeindexalt}}{(\node_{\nodeindex},\node_{Z})}$,
leading to the following important consequence. 
\begin{thm}[Single door criterion for LDIMs]\label{thm:single door}
	Let
	$\graph=(\nodeset,\directededgeset)$
	be a graphical representation of the LDIM
	$\ldim=(\adjTF(\ZTvar),\noise)$
	with output
    $\node=(\node_{1}^{T},...,\node_{\nodeindexend}^{T})$,
	where $G$ has no self-loop in $y_{j}$.
	Let
	$\mutilatedgraph$
	be the graph obtained by removing the link
	$(\node_{\nodeindex},\node_{\nodeindexalt})$
	from the graph $\graph$.
    Assume $\PSD_{(\node_{\nodeindex},\node_{\nodeindexalt},\node_{Z})}(\ZTvar)>0$ for all $|\ZTvar|=1$ and
    \begin{itemize}
    \item[A1.]\label{item:A1b} $\node_i$ and $y_j$ are $d$-separated by $\node_Z$  in the graph $\overline{G}$;
    \item[A4.]\label{item:A4}  $\node_Z$  has no descendants of $y_{j}$.
    \end{itemize}
    Then, $\adjTF_{ji}=\Wiener[\node_{\nodeindex}]{\node_{\nodeindexalt}}{(\node_{\nodeindex},\node_{Z})}$.
\end{thm}
\begin{proof}
    The absence of loops in $y_{j}$ and assumption A4. implies assumptions A2. and A3. in Lemma~\ref{lem:generalized single door} and the fact that $y_{K}=\emptyset$, giving immediately the assertion.
\end{proof}
A criterion similar to Theorem~\ref{thm:single door}, but valid for the identification of simple proportional gains and not general transfer functions, is known in the area of Structural Equation Models as the Single Door criterion \cite{Pea00}.
The main intuition behind Theorem~\ref{thm:single door} is the following:
if there is only a ``single door'' (edge) from
$\node_{\nodeindex}$
to
$\node_{\nodeindexalt}$
that prevents these two nodes from being ``separated''
(in the sense of $d$-separation) by a set $y_{Z}$ that does not contain any descendants of $y_{j}$, 
then the Wiener filter estimating $y_{j}$ from $y_{i}\cup y_{Z}$ has component associated with $y_{i}$ equal to the transfer function $\adjTF_{ji}$.

The identification of transfer functions in presence of confounding variables is a topic of active investigation (see \cite{DanVan17}) and Theorem~\ref{thm:single door} can leverage the notion of $d$-separation towards such a goal, as the following example shows.
\begin{example}[A network with a confounder and no observable external variable]\label{ex:confounder}
	\normalsize
	Consider the graph
	$\graph$
	of Figure~\ref{fig:confounding}(a) representing a LDIM following the dynamics
	$\node=\noise+\adjTF(\ZTvar)\node$
	where the only potentially non-zero entries of
	$\adjTF(\ZTvar)$
	are
	$\adjTF_{21}$, $\adjTF_{32}$, $\adjTF_{14}$, $\adjTF_{41}$ and $\adjTF_{34}$. 
	\begin{figure}[hbt!]
		\centering
		\begin{tabular}{ccc}
		\includegraphics[width=0.29\columnwidth]{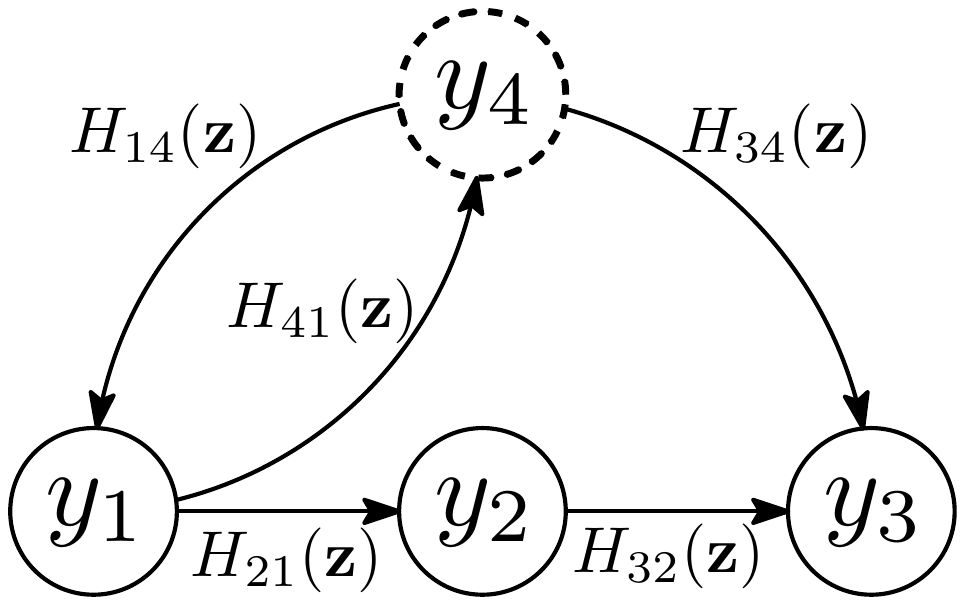} &
		\includegraphics[width=0.29\columnwidth]{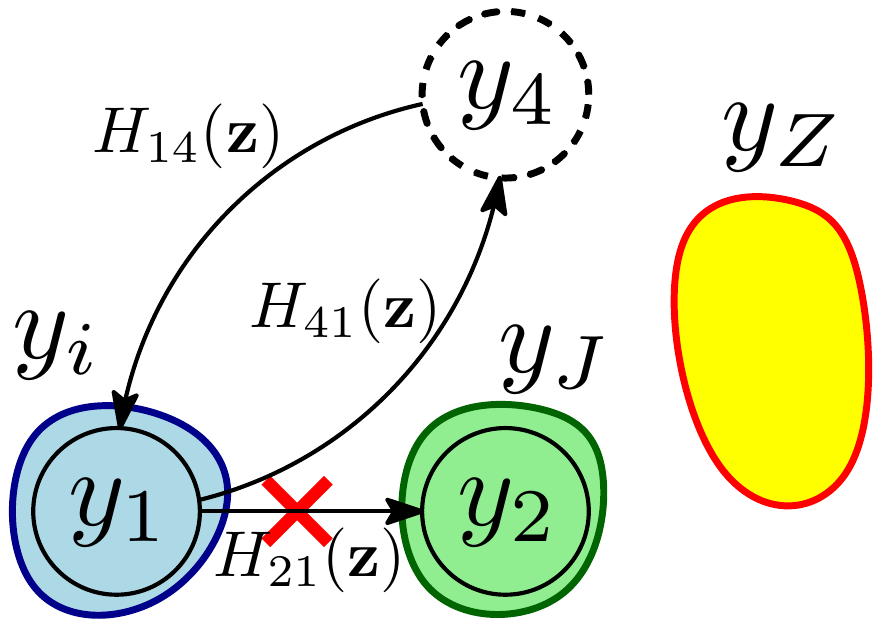} &
		\includegraphics[width=0.29\columnwidth]{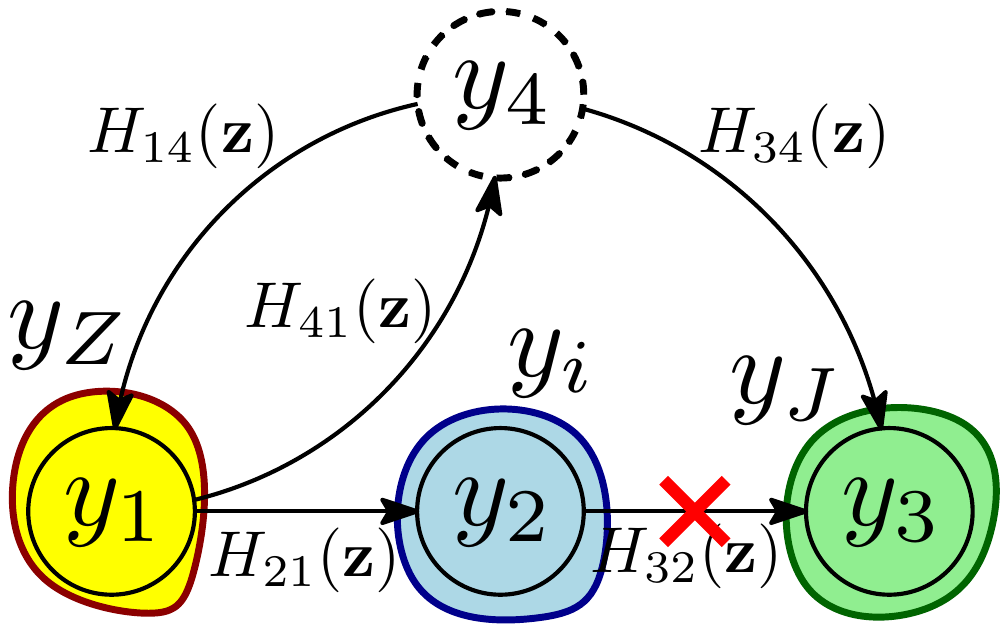}\\
		(a) & (b) & (c)
		\end{tabular}
		\caption{
			(a) Network with a confounding variable $\node_{4}$.
			If the confounding variable were not present, the identification of the transfer functions
			$\adjTF_{21}(\ZTvar)$
			and
			$\adjTF_{32}(\ZTvar)$
			could be obtained as
			$\adjTF_{21}(\ZTvar)=\PSD_{\node_{2}\node_{1}}\PSD_{\node_{1}\node_{1}}^{-1}$
			and
			$\adjTF_{32}(\ZTvar)=\PSD_{\node_{3}\node_{2}}\PSD_{\node_{2}\node_{2}}^{-1}$.
			The presence of the confounding variable makes these relations not correct in the general case.
			The application of the Single Door Criterion
			(Theorem~\ref{thm:single door}) provides a solution to this identification problem that is robust with respect to the presence/absence of the $\node_{4}$.
			(b) Application of the Single Door Criterion for the identification of 
			$\adjTF_{21}(\ZTvar)$.
			(c) Application of the Single Door Criterion for the identification of 
            $\adjTF_{32}(\ZTvar)$
			\label{fig:confounding}
		}
	\end{figure}
	Also assume that only $\node_{1}$, $\node_{2}$ and $\node_{3}$ are observed, but not $\node_{4}$.
	The goal is to identify the transfer functions
	$\adjTF_{21}(\ZTvar)$
	and
	$\adjTF_{32}(\ZTvar)$.
	Let $\graph_{\ov{1\to 2}}$ be the graph obtained by removing the edge $1\to 2$ from $\graph$.
	Observe that 
	$\dsep{\graph_{\ov{1\to 2}}}(\node_{1},\emptyset,\node_{2})$.
	By applying Theorem~\ref{thm:single door}, with $y_{Z}=\emptyset$, as shown in Figure~\ref{fig:confounding}(b), we get
	\begin{align}\label{eq:ex confounding estimate H21}
		\adjTF_{21}
		=
		\Wiener[\node_{1}]{\node_{2}}{\node_{1}}
		=
		\PSD_{\node_{2}\node_{1}}\PSD_{\node_{1}\node_{1}}^{-1}.
	\end{align}
	This is not surprising because
	$\node_{2}$
	is influenced only by node
	$\node_{1}$
	and Equation~\eqref{eq:ex confounding estimate H21} is a trivial consequence of Wiener-Khinchin Theorem \cite{KaiSay00}.
	A more difficult task is to identify
	$\adjTF_{32}(\ZTvar)$
	because of the presence of the confounding process
	$\node_{4}$
	that is not observed.
	Indeed the relation
	\begin{align}\label{eq:ex nonrobust ID}
		\adjTF_{32}=\PSD_{\node_{3}\node_{2}}\PSD_{\node_{2}\node_{2}}^{-1}.
	\end{align}
	could be obtained from Wiener-Khinchin Theorem only if $\adjTF_{34}=0$ or $\PSD_{\node_{4}}=0$.
	Instead, Theorem~\ref{thm:single door} provides a way to identify 
	$\adjTF_{32}(\ZTvar)$.
	Let $\graph_{\ov{2\to 3}}$ be the graph obtained by removing the edge $2\to 3$ from $\graph$.
	Observe that
	$\dsep{\graph}(\node_{2},\node_{1},\node_{3})$.
	By applying Theorem~\ref{thm:single door} with $\node_{Z}=\{y_{1}\}$, as shown in Figure~\ref{fig:confounding}(b), we obtain
	\begin{align}\label{eq:ex confounding estimate H32}
		\adjTF_{32}
		&=
		\Wiener[\node_{2}]{\node_{3}}{(\node_{2},\node_{1})}\\
		&=
		\left(\begin{array}{cc}
			\PSD_{\node_{3}\node_{2}} & \PSD_{\node_{3}\node_{1}}
		\end{array}\right)
		\left(\begin{array}{cc}
			\PSD_{\node_{2}\node_{2}} & \PSD_{\node_{2}\node_{1}}\\
			\PSD_{\node_{1}\node_{2}} & \PSD_{\node_{1}\node_{1}}
		\end{array}\right)^{-1}
		\left(\begin{array}{c}
			1\\ 0
		\end{array}\right). \nonumber
	\end{align}
	Thus, from \eqref{eq:ex confounding estimate H32}, we have that the identification of
	$\adjTF_{32}(\mathbf{z})$
	can be obtained from the knowledge of the power spectral densities of
	$\node_{1}$,
	$\node_{2}$
	and
	$\node_{3}$
	only.
	Notice that, irrespective of node
	$\node_{4}$
	and how it influences
	$\node_{1}$ and $\node_{3}$,
	Equation~\eqref{eq:ex confounding estimate H32}
	always provides an unbiased identification for $\adjTF_{32}(\mathbf{z})$.
	In particular $\node_{4}$ might exist or not in the network and still 
	\eqref{eq:ex confounding estimate H32} would be correct.
	In other words, \eqref{eq:ex confounding estimate H32} is modification of 
	\eqref{eq:ex nonrobust ID} that is robust with respect to the presence of $\node_{4}$ or the link $\node_{4}\to\node_{3}$.
	Thus, the identification approach of this article provides a novel form of robustness against uncertainties in the network structure.
\end{example}

We also consider an example similar to Example~\ref{ex:confounder} to show that the assumption that the unknown forcing inputs $\noise$ are mutually independent can be relaxed in the framework developed in this article.
This can be done by exploiting the fact that not all outputs $y$ need to be measure to apply the single door criterion. 
\begin{example}\label{ex:DSF comparison}
    Consider a network in the DSF form $y=Q(\ZTvar)y+P(\ZTvar)u$ with a non-strictly causal $Q$ and a non-diagonal $P$
    \small
\begin{align*}
	\left(\begin{array}{c}
	y_{1}\\
	y_{2}\\
	y_{3}
	\end{array}\right)
	=
	&
	\left(\begin{array}{ccc}
	0 & 0 & 0\\
	\adjTF_{21} & 0 & 0\\
	0 & \adjTF_{32} & 0
	\end{array}\right)
	\left(\begin{array}{c}
	y_{1}\\
	y_{2}\\
	y_{3}
	\end{array}\right)\\
	&+
	\left(\begin{array}{ccc}
	P_{11} & 0 & P_{13}\\
	0 & P_{22} & 0\\
	0 & 0 & P_{33}
	\end{array}\right)
	\left(\begin{array}{c}
	u_{1}\\
	u_{2}\\
	u_{3}
	\end{array}\right)
\end{align*}
\normalsize
where $\PSD_{u}(\ZTvar)$ is diagonal. The objective is to determine $\adjTF_{32}$ and $\adjTF_{21}$ by knowing
\small
\begin{align*}
	\Phi_{(y_{1},y_{2},y_{3})}
	=
	\left(\begin{array}{ccc}
	2 & 2 & 3\\
	2 & 3 & 4\\
	3 & 4 & 6
	\end{array}\right);
	&
    \Phi_{u}
	=
	\left(\begin{array}{ccc}
	1 & 0 & 0\\
	0 & 1 & 0\\
	0 & 0 & 1
	\end{array}\right);\\
    Q
	=
	\left(\begin{array}{ccc}
	0 & 0 & 0\\
	* & 0 & 0\\
	0 & * & 0
	\end{array}\right);
	&
    P
	=
	\left(\begin{array}{ccc}
	* & 0 & *\\
	0 & * & 0\\
	0 & 0 & *
	\end{array}\right).
\end{align*}
\normalsize
where $*$ denotes an entry potentially different from zero.
Observe that this DSF model does not meet the assumptions in \cite{hayden2017network}, since the $Q$ is not strictly causal and $P$ is not diagonal.
The results in \cite{GonWar08} can provide $Q$ and $P$ from the knowledge of  $G(\ZTvar)=(I-Q(\ZTvar))^{-1}P(\ZTvar)$,
by solving the matrix equality
\small
\begin{align}\label{eq:DSF informativity conditions}
	G=QG+P=(Q, P)
		\left(\begin{array}{cc}
		G\\
		I
		\end{array}\right)
	\Leftrightarrow
	G^{T}= (G^{T}, I)
		\left(\begin{array}{cc}
		Q^{T}\\
		P^{T}
		\end{array}\right).
\end{align}
\normalsize

Such an equation admits a solution if $3$ elements in each row of the matrix $(Q,P)$ are already known (DSF informativity conditions in \cite{GonWar08}). Observe that $4$ elements in each row of $(Q,P)$ are known, thus the problem of computing $Q$ and $P$ from the knowledge $G$ is overdetermined.
On the other hand, $G(\ZTvar)$ is not known directly, but it is known that $G(\ZTvar)G(\ZTvar)^{*}=\Phi_{(y_{1},y_{2},y_{3})}$. From the sparsity patterns of $Q$ and $P$, we have the sparsity pattern of $G$
\small
\begin{align*}
	G=
	(I-Q)^{-1}P
	=
	\left(\begin{array}{ccc}
		* & 0 & *\\
		* & * & *\\
		* & * & *
	\end{array}\right).
\end{align*}
\normalsize
Multiple $G(\ZTvar)$ satisfy such a factorization with such a sparsity pattern, for example
\small
\begin{align*}
	G_2=
	\left(\begin{array}{ccc}
	\sqrt{2} & 0 & 0\\
	\sqrt{2} & 1 & 0\\
	\frac{3}{\sqrt{2}} & 1 & \frac{1}{\sqrt{2}}
	\end{array}\right);
	&
    G_{3}
    =
	\left(\begin{array}{ccc}
        -\frac{2}{\sqrt{3}} & 0 &  \frac{\sqrt{2}}{\sqrt{3}}\\
        \frac{1-2\sqrt{2}}{\sqrt{6}} & \frac{1}{\sqrt{2}} & \frac{1+\sqrt{2}}{\sqrt{3}}\\
        \frac{\sqrt{2}-7}{2\sqrt{3}} & \frac{1+\sqrt{2}}{2} & \frac{1+\sqrt{2}}{3}
	\end{array}\right).
\end{align*}
\normalsize
Thus, finding the correct factor $G$ from $G(\ZTvar)G(\ZTvar)^{*}=\Phi_{(y_{1},y_{2},y_{3})}$ is an underdetermined problem. If we try to plug $G_{2}$ or $G_{3}$ into Equation~\eqref{eq:DSF informativity conditions}, we find no $(Q,P)$ with the known sparsity pattern (we remind that Equation~\eqref{eq:DSF informativity conditions} was overdetermined).
The fundamental challenge is that the factorization $\Phi_{(y_{1},y_{2},y_{3})}=GG^{*}$ admits an infinite number of solutions and it is a quadratic problem in the entries of $G$ which in turn needs to satisfy $G=(I-Q)P^{-1}$ with specific sparsity patterns for $Q$ and $P$. The infinite solutions, then, need to be checked against Equation~\eqref{eq:DSF informativity conditions} to see if they are admissible.
To the best knowledge of the authors, this problem is difficult to tackle using current DSF tools.

On the other hand, such a network is not a LDIM, since $e=P(\ZTvar)u$ does not have a diagonal PSD.
The problematic input is $u_{3}$ which influences both $y_{1}$ and $y_{3}$.
However, by introducing a new output node $y_{4}$ we can rewrite the system, so that it is a LDIM
\small
\begin{align*}
	\left(\begin{array}{c}
	y_{1}\\
	y_{2}\\
	y_{3}\\
	y_{4}
	\end{array}\right)
	=
	\left(\begin{array}{cccc}
	0 & 0 & 0 & 1\\
	1 & 0 & 0 & 0\\
	0 & 1 & 0 & 1\\
	0 & 0 & 0 & 0
	\end{array}\right)
	\left(\begin{array}{c}
	y_{1}\\
	y_{2}\\
	y_{3}\\
	y_{4}
	\end{array}\right)
	+
	\left(\begin{array}{c}
	e_{1}\\
	e_{2}\\
	e_{3}\\
	e_{4}
	\end{array}\right).
\end{align*}
\normalsize
    In this form, the system is recognized as a LDMI of the same form of the LDMI in Example~\ref{ex:confounder}.
    The application of the single door criterion gives
\begin{align*}
	Q_{21}(z)&=\Phi_{y_2 y_1}(z)/\Phi_{y_1 y_1}(z)=\frac{2}{2}=1\\
    Q_{32}(z)
    &=
    \left( \Phi_{y_{3}y_{1}}, \Phi_{y_{3}y_{2}}\right) 
	\left(\begin{array}{ccc}
        \Phi_{y_{1}y_{1}} &  \Phi_{y_{1}y_{2}}\\
        \Phi_{y_{1}y_{2}} & \Phi_{y_{2}y_{2}}
	\end{array}\right)^{-1}
	\left(\begin{array}{c}
	1\\0
	\end{array}\right)\\
	&=
	(3, 4)
	\left(\begin{array}{cc}
        2 & 2 \\
        2 & 3
	\end{array}\right)^{-1}
	\left(\begin{array}{c}
	0\\1
	\end{array}\right)
	=1.
\end{align*}
    Observe that the single door criterion has not used the fact that $\PSD_{u}=I$, but only that it is diagonal.
\end{example}

We have shown that the single door criterion allows us to identify a single transfer function
$\adjTF_{\nodeindexalt\nodeindex}(\ZTvar)$
in a LDIM with no self-loop at $\node_{\nodeindexalt}$ without necessarily observing all the signals of the network.
The main assumption is that, after removing the edge
$\node_{\nodeindex}\to \node_{\nodeindex}$
it must be possible to $d$-separate
$\node_{\nodeindex}$
and
$\node_{\nodeindexalt}$
in the graph of the LDIM using no descendants of
$\node_{\nodeindexalt}$.
This assumption is not very limiting if the edge
$\node_{\nodeindex}\to \node_{\nodeindexalt}$
is not involved in a directed loop. Indeed, the standard formulation of the single door criterion has been developed in the context of Directed Acyclic Graphs where this assumption is always verified. 
However, if the edge
$\node_{\nodeindex}\to \node_{\nodeindexalt}$
is involved in a directed loop, it is not possible to $d$-separate
$\node_{\nodeindex}$
and
$\node_{\nodeindexalt}$
making no use of descendants of
$\node_{\nodeindexalt}$. 
The following example shows how to use Lemma~\ref{lem:generalized single door} instead of Theorem~\ref{thm:single door} when the node 
$\node_{\nodeindexalt}$
is involved in a loop if other transfer functions of the network are known.
\begin{example}[Closed loop identification with knowledge of a transfer function]\label{ex:CLIDwithTF}
	\normalsize
	\begin{figure}[ht]
		\centering
        \begin{tabular}{cc}
		\includegraphics[width=0.28\columnwidth]{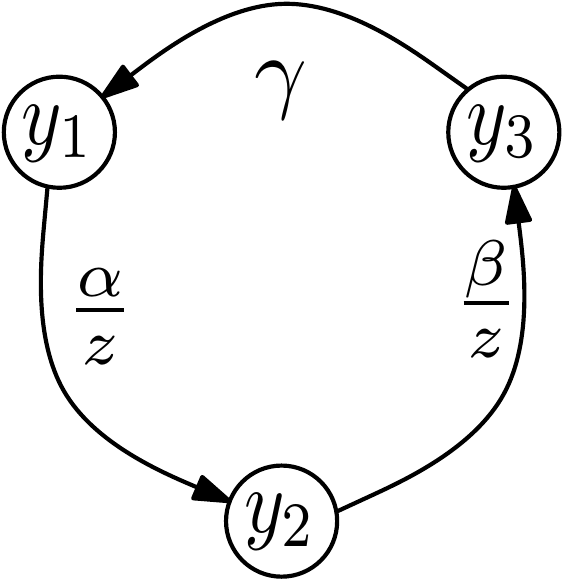} & 
		\includegraphics[width=0.28\columnwidth]{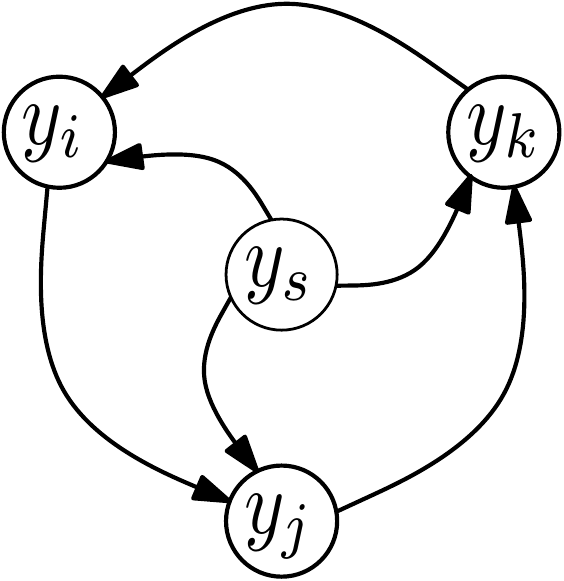}\\
		(a) & (b)
        \end{tabular}
		\caption{
			(a) The loop network of Example~\ref{ex:CLIDwithTF} and Example~\ref{ex:CLIDwithoutTF}.
			(b)         Network configuration of Theorem~\ref{thm:revolvingDoor}.
			\label{fig:loopID}
            \label{fig:revolvingDoorWithOneExternal}
		}
	\end{figure}
	Consider again the loop network in Figure~\ref{fig:loopID}(a) and this time assume that it is only known that the network only involves causal transfer functions.
	The objective is to identify the transfer function
	$\adjTF_{21}$ relying on the power spectral densities of the output processes of the LDIM.
	After removing the link
	$\node_{1}\to \node_{2}$,
	it is not possible to
	$d$-separate nodes
	$\node_{1}$
	and
	$\node_{2}$
	without using node $\node_{3}$
	which is a descendant of $\node_{2}$.
	Thus, the single door criterion (Theorem~\ref{thm:single door}) cannot be used.
	However, if it is known that
	$\adjTF_{32}(\ZTvar)=\frac{\beta}{\ZTvar}$
	we can apply Lemma~\ref{lem:generalized single door}
	where
	$\node_{\nodeindex}=\node_{1}$,
	$\node_{\nodeindexalt}=\node_{2}$,
	$\subsetnodesep=\node_{K}=\{\node_{3}\}$,
	$\adjTF_{K{\nodeindexalt}}^{(f)}=\adjTF_{32}(\ZTvar)=\frac{\beta}{\ZTvar}$.
	For the application of Lemma~\ref{lem:generalized single door} to compute
	$\adjTF_{21}$
	we need the quantities
	$\Wiener[\node_{1}]{\node_{2}}{(\node_{1},\node_{3})}$
	and
	$\Wiener[\node_{3}]{\node_{2}}{(\node_{1},\node_{3})}$
	which can be computed from the knowledge of the power and cross-power spectral densities of the signals
	$\node_{1}$,
	$\node_{2}$, and
	$\node_{3}$. 
	Several methods exist to estimate the power and cross-power spectral densities from data.
	We illustrate the application of Lemma~\ref{lem:generalized single door} using the analytical expressions of
	$\Wiener[\node_{1}]{\node_{2}}{(\node_{1},\node_{3})}$
	and
	$\Wiener[\node_{3}]{\node_{2}}{(\node_{1},\node_{3})}$
	which are
	\begin{align}\label{eq: example loop Wiener for node 1}
		\Wiener[\node_{1}]{\node_{2}}{(\node_{1},\node_{3})}
		=
		\frac{\alpha\PSD_{\noise_{3}}\ZTvar}{\PSD_{\noise_{3}}\ZTvar^2+\beta^{2}\PSD_{\noise_{2}}}\text{ and}\\
		\Wiener[\node_{3}]{\node_{2}}{(\node_{1},\node_{3})}
		=
		\frac{\beta\PSD_{\noise_{2}}\ZTvar}{\PSD_{\noise_{3}}\ZTvar^2+\beta^{2}\PSD_{\noise_{2}}}.
	\end{align}
	From Lemma~\ref{lem:generalized single door} we find
	\begin{align*}
		\adjTF_{21}
		&=
		\left(
			1-\Wiener[\node_{3}]{\node_{2}}{(\node_{1},\node_{3})}\adjTF_{32}
		\right)
		\Wiener[\node_{1}]{\node_{2}}{(\node_{1},\node_{3})}\\
		&=
		\left(
			1
			-\frac{\beta\PSD_{\noise_{2}}\ZTvar}
					{\PSD_{\noise_{3}}\ZTvar^2+\beta^{2}\PSD_{\noise_{2}}}\frac{\beta}{\ZTvar}
		\right)
		\frac{\alpha\PSD_{\noise_{3}}\ZTvar}{\PSD_{\noise_{3}}\ZTvar^2+\beta^{2}\PSD_{\noise_{2}}}
		=\frac{\alpha}{\ZTvar}
	\end{align*}
	and thus the transfer function $\adjTF_{21}$ is recovered correctly.
	Notice how the identification results without measuring the input signals of the network
	$\noise_{1}$,
	$\noise_{2}$, and
	$\noise_{3}$,
	but just its outputs. Also no assumption is needed about the color of the noises.
\end{example}

\subsection{Generalized single door criterion for closed loop identification}
Lemma~\ref{lem:generalized single door} can overcome the limitations of the standard single-door criterion to deal with closed loop identification problems even when other transfer functions in the network are not known and potentially null (to model partial knowledge of the network). 
To exemplify the versatility of Lemma~\ref{lem:generalized single door}, we derive a novel methodology for the identification of a transfer function involved in a closed loop in the potential presence of a measurable common process affecting all the nodes of the loop.
A wealth of other, generally parametric, techniques exist to identitify individual transfer functions in networks \cite{van2013identification,DanVan15,DanVan16}.
These techniques can be applied to situations where loops are present, but require the presence of at least a strictly causal transfer function in each loop and some form of knowledge about which of the transfer functions in the loop is strictly causal \cite{van2013identification,DanVan15}, or cannot deal in general with uncertainties in the network topology because edges in the graph cannot represent null transfer functions \cite{DanVan16}.
Our method has the advantage of not requiring such knowledge and of being applicable to scenarios where edges are only potentially present. Furthermore, our method, being non-parametric, does not need an explicit noise model and the solution of non-convex optimization problems.
However such a technique requires the solution of an implicit relation involving the transfer function $\adjTF_{ji}(\ZTvar)$.
\begin{thm}\label{thm:revolvingDoor} [Revolving Door Criterion]
	Consider a LDIM
	$\ldim=(\adjTF(\ZTvar),\noise)$
	with output
	$(\node_{i},\node_{j},\node_{k},\node_{s})$ and graphical representation, $\graph=(\nodeset,\directededgeset)$, shown in Figure~\ref{fig:revolvingDoorWithOneExternal}(b).
	Then, the following implicit relation involving the transfer function $\adjTF_{ji}(\ZTvar)$ holds
	\small
	\begin{align}\label{eq:three floor equation}
		\adjTF_{ji}=
		\left\{
			\identitymatrix-\Wiener[\node_{k}]{j}{i,\node_{k},s}\Wiener[j]{\node_{k}}{J,i,s}-\Wiener[\node_{k}]{j}{i,\node_{k},s}\Wiener[i]{\node_{k}}{j,i,s}\cdot
			\right. \nonumber \\
		\left.
			\cdot(\identitymatrix-\Wiener[j]{i}{\node_{k},j,s}\adjTF_{ji}-\Wiener[\node_{k}]{i}{\node_{k},j,s}\Wiener[i]{\node_{k}}{j,i,s})^{-1}\cdot
			\right.\nonumber \\
		\left.
			\Wiener[\node_{k}]{i}{\node_{k},j,s}\Wiener[j]{\node_{k}}{j,i,s}
		\right\}^{-1}\Wiener[i]{j}{i,\node_{k},s}.
	\end{align}
	\normalsize
\end{thm}

\begin{proof}
Consider the graph, $G_{\overline{i\rightarrow j}}$, resulting from removing the link from $\node_i\rightarrow \node_j $ from graph $G.$ In the graph $G_{\overline{i\rightarrow j}}$, the paths connecting $i$ and $j$ are: $\node_j\rightarrow \node_k\rightarrow\node_i,\ \node_j\rightarrow \node_k\leftarrow \node_s\rightarrow\node_i,\ \node_j\leftarrow \node_s\rightarrow\node_k\rightarrow\node_i,\ \mbox{and, }\  \node_j\leftarrow \node_s\rightarrow \node_i$ which are all paths that are d-separated by the set $\{\node_s,\ \node_k\}.$ Moreover,  the set $J:=\{\node_m|dsep_{G_{\overline{i\rightarrow j}}}(\node_m,\{\node_s,\ \node_k\},\node_i)\}=\{\node_j\}.$  Note that $\node_j$ is not in any loop in $G$ restricted to $\node_J\cup \node_k=\node_j\cup\node_k.$ Thus, all the conditions of Lemma~\ref{lem:generalized single door} are satisfied and thus 
\begin{align}\label{eq:Hji}
		\adjTF_{\nodeindexalt\nodeindex}
		=\left[
				\identitymatrix
				-\Wiener[\node_{k}]{\node_{\nodeindexalt}}{(\node_{\nodeindex},\node_{k},\node_{s})}
				\adjTF_{kj}^{(f)}
			\right]^{-1}\Wiener[\node_{\nodeindex}]{\node_{\nodeindexalt}}{(\node_{\nodeindex},\node_{k},\node_{s})}.
	\end{align}
where 
	\begin{align*}
		\adjTF_{kj}^{(f)}=&
				\left(
					\identitymatrix-\left( \adjTF_{kk}+\adjTF_{k\bar J}(\identitymatrix-\adjTF_{\bar J \bar J})^{-1}\adjTF_{\bar Jk}
				\right)\right)^{-1}\cdot\\
				&\quad\cdot\left(
					\adjTF_{kj}+\adjTF_{k\bar{J}}(\identitymatrix-\adjTF_{\bar{J}\bar{J}})^{-1}\adjTF_{\bar{J}j}
				\right).
	\end{align*}
Note that $\bar{J}=\emptyset$ and there is no self-loop at $k$. Thus, $\adjTF_{kj}^{(f)}=\adjTF_{kj}.$
It follows in a similar manner that:
	\begin{align}\label{eq:Hkj}
		\adjTF_{kj}=
			\left[
				\identitymatrix
				-\Wiener[\node_{\nodeindex}]{\node_{k}}{(\node_{\nodeindexalt},\node_{\nodeindex},\node_{s})}
				\adjTF_{ik}
			\right]^{-1}\Wiener[\node_{\nodeindexalt}]{\node_{k}}{(\node_{\nodeindexalt},\node_{\nodeindex},\node_{s})}
	\end{align}
and 
	\begin{align}\label{eq:Hik}
		\adjTF_{ik}=
			\left[
				\identitymatrix
				-\Wiener[\node_{\nodeindexalt}]{\node_{\nodeindex}}{(\node_{k},\node_{\nodeindexalt},\node_{s})}
				\adjTF_{ji}
			\right]^{-1}\Wiener[\node_{k}]{\node_{\nodeindex}}{(\node_{k},\node_{\nodeindexalt},\node_{s})}.
	\end{align}
	By replacing  \eqref{eq:Hik} in  \eqref{eq:Hkj} and the resulting expression into \eqref{eq:Hji} we obtain the assertion using the matrix inversion lemma.
\end{proof}

\noindent{\bf Remark:}
Theorem~\ref{thm:revolvingDoor} represents only an instance of the application of Lemma~\ref{lem:generalized single door}. There are many results that can be derived on identification of  transfer functions involved in loops. For instance, Theorem~\ref{thm:revolvingDoor} can easily be extended to involve a closed-loop with multiple  intermediate nodes between $\node_i$ and $\node_j$ and not just $\node_K;$ the resulting analog of relation (\ref{eq:three floor equation}) will still be quadratic in $\adjTF_{21}.$ Furthermore, it is straightforward to generalize the results to include non-ancestors of the variables $\node_i,\ \node_j\, \node_k$ and $\node_s$ as there is no loss of generality, in the conclusions of Lemma~\ref{lem:generalized single door},  by  restricting the graph to the set of ancestors $U=an\{\node_i,\ \node_j\, \node_k,\ \node_s\}$.

\noindent{\bf Remark:}
Theorem~\ref{thm:revolvingDoor} provides an implicit relation for the transfer function
$\adjTF_{\nodeindexalt\nodeindex}$ that needs to be solved.
This relation is quadratic in the variable $\adjTF_{\nodeindexalt\nodeindex}$ providing in general two solutions.
The following example shows that, at least in some situations, imposing the additional condition that 
$\adjTF_{\nodeindexalt\nodeindex}$
is a causal transfer function helps to determine
$\adjTF_{\nodeindexalt\nodeindex}$
unequivocally.

\noindent{\bf Remark:}
The presence of the observable external signal $y_{s}$ is not required, since, contrary to other identification methods, Theorem~\ref{thm:revolvingDoor} can seamlessly deal with null transfer functions, as shown in the following example.

\normalsize 
\begin{example}[Closed loop identification with no knowledge of a transfer function]\label{ex:CLIDwithoutTF}
	Consider again the loop network in Figure~\ref{fig:loopID} where all nodes are measured and the dynamics is known to be causal.
	The objective is  to identify the transfer function, $\adjTF_{21}$.
    Note that the conditions for the application of Theorem~\ref{thm:revolvingDoor} hold.
	The following can  be determined from  power spectral densities:

	\scriptsize
	\begin{align*}
		\Wiener[\node_{2}]{\node_{3}}{(\node_{2},\node_{1})}
		=
		\frac{\beta\,\PSD_{\noise_{1}}}{\PSD_{\noise_{3}}\,\ZTvar\,{\gamma}^{2}+\Phi_{e_{1}}\,\ZTvar};
		&
		\quad
		\Wiener[\node_{1}]{\node_{3}}{(\node_{2},\node_{1})}
		=
		\frac{\PSD_{\noise_{3}}\,\gamma}{\PSD_{\noise_{3}}\,{\gamma}^{2}+\PSD_{e_{1}}}.
	\end{align*}
	\begin{align*}
		\Wiener[\node_{3}]{\node_{1}}{(\node_{3},\node_{2})}
		=
		\frac{\PSD_{\noise_{2}}\,{\ZTvar}^{2}\,\gamma}{\PSD_{\noise_{2}}\,{\ZTvar}^{2}+{\alpha}^{2}\,\PSD_{\noise_{1}}};
		&
		\quad
		\Wiener[\node_{2}]{\node_{1}}{(\node_{3},\node_{2})}
		=
		\frac{\alpha\,\PSD_{\noise_{1}}\,\ZTvar}{\PSD_{\noise_{2}}\,{\ZTvar}^{2}+{\alpha}^{2}\,\PSD_{\noise_{1}}}.
	\end{align*}
	\normalsize
	and can thus be determined from measured time-series data. Equation~(\ref{eq:three floor equation}) reduces to (\ref{eq:three floor equation 2}) described by:

    \scriptsize
\begin{align}\label{eq:three floor equation 2}
&\adjTF_{21}=\left[\alpha \ZTvar\left( \alpha\,\Phi_{e_{3}}\,\adjTF_{21}\ZTvar{\gamma}^{2}-{\alpha}^{2}\Phi_{e_{3}}{\gamma}^{2}-\Phi_{e_{2}}{\ZTvar}^{2}+\alpha \Phi_{e_{1}}\,\adjTF_{21}\ZTvar-{\alpha}^{2}\Phi_{e_{1}}\right)\right]\nonumber \\
/ &\left[\alpha\,\Phi_{e_{3}}\,\adjTF_{21}\,{\ZTvar}^{3}\,{\gamma}^{2}
	-{\alpha}^{2}\,\Phi_{e_{3}}\,{\ZTvar}^{2}\,{\gamma}^{2}
	+\alpha\,{\beta}^{2}\,\Phi_{e_{2}}\,\adjTF_{21}\,\ZTvar\,{\gamma}^{2}
	-{\alpha}^{2}\,{\beta}^{2}\,\Phi_{e_{2}}\,{\gamma}^{2}\right. \nonumber \\
	&\left. -\Phi_{e_{2}}\,{\ZTvar}^{4}
	+\alpha\,\Phi_{e_{1}}\,\adjTF_{21}\,{\ZTvar}^{3}
	-{\alpha}^{2}\,\Phi_{e_{1}}\,{\ZTvar}^{2}\right].
\end{align}
\normalsize
Relation (\ref{eq:three floor equation 2}) is quadratic in $\adjTF_{21}(z)$ and can be solved to yield two solutions described by:
\begin{align*}
	\adjTF_{21}=\frac{\alpha}{\ZTvar}; & \quad 
	\adjTF_{21}=\frac{{\alpha}^{2}\,\Phi_{e_{3}}\,\ZTvar\,{\gamma}^{2}+\Phi_{e_{2}}\,{\ZTvar}^{3}+{\alpha}^{2}\,\Phi_{e_{1}}\,\ZTvar}{\left( \alpha\,\Phi_{e_{3}}\,{\ZTvar}^{2}+\alpha\,{\beta}^{2}\,\Phi_{e_{2}}\right) \,{\gamma}^{2}+\alpha\,\Phi_{e_{1}}\,{\ZTvar}^{2}}.
\end{align*}
The second solution can  be discarded since it is not causal.
\end{example}

\section{Conclusion}
The article builds bridges between tools from probabilistic graphical models and systems theory.
As a first contribution, the article provides graphical criteria, based on $d$-separation, to determine the relevance of any set of nodes on the estimation of a specific node's activity. 
As a second contribution, the article lays a framework for the identification of a transfer function between a specified pair of network nodes in a network. 
Assuming a partially known topology of linear interactions between agents, the effectiveness of the methods rests on the ease of selecting a set of additional nodes whose activities, if used in a linear prediction, provide as a by-product a consistent estimate of the transfer between the specified pair of nodes.
Again, the set of nodes to be measured for identifying a specified transfer function is achieved using graphical criteria based on $d$-separation.
The results are applicable to topologies that include feedback loops, providing novel perspective on closed loop identification techniques as well. 
Furthermore, the article provides insights into notions of estimation and identification that are robust with respect against uncertainties on the network topology.

\section*{Acknowledgments}
Donatello Materassi acknowledges NSF for partially supporting this work (CAREER \#1553504).

\bibliography{./topident}

\section{Appendix}

\subsection{Proof of Lemma~\ref{lem:generalized single door}}
	The proof proceeds by manipulating the original LDIM in order to obtain an explicit expression for
	$\Wiener{\node_{\nodeindexalt}}{(\node_{\nodeindex},\node_{Z})}$.  First, marginalize all nodes that are not ancestors of $\node_i,\ \node_j$ and $\node_Z.$\\
	\noindent{\it Marginalizing $\nodesetsmallertocheckdsepcomplement:=\node\setminus \nodesetsmallertocheckdsep$:}\\
	Note that the set of vertices
	$y=\nodesetsmallertocheckdsep\cup \nodesetsmallertocheckdsepcomplement$.
	As
	$\nodesetsmallertocheckdsep$
	is an ancestor set, there can be no parents to any element of 
	$\nodesetsmallertocheckdsep$
	in
	$\nodesetsmallertocheckdsepcomplement$.
	Thus, directed links from
	$\nodesetsmallertocheckdsepcomplement$
	to 
	$\nodesetsmallertocheckdsep$
	are not possible, as depicted in Figure~\ref{fig:step}(a).
		\begin{figure}[ht!]
		\centering 
		\begin{tabular}{cc}
			\includegraphics[width=0.45\columnwidth]{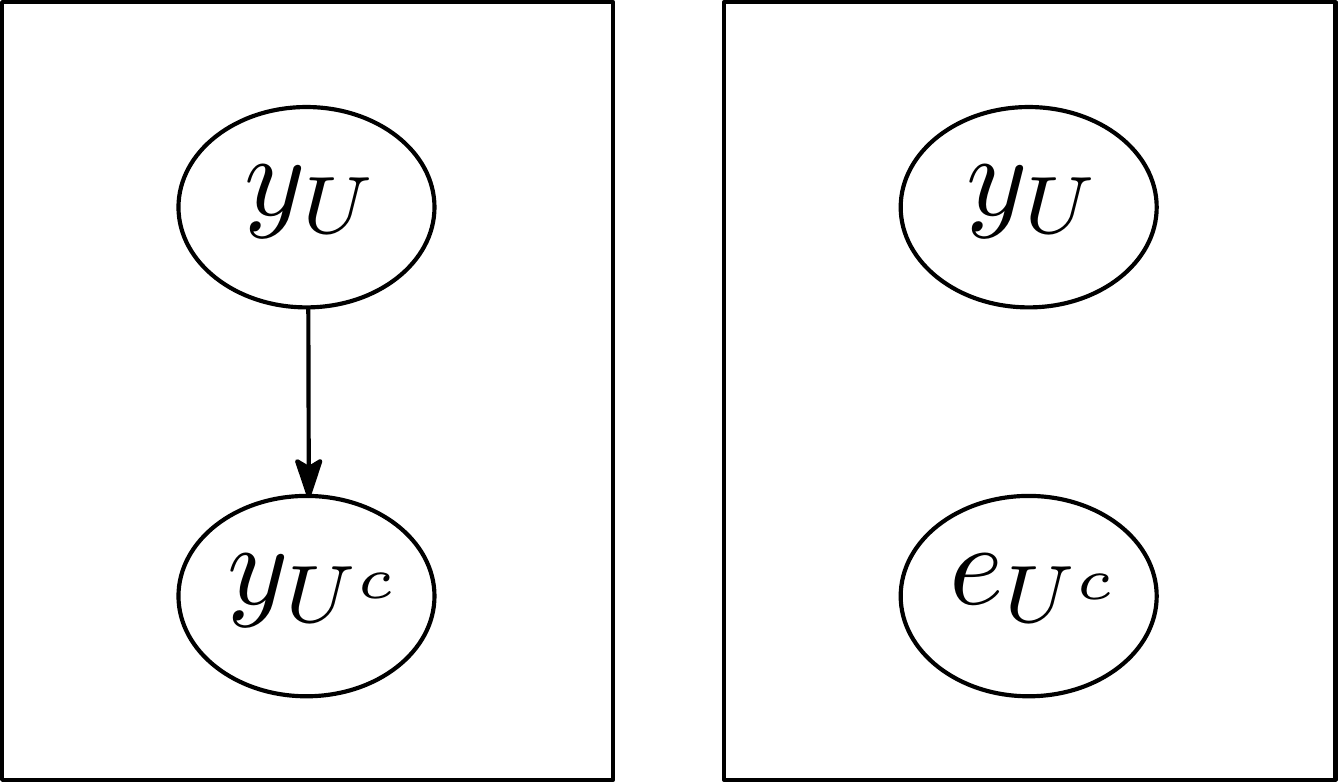} &
			\includegraphics[width=0.45\columnwidth]{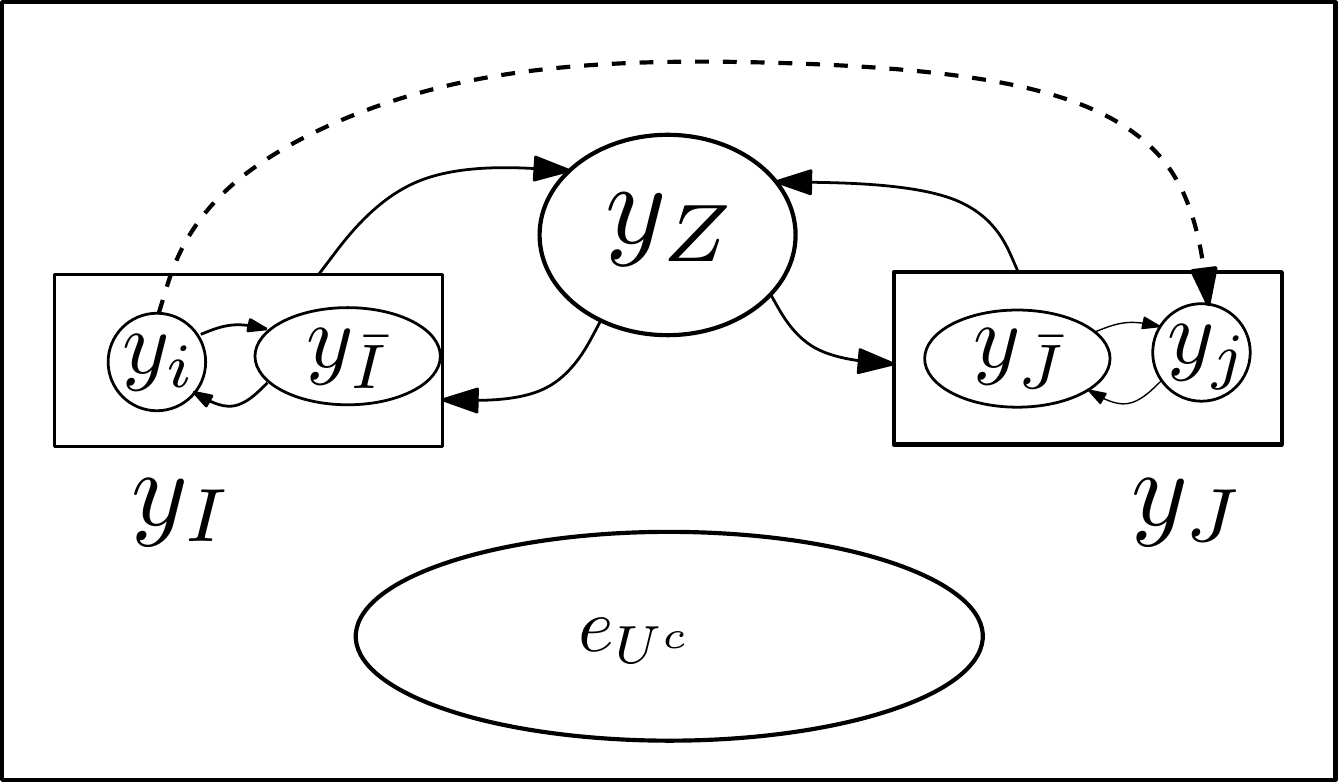}
			\\
			(a) \qquad\qquad\quad (b) & (c)\\
			\includegraphics[width=0.45\columnwidth]{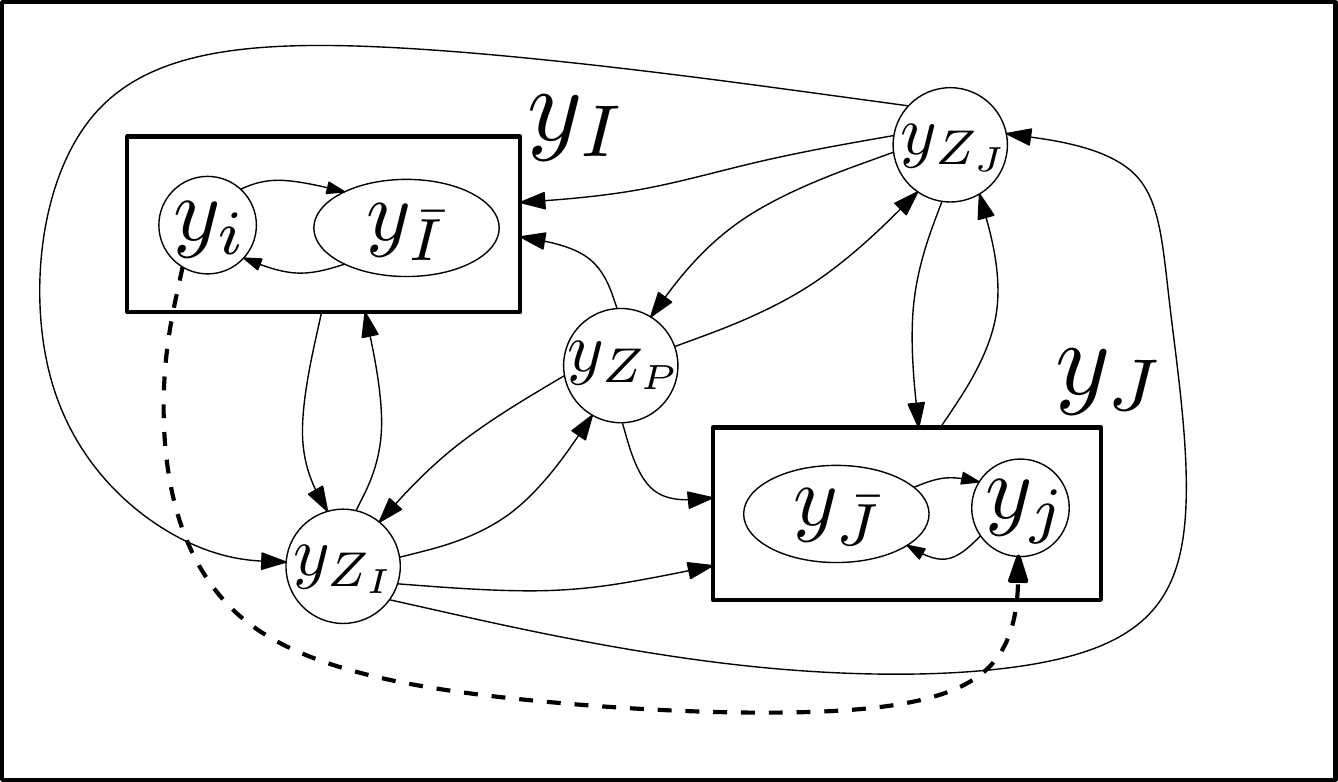} &
			\includegraphics[width=0.45\columnwidth]{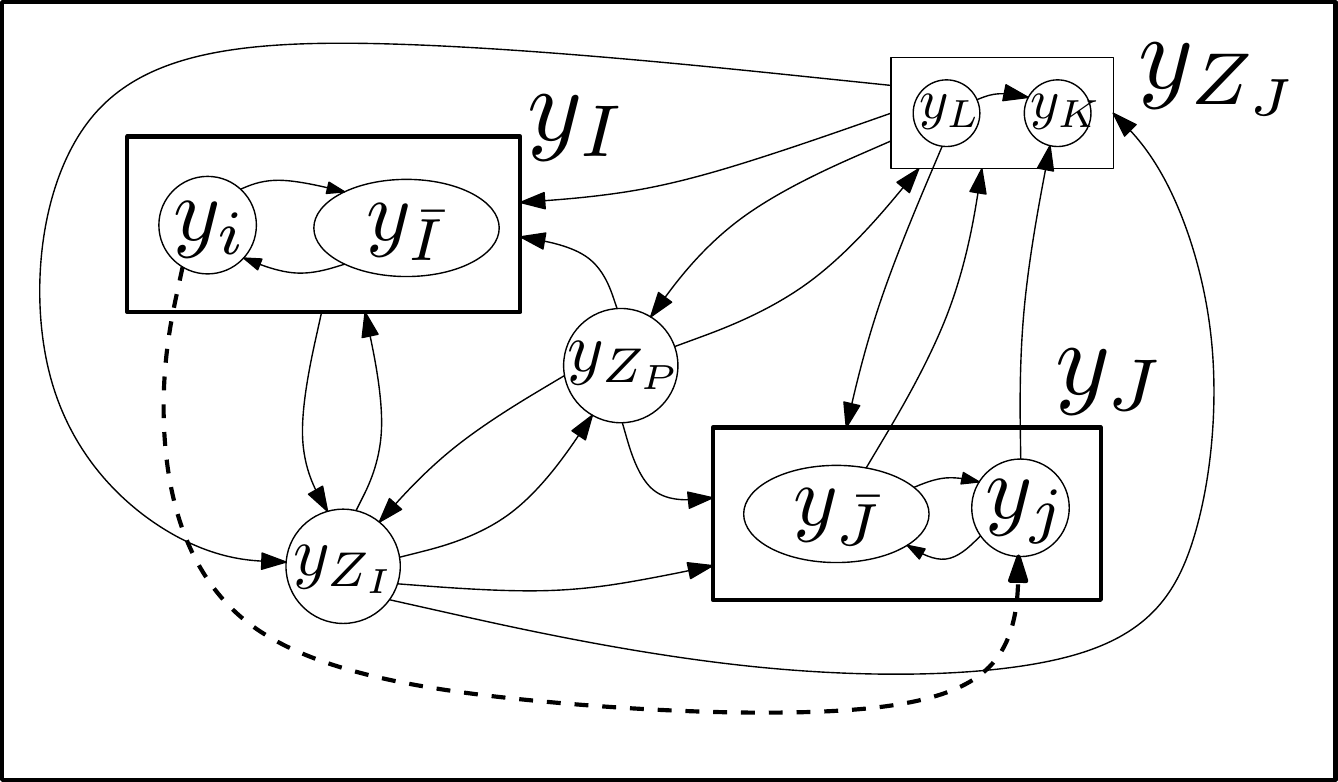}
			\\
			(d) & (e)\\
		\end{tabular}
		\caption{Steps in the proof of Lemma~\ref{lem:generalized single door}.
		\label{fig:step} \label{fig:step2} \label{fig:step3} \label{fig:step4}}
	\end{figure}
	Note, from Lemma~\ref{lem:node marginalization}, that, as there are no links from
	$\nodesetsmallertocheckdsepcomplement$
	to
	$\nodesetsmallertocheckdsep$
	in the original
	LDIM $\ldim$,
	it is not possible to have a chain from
	$\nodesetsmallertocheckdsep$
	to
	$\nodesetsmallertocheckdsep$
	with all intermediate links in
	$\nodesetsmallertocheckdsepcomplement$.
	It follows that the transfer matrix that maps elements in the set
	$\nodesetsmallertocheckdsep$
	to elements in the set
	$\nodesetsmallertocheckdsep$
	after marginalizing
	$\nodesetsmallertocheckdsepcomplement$,
	remains unaltered and is given by
	$\adjTF_{\nodesetsmallertocheckdsep\nodesetsmallertocheckdsep}$
	(see property \ref{item:tfUnchanged} of Lemma~\ref{lem:node marginalization}).
	Thus after marginalizing the set
	$\nodesetsmallertocheckdsepcomplement$,
	the LDIM in Figure~\ref{fig:step}(b) results.
	
	The marginalization of the LDIM, $\ldim$, with respect to
	$\nodesetsmallertocheckdsep^{c}$
	is complete. We further analyze the structure of the LDIM after marginalizing
	$\nodesetsmallertocheckdsepcomplement$.
	Let,
	\begin{align*}
		&\subsetnode:=\nodesetsmallertocheckdsep\setminus \{\subsetnodealt\cup \subsetnodesep\}; & \qquad &
		\subsetnodeenlarged :=\subsetnode\setminus \{\node_{\nodeindex}\}.
	\end{align*}
	As
	$dsep_{\mutilatedgraph}(\node_{\nodeindex},\subsetnodesep,\node_{\nodeindexalt})$
	holds,
	it follows from Theorem~\ref{lem:inferDsepOnEnlargedSets},
	that $dsep_{\mutilatedgraph}(\subsetnode,\subsetnodesep,\subsetnodealt)$
	holds. 
	Observe that
	$\nodesetsmallertocheckdsep=\subsetnode\cup\subsetnodesep\cup\subsetnodealt$
	and that
	$\subsetnode$,
	$\subsetnodealt$
	and
	$\subsetnodesep$
	are pairwise disjoint.
	Thus
	$\subsetnode$,
	$\subsetnodealt$
	and
	$\subsetnodesep$
	partition $U$.
	As in the graph
	$\mutilatedgraph$
	we have
	$dsep_{\mutilatedgraph}(\subsetnode,\subsetnodesep,\subsetnodealt)$,
	links of the forms
	$\node_{k}\to \node_{l}$,
	$\node_{l}\to\node_{k}$
	and $\node_{l}\rightarrow \node_{c}\leftarrow \node_{k}$
	for any
	$\node_{k}\in \subsetnode$,
	$\node_{l}\in \subsetnodealt$
	and
	$\node_{c}\in \subsetnodesep$
	are not allowed in $\mutilatedgraph$ (see Lemma~\ref{lem:linkTypeRestricted}).
	Thus Figure~\ref{fig:step}(c) follows where the set
	$\nodesetsmallertocheckdsep$
	is partitioned by
	$\subsetnode$,
	$\subsetnodealt$
	and
	$\subsetnodesep$,
	and there are  no direct links between 
	$\subsetnode$
	and
	$\subsetnodealt$
	in
	$\mutilatedgraph$.
	In  Figure~\ref{fig:step}(c) all solid links belong to
	$\mutilatedgraph$
	and the link
	$\node_{\nodeindex}\rightarrow \node_{\nodeindexalt}$
	shown as a dashed link is added to recover the graphical structure of all links.

	Define sets
	\begin{align*}
		&\subsetnodesepchildalt:=\child{\graph}{\subsetnodealt}\cap \subsetnodesep,\\
		&\subsetnodesepchild:=\child{\graph}{\subsetnode}\cap \subsetnodesep,~\text{and}\\
		&\subsetnodesep_{P}:=\subsetnodesep\setminus
		(\subsetnodesep_{J}\cup\subsetnodesep_{I}).
	\end{align*}
	Consider a link  $\node_\ell\rightarrow \node_c$ where $\node_\ell\in \node_I$ and $\node_c\in \node_{Z_J}$, present.  This implies that there exists $\node_p\in \node_{J}$ which is a parent of $\node_c.$ Thus the path 
	$\node_\ell \rightarrow \node_c\leftarrow \node_p$ exists with $\node_\ell \in \node_I,\ \node_c\in \node_{Z_J}\subset \node_Z$ and $\node_p\in \node_J$. This is a contradiction as $\node_I,\ \node_Z$ and $\node_J$ partition $\node_U$ and $dsep_{\bar{G}} (\node_I,\node_Z,\node_J)$ holds (see Lemma~\ref{lem:linkTypeRestricted}). Thus there is no link from  $\node_I$ to $\node_{Z_J}$ in $\bar{G}.$ Moreover, from the definition of $\node_{Z_P}$, there can be no link from $\node_I$ to $\node_{Z_P}$ in $\bar{G}.$ For similar reasons there can be no links from $\node_J$ to $\node_{Z_I}$ or to $\node_{Z_P}$ in $\bar{G}.$
	  Furthermore note that as,   $\node_K=de(\node_j)\cap \node_{Z_J},$ it follows that, $\node_K\subset \node_{Z_J}$. Let $\node_L:=\node_{Z_J}\backslash \node_K$. 
	 Suppose there is a link from  $\node_k\in \node_K$ to $\node_l\in \node_L$. Then $\node_l\in Z_J$ and $\node_l\in de(\node_j)$ implying $\node_l\in \node_K$. This is a contradiction as $\node_L$ and $\node_K$ are disjoint. Thus there can be no link from $\node_K$ to $\node_L$. For similar reasons there can be no link from $\node_j$ to $\node_L.$ 
	 
	 Also, from assumption A3. there are no links of the form $\node_{k} \rightarrow \node_j$ for $\node_k\in \node_K. $ Moreover, $\node_K$ and $\node_L$ form partition of $\node_{Z_J}.$ Restricting attention to $\node_{U}$ and introducing the link between $\node_i$ and $\node_j$ (the dashed line) 	Figure~\ref{fig:step}(d) follows. Also, note that none of the transfer functions are altered; thus the transfer matrix mapping $\node_U$ to $\node_U$ remains $\adjTF_{UU}.$. Thus the the sparsity pattern  depicted by Figure~\ref{fig:step}(d) along with the LDIM of Figure~\ref{fig:bigeq1}.
	 \begin{figure*}[!h]
	 {
	 \small
	 \begin{align*}
	 \qquad\qquad\qquad\qquad
	 \left(\begin{array}{c}
	 \node_i \\ 
	 \node_{\bar{I}} \\ 
	 \node_{Z_P} \\ 
	 \node_{Z_I} \\ 
	 \node_K \\ 
	 \node_L\\
	 \node_j \\ 
	 \node_{\bar{J}}
	 \end{array}\right)= \left(\begin{array}{cccccccc}
	 \adjTF_{ii} & \adjTF_{i\bar{I}} & \adjTF_{iZ_P} & \adjTF_{iZ_I} & \adjTF_{iK} &\adjTF_{iL} & 0 & 0 \\ 
	 \adjTF_{\bar{I},i} & \adjTF_{\bar{I}\bar{I}} & \adjTF_{\bar{I}Z_P} & \adjTF_{\bar{I}Z_I} & \adjTF_{\bar{I}K} & \adjTF_{\bar{I}L} &0 & 0 \\ 
	 0 & 0 & \adjTF_{Z_PZ_p} & \adjTF_{Z_PZ_I} & \adjTF_{Z_PK} & \adjTF_{Z_PL} &0 & 0 \\ 
	 \adjTF_{Z_Ii} & \adjTF_{Z_I \bar{I}} & \adjTF_{Z_I Z_P} & \adjTF_{Z_IZ_I} & \adjTF_{Z_IK} &\adjTF_{Z_IL} & 0 & 0 \\ 
	 0 & 0 & \adjTF_{KZ_P} & \adjTF_{K Z_I} & \adjTF_{KK}  & \adjTF_{KL}  &\adjTF_{Kj} & \adjTF_{K\bar{J}} \\ 
	  0 & 0 & \adjTF_{LZ_P} & \adjTF_{L Z_I} & 0  & \adjTF_{LL}  &0 & \adjTF_{L\bar{J}} \\
	 \adjTF_{ji} & 0 & \adjTF_{jZ_P} & \adjTF_{jZ_I} & 0 &  \adjTF_{jL} & 0 & \adjTF_{j\bar{J}} \\ 
	 0 & 0 & \adjTF_{\bar{J}Z_P} & \adjTF_{\bar{J} Z_I} & \adjTF_{\bar{J}K} &\adjTF_{\bar{J}L} & \adjTF_{\bar{J}j} & \adjTF_{\bar{J}\bar{J}}
	 \end{array} \right)\left(\begin{array}{c}
	 \node_i \\ 
	 \node_{\bar{I}} \\ 
	 \node_{Z_P} \\ 
	 \node_{Z_I} \\ 
	 \node_K \\ 
	 \node_L\\
	 \node_j \\ 
	 \node_{\bar{J}}
	 \end{array}\right)+\left(\begin{array}{c}
	 e_i \\ 
	 e_{\bar{I}} \\ 
	 e_{Z_P} \\ 
	 e_{Z_I} \\ 
	 e_{K} \\ 
	 e_L\\
	 e_j \\ 
	 e_{\bar{J}}
	 \end{array}\right).
	 \end{align*}
	 }
	 \caption{
		LDIM obtained after reaching the step represented in Figure~\ref{fig:step}(d).
		\label{fig:bigeq1}
	 }
	 \end{figure*}
	\normalsize
	
	\noindent{\it Marginalizing set $\subsetnodealtenlarged$:}\\
	Consider marginalization with respect to $\node_{\bar{J}}$, given the LDIM in Figure~\ref{fig:step}(d). Consider any two elements, $r $ and $s$ in $ U\backslash (\node_{\bar{J}}\cup \node_j\cup \node_{Z_J})=\node_i\cup \node_{\bar{I}}\cup\node_{Z_P}\cup\node_{Z_I}.$ There is no chain of the form $r\rightarrow\node_{\pi_1}\rightarrow \node_{\pi_2}\rightarrow \ldots \rightarrow \node_{\pi_t}\rightarrow s$ with  $\node_{\pi_m} \in \node_{\bar{J}}$ for $m=1,\ldots,t.$ From property \ref{item:tfUnchanged}) of Lemma~\ref{lem:node marginalization}, it follows that the transfer matrix  between $\node_i\cup \node_{\bar{I}}\cup\node_{Z_P}\cup\node_{Z_I}$  and itself remain unchanged in the LDIM resulting after marginalizing $\node_{\bar{J}}$.
	Also there is a no chain of the form $\node_{\pi_1}\rightarrow \ldots \node_{\pi_\ell}\rightarrow s$ with $\node_{\pi_m}\in \node_{\bar{J}}$ for any $s\in \node_i\cup \node_{\bar{I}}\cup\node_{Z_P}\cup\node_{Z_I}$. It follows from property \ref{item:Hji'=0}) of Lemma~\ref{lem:node marginalization} that the transfer matrix from $\node_{\bar{J}}$ to any $s\in \node_i\cup \node_{\bar{I}}\cup\node_{Z_P}\cup\node_{Z_I}$ remains unchanged and thus remains zero. 
	Moreover, for any $s\in\node_i\cup \node_{\bar{I}}\cup\node_{Z_P}\cup\node_{Z_I}$ there is no chain of the form $\node_j\rightarrow \node_{\pi_1}\rightarrow \ldots \node_{\pi_\ell}\rightarrow s$ with $\node_{\pi_m}\in \node_{\bar{J}}.$ Thus all transfer matrices mapping $\node_j$ into $\node_s$ remain unchanged. Similarly there are no chains contained in $\node_{\bar{J}}$ that connect $\node_i$ to $\node_{Z_J}$, $\node_{\bar{I}}$ to $\node_{Z_J}$, $\node_i$ to $\node_j$ and $\node_{\bar{I}}$ to $\node_j$ and thus all corresponding transfer functions matrices remain unchanged. 

	Assumption A2. precludes a chain connecting $\node_j$ to $\node_j$ with all intermediate nodes in $\node_{\bar{J}}$; thus the map from $j$ to $j$ remains unchanged and is zero (from assumption A2. the original LDIM does not have a self-loop at $j$ and this $\adjTF_{jj}=0.$). Also, from Assumption A3. there is no chain connecting any element of $\node_K$ to $\node_j$ which is contained in $\bar{J}.$ Thus the transfer matrix from $\node_K$ to $\node_j$ remains at zero after marginalizing $\node_{\bar{J}}.$

	After marginalizing $\node_{\bar{J}}$ the LDIM  is described by the equation in Figure~\ref{fig:bigeq2}
	\begin{figure*}[!h]
	{\tiny
	 \begin{align*}\begin{array}{ccc}
	\left(\begin{array}{c}
	 \node_i \\ 
	 \node_{\bar{I}} \\ 
	 \node_{Z_P} \\ 
	 \node_{Z_I} \\ 
	 \node_{K} \\ 
	 \node_{L} \\
	 \node_j \\ 
	 e_{\bar{J}}
	 \end{array}\right)= \left(\begin{array}{cccccccc}
	 \adjTF_{ii} & \adjTF_{i\bar{I}} & \adjTF_{iZ_P} & \adjTF_{iZ_I} & \adjTF_{iK} & \adjTF_{iL} & 0 & 0 \\ 
	 \adjTF_{\bar{I},i} & \adjTF_{\bar{I}\bar{I}} & \adjTF_{\bar{I}Z_P} & \adjTF_{\bar{I}Z_I} & \adjTF_{\bar{I}K} & \adjTF_{\bar{I}L} & 0 & 0 \\ 
	 0 & 0 & \adjTF_{Z_PZ_p} & \adjTF_{Z_PZ_I} & \adjTF_{Z_PK} &\adjTF_{Z_PL} & 0 & 0 \\ 
	 \adjTF_{Z_Ii} & \adjTF_{Z_I \bar{I}} & \adjTF_{Z_I Z_P} & \adjTF_{Z_IZ_I} & \adjTF_{Z_IK} &\adjTF_{Z_IL} & 0 & 0 \\ 
	 0 & 0 & \adjTF_{KZ_P} +\adjTF_{K\bar{J}}\adjTF_{\bar{J}Z_P}^{(e)}& \adjTF_{K Z_I}+\adjTF_{K\bar{J}}\adjTF_{\bar{J}Z_I}^{(e)} & \adjTF_{K K}+\adjTF_{K\bar{J}}\adjTF_{\bar{J}K}^{(e)}  &  \adjTF_{K L}+\adjTF_{K\bar{J}}\adjTF_{\bar{J}L}^{(e)}  &  \adjTF_{Kj} +\adjTF_{K\bar{J}}\adjTF_{\bar{J}j}^{(e)}& \adjTF_{K\bar{J}}^{(e)} \\ 
	  0 & 0 & \adjTF_{LZ_P} +\adjTF_{L\bar{J}}\adjTF_{\bar{J}Z_P}^{(e)}& \adjTF_{L Z_I}+\adjTF_{L\bar{J}}\adjTF_{\bar{J}Z_I}^{(e)} & 0  &  \adjTF_{L L}+\adjTF_{L\bar{J}}\adjTF_{\bar{J}L}^{(e)}  &  0& \adjTF_{L\bar{J}}^{(e)} \\ 
	 \adjTF_{ji} & 0 & \adjTF_{jZ_P}+\adjTF_{j\bar{J}}\adjTF_{\bar{J}Z_P}^{(e)} & \adjTF_{jZ_I}+\adjTF_{j\bar{J}}\adjTF_{\bar{J}Z_I}^{(e)}  &0 & \adjTF_{jL} +\adjTF_{j\bar{J}}\adjTF_{\bar{J}L}^{(e)} & 0 & \adjTF_{j\bar{J}}^{(e)} \\ 
	 0 & 0 & 0 & 0& 0 & 0 & 0&0
	 \end{array} \right)\left(\begin{array}{c}
	 \node_i \\ 
	 \node_{\bar{I}} \\ 
	 \node_{Z_P} \\ 
	 \node_{Z_I} \\ 
	 \node_K \\ 
	  \node_L\\ 
	 \node_j \\ 
	 e_{\bar{J}}
	 \end{array}\right)\\ \qquad\qquad+\left(\begin{array}{c}
	 e_i \\ 
	 e_{\bar{I}} \\ 
	 e_{Z_P} \\ 
	 e_{Z_I} \\ 
	 e_K \\ 
	  e_L \\ 
	 e_j \\ 
	 e_{\bar{J}}
	 \end{array}\right)
	 \end{array}
	 \end{align*}
	 }
	 \caption{
		LDIM after the marginalization of $\node_{\bar{J}}$.
	 \label{fig:bigeq2}
	 }
	\end{figure*}
	 where 
	 \begin{align}\label{eq:HhatE}
	 \begin{array}{ccc}
	 \adjTF_{\bar{J}m}^{(e)}&=&(\identitymatrix -\adjTF_{\bar{J}\bar{J}})^{-1}\adjTF_{\bar{J}m}	\mbox{ and }\\
	   \adjTF_{m\bar{J}}^{(e)}&=&\adjTF_{m\bar{J}}(\identitymatrix -\adjTF_{\bar{J}\bar{J}})^{-1}.
	  \end{array}
	  \end{align}
	This completes the marginalization with respect to the set $\node_{\bar{J}}.$
	 
	Removing self-loops at $\node_K$ and $\node_L$ and subsuming the effect of $e_{\bar{J}}$ into corresponding noise terms we have the relation described in Figure~\ref{fig:bigeq3}
	\begin{figure*}[!h]
	{\scriptsize
	\begin{equation}\label{eq:nonLdim}
	\begin{array}{ccc}
	\left(\begin{array}{c}
	\node_i \\ 
	\node_{\bar{I}} \\ 
	\node_{Z_P} \\ 
	\node_{Z_I} \\ 
	\node_{K} \\ 
	\node_{L} \\
	\node_j 
	\end{array}\right)&=& \left(\begin{array}{ccccccc}
	\adjTF_{ii} & \adjTF_{i\bar{I}} & \adjTF_{iZ_P} & \adjTF_{iZ_I} & \adjTF_{iK} & \adjTF_{iL} & 0  \\ 
	\adjTF_{\bar{I},i} & \adjTF_{\bar{I}\bar{I}} & \adjTF_{\bar{I}Z_P} & \adjTF_{\bar{I}Z_I} & \adjTF_{\bar{I}K} & \adjTF_{\bar{I}L} & 0 \\ 
	0 & 0 & \adjTF_{Z_PZ_p} & \adjTF_{Z_PZ_I} & \adjTF_{Z_PK} &\adjTF_{Z_PL} & 0 \\ 
	\adjTF_{Z_Ii} & \adjTF_{Z_I \bar{I}} & \adjTF_{Z_I Z_P} & \adjTF_{Z_IZ_I} & \adjTF_{Z_IK} &\adjTF_{Z_IL} & 0 \\ 
	0 & 0 & \adjTF_{KZ_P}^{(f)} & \adjTF_{K Z_I}^{(f)} &0  &  \adjTF_{K L}^{(f)}  &  \adjTF_{Kj}^{(f)} \\ 
	0 & 0 & \adjTF_{LZ_P}^{(f)}& \adjTF_{L Z_I}^{(f)} & 0  &  0 &  0\\ 
	\adjTF_{ji} & 0 & \adjTF_{jZ_P}^{(f)} & \adjTF_{jZ_I}^{(f)}  &0 & \adjTF_{jL}^{(f)}& 0 
	\end{array} \right)\left(\begin{array}{c}
	\node_i \\ 
	\node_{\bar{I}} \\ 
	\node_{Z_P} \\ 
	\node_{Z_I} \\ 
	\node_K \\ 
	\node_L\\ 
	\node_j 
	\end{array}\right)+\left(\begin{array}{c}
	e_i \\ 
	e_{\bar{I}} \\ 
	e_{Z_P} \\ 
	e_{Z_I} \\ 
	\epsilon_K-\adjTF^{(f)}_{Kj}\epsilon_j \\ 
	\epsilon_L \\ 
	\epsilon_j
	\end{array}\right).
	\end{array}
	\end{equation}
	}
	\caption{
	Relation obtained after removing self-loops at $\node_K$ and $\node_L$ and subsuming the effect of $e_{\bar{J}}$ into corresponding noise terms. This relation does not represent a LDIM because the noise vector does not have a block diagonal power spectral density.
		\label{fig:bigeq3}
	}
	\end{figure*}
	where
	\begin{itemize}
	\item 
	$\adjTF_{K\alpha}^{(f)}=[\identitymatrix-(\adjTF_{KK}+\adjTF_{K\bar{J}}\adjTF_{\bar{J}K}^{(e)})]^{-1}[\adjTF_{K\alpha}+\adjTF_{K\bar{J}}\adjTF^{(e)}_{\bar{J}\alpha}]$ for $\alpha \in \{Z_P, Z_I,\ L,\ j\}, $  $\adjTF_{K\bar{J}}^{(f)}=[\identitymatrix-(\adjTF_{KK}+\adjTF_{K\bar{J}}\adjTF_{\bar{J}K}^{(e)})]^{-1}\adjTF^{(e)}_{K\bar{J}}]$, with $e_K^{(f)}=[\identitymatrix-(\adjTF_{KK}+\adjTF_{K\bar{J}}\adjTF_{\bar{J}K}^{(e)})]^{-1}e_K$ and $\epsilon_K=e_K^{(f)}+\adjTF_{K\bar{J}}^{(f)}e_{\bar{J}}+\adjTF_{Kj}^{(f)}\epsilon_j.$
 
	\item $\adjTF_{L\alpha}^{(f)}=[\identitymatrix-(\adjTF_{LL}+\adjTF_{L\bar{J}}\adjTF_{\bar{J}L}^{(e)})]^{-1}[\adjTF_{L\alpha}+\adjTF_{L\bar{J}}\adjTF^{(e)}_{\bar{J}\alpha}]$ for $\alpha \in \{Z_P, Z_I\}, $  $\adjTF_{L\bar{J}}^{(f)}=[\identitymatrix-(\adjTF_{LL}+\adjTF_{L\bar{J}}\adjTF_{\bar{J}L}^{(e)})]^{-1}\adjTF^{(e)}_{L\bar{J}}]$, with $e_L^{(f)}=[\identitymatrix-(\adjTF_{LL}+\adjTF_{L\bar{J}}\adjTF_{\bar{J}L}^{(e)})]^{-1}e_L,$ and $\epsilon_L=e_L^{(f)}+\adjTF_{L\bar{J}}^{(f)}e_{\bar{J}}$
	
	\item $\adjTF_{j\alpha}^{(f)}=\adjTF_{j\alpha}+\adjTF_{j\bar{J}}\adjTF_{\bar{J}\alpha}^{(e)}$ for $\alpha \in \{Z_P,\ Z_{I}, L\}$ and $\adjTF_{j\bar{J}}^{(f)}=\adjTF_{j\bar{J}}^{(e)},$ and $\epsilon_j=e_j+\adjTF_{j\bar{J}}^{(f)}e_{\bar{J}}.$
	 \end{itemize}
	 
	Note that (\ref{eq:nonLdim}) is not in the LDIM form because the noise processes are no longer independent.  The relation of $\epsilon_j$ from the last row of (\ref{eq:nonLdim}) can be substituted in the row corresponding $\node_K$ to eliminate $\epsilon_j$ from the row to obtain the relation in Figure~\ref{fig:bigeq4}.
	\begin{figure*}[!h]
	   {\tiny
	 \begin{equation}\label{eq:nonLdimFinal}
	 \begin{array}{ccc}
	\left(\begin{array}{c}
	 \node_i \\ 
	 \node_{\bar{I}} \\ 
	 \node_{Z_P} \\ 
	 \node_{Z_I} \\ 
	 \node_{K} \\ 
	 \node_{L} \\
	 \node_j 
	 \end{array}\right)&=& \left(\begin{array}{ccccccc}
	 \adjTF_{ii} & \adjTF_{i\bar{I}} & \adjTF_{iZ_P} & \adjTF_{iZ_I} & \adjTF_{iK} & \adjTF_{iL} & 0  \\ 
	 \adjTF_{\bar{I},i} & \adjTF_{\bar{I}\bar{I}} & \adjTF_{\bar{I}Z_P} & \adjTF_{\bar{I}Z_I} & \adjTF_{\bar{I}K} & \adjTF_{\bar{I}L} & 0 \\ 
	 0 & 0 & \adjTF_{Z_PZ_p} & \adjTF_{Z_PZ_I} & \adjTF_{Z_PK} &\adjTF_{Z_PL} & 0 \\ 
	 \adjTF_{Z_Ii} & \adjTF_{Z_I \bar{I}} & \adjTF_{Z_I Z_P} & \adjTF_{Z_IZ_I} & \adjTF_{Z_IK} &\adjTF_{Z_IL} & 0 \\ 
	  \adjTF^{(f)}_{Kj}\adjTF_{ji}& 0 & \adjTF_{KZ_P}^{(f)} + \adjTF^{(f)}_{Kj}\adjTF_{jZ_P}^{(f)} & \adjTF_{K Z_I}^{(f)} +\adjTF^{(f)}_{Kj}\adjTF_{jZ_I}^{(f)}&0  &  \adjTF_{K L}^{(f)} +\adjTF^{(f)}_{Kj}\adjTF_{jL}^{(f)} &  0 \\ 
	  0 & 0 & \adjTF_{LZ_P}^{(f)}& \adjTF_{L Z_I}^{(f)} & 0  &  0 &  0\\ 
	 \adjTF_{ji} & 0 & \adjTF_{jZ_P}^{(f)} & \adjTF_{jZ_I}^{(f)}  &0 & \adjTF_{jL}^{(f)}& 0 
	 \end{array} \right)\left(\begin{array}{c}
	 \node_i \\ 
	 \node_{\bar{I}} \\ 
	 \node_{Z_P} \\ 
	 \node_{Z_I} \\ 
	 \node_K \\ 
	  \node_L\\ 
	 \node_j 
	 \end{array}\right)+\left(\begin{array}{c}
	 e_i \\ 
	 e_{\bar{I}} \\ 
	 e_{Z_P} \\ 
	 e_{Z_I} \\ 
	 \epsilon_K \\ 
	  \epsilon_L \\ 
	 \epsilon_j
	 \end{array}\right).
	 \end{array} \end{equation}
	 }	 
	 \caption{
	Relation obtained after substituting $\node_{j}$ from the last row of (\ref{eq:nonLdim}) in the row corresponding $\node_K$.
	 \label{fig:bigeq4}
	 }
	\end{figure*}
	\normalsize
	Let us compute the Wiener estimate
	$\hat \node_{\nodeindexalt}$
	of
	$\node_{\nodeindexalt}$
	using the processes
	$\{\node_{\nodeindex}\}\cup\node_{Z}$.
	We have that
	\begin{align*}
		&\hat \node_{\nodeindexalt}
		:=\arg\min_{q\in Q}\|\node_{\nodeindexalt}-q\|\\
		&\text{s.t. } Q=tfspan\{\node_{\nodeindex},\node_{Z}\}.
	\end{align*}
	Since $\node_{\bar{I}}$ is $d$-separated from $\node_{j}$ by $\node_Z\cup\{\node_i\}$ (see Figure~\ref{fig:step}(d)), it follows from Theorem~\ref{thm:non-causal Wiener-separation} that $\node_Z\cup\{\node_i\}$ Wiener separates $\node_{\bar{I}}$ and $\node_j.$  we have that
	\begin{align*}
		&\hat \node_{\nodeindexalt}
		=\arg\min_{q\in Q}\|\node_{\nodeindexalt}-q\|\\
		&\text{s.t. } Q=tfspan\{\node_{\nodeindex},\node_{\bar I},\node_{Z}\}.
	\end{align*}
	Using (\ref{eq:nonLdimFinal}) we have:
	\begin{align*}
		\hat \node_{\nodeindexalt}
		&=\adjTF_{\nodeindexalt\nodeindex}\node_{\nodeindex}
			+\adjTF_{\nodeindexalt{Z_{P}}}^{(f)}\node_{Z_{P}}
			+\adjTF_{\nodeindexalt{Z_{I}}}^{(f)}\node_{Z_{I}}\\
			&\qquad +\adjTF_{\nodeindexalt{L}}^{(f)}\node_{L}
			+\arg\min_{q\in Q}\|\epsilon_j-q\|\\
		&\text{s.t. } Q=tfspan\{\node_{\nodeindex},\node_{\bar I},\node_{Z_{P}},\node_{Z_{I}},\node_{L},\node_{K}\}.
	\end{align*}
	From (\ref{eq:nonLdimFinal}) we have that 
    $tfspan\{\node_{\nodeindex},\node_{\bar I},\node_{Z_{P}},\node_{Z_{I}},\node_{L},\node_{K}\}
    =tfspan\{e_i,e_{\bar{I}},e_{Z_P},e_{Z_I},\epsilon_K,\epsilon_L\}$.

	We also note that processes $\{e_i,e_{\bar{I}},e_{Z_P},e_{Z_I}\}$ are independent of $\epsilon_j=e_j+\adjTF_{j\bar{J}}^{(f)}e_{\bar{J}}.$
	Thus the optimization problem
	\begin{align*}
	\begin{array}{ccc}
		\begin{array}{c}
			\arg\min_{q\in Q}\|\epsilon_j-q\|\\
			\text{s.t. } Q=tfspan\{e_i,e_{\bar{I}},e_{Z_P},e_{Z_I},\epsilon_K,\epsilon_L\}
		\end{array}
	\end{array}.
	\end{align*}
	is equivalent to the optimization problem 
	\begin{align*}
	\begin{array}{ccc}
		\begin{array}{c}
			\arg\min_{q\in Q}\|\epsilon_j-q\|\\
		\text{s.t. } Q=tfspan\{\epsilon_K,\epsilon_L\}
		\end{array}.
	\end{array}
\end{align*}
	Thus
\[\begin{array}{ccc}
		\hat \node_{\nodeindexalt}
		&=\adjTF_{\nodeindexalt\nodeindex}\node_{\nodeindex}
			+\adjTF_{\nodeindexalt{Z_{P}}}^{(f)}\node_{Z_{P}}
			+\adjTF_{\nodeindexalt{Z_{I}}}^{(f)}\node_{Z_{I}}+\adjTF_{\nodeindexalt{L}}^{(f)}\node_{L}+\\
			&+\arg\min_{q\in Q}\|\epsilon_j-q\|\\
		&\text{s.t. } Q=tfspan\{\node_{\nodeindex},\node_{\bar I},\node_{Z_{P}},\node_{Z_{I}},\node_{L},\node_{K}\}
		\end{array}\]
		leads to
\[\begin{array}{ccc}
		\hat \node_{\nodeindexalt}
		&=\adjTF_{\nodeindexalt\nodeindex}\node_{\nodeindex}
			+\adjTF_{\nodeindexalt{Z_{P}}}^{(f)}\node_{Z_{P}}
			+\adjTF_{\nodeindexalt{Z_{I}}}^{(f)}\node_{Z_{I}}+\adjTF_{\nodeindexalt{L}}^{(f)}\node_{L}\\
			&+\arg\min_{q\in Q}\|\epsilon_j-q\|\\
			&\text{s.t. } Q=tfspan\{\epsilon_K,\epsilon_L\}.
	\end{array}\]
	Thus
	\begin{align}\label{eq:projYjIni}
	\hat{y}_j=\adjTF_{\nodeindexalt\nodeindex}\node_{\nodeindex}
			+\adjTF_{\nodeindexalt{Z_{P}}}^{(f)}\node_{Z_{P}}
			+\adjTF_{\nodeindexalt{Z_{I}}}^{(f)}\node_{Z_{I}}+\adjTF_{\nodeindexalt{L}}^{(f)}\node_{L} \nonumber \\
			+W_{\epsilon_j,[\epsilon_K]|(\epsilon_K,\epsilon_L)}\epsilon_K+W_{\epsilon_j,[\epsilon_L]|(\epsilon_K,\epsilon_L)}\epsilon_L
	\end{align}
	Note from (\ref{eq:nonLdimFinal}) that 
	\begin{align}\label{eq:epsInTermsOfOrginals}
		&\epsilon_K=\node_K-\adjTF_{Kj}^{(f)}\nonumber \adjTF_{ji}\node_i-[\adjTF_{KZ_P}^{(f)}+\adjTF_{Kj}^{(f)}\adjTF_{jZ_P}^{(f)}]+\node_{Z_P}\\	&\qquad+[\adjTF_{KZ_I}^{(f)}+\adjTF_{Kj}^{(f)}\adjTF_{jZ_I}^{(f)}]\node_{Z_I}-[\adjTF_{KL}^{(f)}+\adjTF_{Kj}^{(f)}\adjTF_{jL}^{(f)}]\node_L\\ 
		&\epsilon_L=\node_L-\adjTF_{LZ_P}^{(f)}\node_{Z_P}-\adjTF_{LZ_I}^{(f)}\node_{Z_I}\nonumber
	\end{align}
	Substituting, (\ref{eq:epsInTermsOfOrginals}) in (\ref{eq:projYjIni}) it follows that 
	\begin{equation}
	\begin{array}{lll}
	\hat{y}_j=&[\identitymatrix-W_{\epsilon_j,[\epsilon_K]|(\epsilon_K,\epsilon_L)}\adjTF_{Kj}^{(f)}]\adjTF_{ji}\node_i+\\
	&+[W_{\epsilon_j,[\epsilon_K]|(\epsilon_K,\epsilon_L)}]\node_K\\
	&+[\adjTF_{jZ_P}^{(f)}-W_{\epsilon_j,[\epsilon_K]|(\epsilon_K,\epsilon_L)}(\adjTF_{KZ_P}^{(f)}+\adjTF_{Kj}^{(f)}\adjTF_{jZ_P}^{(f)})\\
	&-W_{\epsilon_j,[\epsilon_L]|(\epsilon_K,\epsilon_L)}\adjTF_{LZ_P}^{(f)}]\node_{Z_P}\\
	&+[\adjTF_{jZ_I}^{(f)}-W_{\epsilon_j,[\epsilon_K]|(\epsilon_K,\epsilon_L)}(\adjTF_{KZ_I}^{(f)}+\adjTF_{Kj}^{(f)}\adjTF_{jZ_I}^{(f)})\\
	&-W_{\epsilon_j,[\epsilon_L]|(\epsilon_K,\epsilon_L)}\adjTF_{LZ_I}^{(f)}]\node_{Z_I}\\
	&+[\adjTF_{jL}^{(f)}-W_{\epsilon_j,[\epsilon_K]|(\epsilon_K,\epsilon_L)}(\adjTF_{KL}^{(f)}+\adjTF_{Kj}^{(f)}\adjTF_{jL}^{(f)})\\
	&\qquad+W_{\epsilon_j,[\epsilon_L]|(\epsilon_K,\epsilon_L)}]\node_L.
	\end{array}
	\end{equation}
	Note that if $v=(\begin{array}{cccccc}\noise_i^T & \noise_{\bar{I}}^T & \noise_{Z_P}^T & \noise_{Z_I}^T & \noise_{K}^T & \noise_{L}^T \end{array})$,
	then the  power spectral density matrix $\Phi_{vv}(e^{j\omega})$ is assumed positive definite for every $\omega\in R.$) From the uniqueness of the multivariate Wiener filter which follows from Proposition~\ref{prop: my wiener}, we get 
	\[\begin{array}{lll}
	W_{\node_j,[\node_K]|(\node_i,\node_Z,\node_{\bar{I}})}&=&W_{\epsilon_j,[\epsilon_K]|(\epsilon_K,\epsilon_L)}\mbox{ and }\\
	W_{\node_j,[\node_i]|(\node_i,\node_Z,\node_{\bar{I}})}&=&[\identitymatrix-W_{\epsilon_j,[\epsilon_K]|(\epsilon_K,\epsilon_L)}\adjTF_{Kj}^{(f)}]\adjTF_{ji}
	\end{array}\]
	Thus
	\[\begin{array}{lll}
	\adjTF_{ji}&=&[\identitymatrix-W_{\epsilon_j,[\epsilon_K]|(\epsilon_K,\epsilon_L)}\adjTF_{Kj}^{(f)}]^{-1}W_{\node_j,[\node_i]|(\node_i,\node_Z,\node_{\bar{I}})}\\
	&=&[\identitymatrix-W_{\node_j,[\node_K]|(\node_i,\node_Z,\node_{\bar{I}})}\adjTF_{Kj}^{(f)}]^{-1}W_{\node_j,[\node_i]|(\node_i,\node_Z,\node_{\bar{I}})}\\
	&=&[\identitymatrix-W_{\node_j,[\node_K]|(\node_i,\node_Z)}\adjTF_{Kj}^{(f)}]^{-1}W_{\node_j,[\node_i]|(\node_i,\node_Z)}
	\end{array}\]
	where we have used $W_{\node_j|(\node_i,\node_Z,\node_{\bar{I}})}=W_{\node_j|(\node_i,\node_Z)}$.
	Summarizing we have established that, 
	\[\adjTF_{ji}=[\identitymatrix-W_{\node_j,[\node_K]|(\node_i,\node_Z)}\adjTF_{Kj}^{(f)}]^{-1}W_{\node_j,[\node_i]|(\node_i,\node_Z)},\]
	where
	\[\adjTF_{Kj}^{(f)}=[\identitymatrix-(\adjTF_{KK}+\adjTF_{K\bar{J}}\adjTF_{\bar{J}K}^{(e)})]^{-1}[\adjTF_{Kj}+\adjTF_{K\bar{J}}\adjTF^{(e)}_{\bar{J}j}],\]
	with
	\[
	\begin{array}{lll}
	\adjTF_{\bar{J}K}^{(e)}&=&(\identitymatrix -\adjTF_{\bar{J}\bar{J}})^{-1}\adjTF_{\bar{J}K}	\mbox{ and }\\
	\adjTF_{\bar{J}j}^{(e)}&=&(\identitymatrix -\adjTF_{\bar{J}\bar{J}})^{-1}\adjTF_{\bar{J}j}.\\
	\end{array}
	\] 
	This completes the proof. The case for $y_{K}=\emptyset$ is analogous.

\end{document}